\documentclass[a4paper,11pt,showpacs,amsmath,amssymb,floatfix]{article}
\usepackage[utf8]{inputenc}
\usepackage{amsmath}
\usepackage{mathrsfs}
\usepackage{amsfonts}
 \usepackage{cancel}
 
 \usepackage{psfrag} 
\usepackage{graphicx}
\usepackage{auto-pst-pdf}  
 
\usepackage{xcolor}
\usepackage{amsthm}
\usepackage{enumerate}
\usepackage{amssymb}
\usepackage{multirow} 
\usepackage{cite}
\usepackage{bm}
\usepackage{makecell}
\usepackage{caption}

\usepackage{hyperref}

\usepackage[colorinlistoftodos]{todonotes}

\usepackage{xcolor}
 	\definecolor{carblue}{rgb}{0.2, 0.30, 5}

\usepackage{mathtools}
\mathtoolsset{showonlyrefs=true}

\newcommand{\dif}{\mathrm{d}}
\newcommand{\lr}[1]{\left( #1 \right)}
\newcommand{\lrbrace}[1]{\left\lbrace #1 \right\rbrace}

\newcommand{\lrbrkt}[1]{\left[ #1 \right]}

\newcommand{\restr}[2]{\left.{#1}\right\rvert_{#2}}

\newcommand{\skwend}[1]{\mathrm{SkewEnd}(#1)} 
 
\newcommand{\m}{{p_a}}

\newcommand{\mink}[1]{{\mathbb{M}^{#1}}} 
\newcommand{\til}{\widetilde} 

\newcommand{\tilg}{\widetilde{g}}
\newcommand{\Om}{\Omega}
\newcommand{\obs}{\mathcal{O}}
\newcommand{\nor}{u}
\newcommand{\tlt}{\mathring{t}}
\newcommand{\scri}{\mathscr{I}}
\newcommand{\con}{\kappa}

\def\Y{\xi} 
\def\Yv{\xi}

\def\Yf{{\bm \Yv}}

\def\a{\mathfrak{a}}
\def\b{\mathfrak{b}}

\def\Am{\mathcal{F}}
\def\Q{S}

\def\tcar{\tau}

\def\cc{\sigma}

\def\H{\mathcal{H}}
\def\y{y}
\def\h{h}
\def\L{L}
\def\w{w}

\def\mlim{\beta}

\def\ar{\mathrm{a}}
\def\br{\mathrm{b}}
\def\m{\alpha}
\def\param{\zeta}  

\def\ckill{\mathrm{CKill}}

\def\nabscr{\nabla^{(\gamma)}}

\def\conf{\mathrm{Conf}}
\def\confloc{\mathrm{ConfLoc}}

\def\limY{\Y'}

\def\man{\mathcal{M}}

\def\barg{\overline{g}}
\def\hatg{\widehat{g}}

 \def\neigh{\mathcal{U}}
 \def\neighp{\mathcal{V}}
 
    \def\ia{a}

   \def\pc{\chi}
   \def\psf{\psi}

\newcommand{\spn}[1]{\mathrm{span}\{ #1 \}}

\newtheorem{theorem}{Theorem}[section]
\newtheorem{proposition}{Proposition}[section]
\newtheorem{corollary}{Corollary}[theorem]
\newtheorem{lemma}{Lemma}[section]
\newtheorem{remark}{Remark}[section]
\newtheorem{definition}{Definition}[section]

\def\0{{\bf 0}}

\title{ Classification of Kerr-de Sitter-like spacetimes with conformally flat $\scri$ in all dimensions}

\author{
  Marc Mars and Carlos Pe\'on-Nieto\\
  Instituto de F\'{\i}sica  Fundamental y Matem\'aticas, Universidad de Salamanca \\
Plaza de la Merced s/n 37008, Salamanca, Spain}

\begin{document}

\maketitle 

% \tableofcontents 

\begin{abstract}
Using asymptotic characterization results of spacetimes at conformal infinity, we prove that Kerr-Schild-de Sitter spacetimes are in one-to-one correspondence with spacetimes in the Kerr-de Sitter-like class with conformally flat $\scri$. Kerr-Schild-de Sitter are spacetimes of Kerr-Schild form with de Sitter background that solve the $(\Lambda>0)$-vacuum Einstein equations and admit a smooth conformal compactification sharing $\scri$ with the background metric. Kerr-de Sitter-like metrics with conformally flat $\scri$ are a generalization of the Kerr-de Sitter metrics, defined originally in four spacetime dimensions \cite{Kdslike} and extended here to all dimensions in terms of their initial data at null infinity. We explicitly construct all metrics in this class as limits or analytic extensions of Kerr-de Sitter. The structure of limits is inferred from corresponding limits of the asymptotic data, which appear to be hard to guess from the spacetime metrics.
\end{abstract}

\section{Introduction}

    The study of asymptotic properties of gravitational fields in general relativity by means of conformal geometry dates back to the early works by Penrose on the gravitational radiation \cite{Penrose63}, \cite{Penrose64}, \cite{Penrose65}. These works give a precise formulation of asymptoticity in general relativity by introducing the concept of {\it conformal infinity} (also {\it null infinity}), which is a hypersurface $\scri$ given by the zero level set of a sufficiently differentiable function $\Om$, that determines a conformal scaling $g = \Om^2 \tilg$ of the physical metric $\tilg$ in such a way that $g$ regularly  extends to $\{\Om = 0 \}$. Aiming to resolve certain controversy raised by some of the asumptions in Penrose's work, Friedrich formulates his celebrated conformal field equations \cite{friedrich81bis,friedrich81} in four spacetime dimensions. These equations require a careful choice of variables and have the interesting property that they remain regular at $\scri$. 
%     The latter property allows one to have control on the asymptotic behaviour of the gravitational field and have numerous relevant applications in numerical relativity (e.g. \cite{Frauend2000} and references therein), gravitational radiation \cite{Franseno21}\cite{} and, of special interest here, the asymptotic Cauchy problem in general relativity. 
    Remarkably, the asymptotic initial value problem with positive cosmological constant of the Friedrich conformal field equations is well-posed \cite{Fried86geodcomp}, which gives a method for characterization of spacetimes in terms of asymptotic initial data. In dimensions higher than four, one can also pose an asymptotic initial value problem of the $(\Lambda>0)$-vacuum Einstein equations in all even spacetime dimensions \cite{Anderson2005} and also in odd spacetime dimensions provided that the initial data are analytic \cite{kichenassamy03}. The formalism which allows for this is due to Fefferman and Graham \cite{FeffGrah85,ambientmetric}. The asymptotic Cauchy problem in the Fefferman-Graham picture is formulated in terms of asymptotic formal series expansions (i.e. ``near'' $\scri$) of conformally Einstein metrics in a particular conformal gauge. We will give more details of this in Section \ref{secFG}.

    The geometric characterizations of black hole spacetimes are of special interest in general relativity because of their relation with the famous uniqueness theorems of stationary black holes. More specifically, the no-hair conjecture asserts, roughly speaking, that every stationary electrovacuum black hole solution is entirely characterized by its (suitably geometrically defined) mass, angular momentum and electric charge. This conjecture has been extensively studied in the zero cosmological constant setting (see e.g. \cite{chrusciel12,mazur} and references therein). Under a few technical conditions, whose complete removal is still an open problem, the no-hair theorem singles our the Kerr spacetime as the unique asymptotically flat vacuum black hole spacetime. 
    
%      to be satisfied by static (i.e. Schwarszchild) and stationary axisymmetric configurations\footnote{The proof of the stationary axisymmetric case is not considered fully general as it assumes non-degenerate analytic horizons. Giving a general proof of this is still today a difficult open problem.} (i.e. Kerr). The latter cases are of particular relevance because of the role that they are believed to play as the endpoint states of collapsing self-gravitating systems.

A local characterization of the Kerr and Kerr-NUT metrics among spacetimes with one Killing vector field can be given in terms of the vanishing of the so-called Mars-Simon tensor \cite{mars99,simon84}. 
%  Remarkably, this characterization has led to proposals for exteding Kerr's rigidity in the non-analytic setting \cite{alexakis10,ionescu10}.
In the non-zero cosmological constant case, the vanishing of the Mars-Simon tensor also characterizes \cite{marsseno15} the Kerr-NUT-(A)de Sitter metrics and related spacetimes. Recall that the latter generalize Kerr-NUT to the arbitrary cosmological constant setting, so they are also important from a physical perspective.  Particularly, in the case of positive cosmological constant, the geometric characterizations of Kerr-de Sitter are interesting in view of a potential extension of a uniqueness theorem of Kerr-de Sitter among $(\Lambda>0)$-vacuum stationary black hole spacetimes, or perhaps among stationary spacetimes admiting a regular cosmological horizon.  
       
 The results in \cite{marsseno15} are used in \cite{Kdslike,KdSnullinfty} to provide a characterization in terms of asymptotic initial data of Kerr-NUT-de Sitter metrics and related spacetimes, which altogether define the so-called {\it Kerr-de Sitter-like class} (see also \cite{gasperinkroon17,gaspthesis} for a similar characterization of Kerr-de Sitter and Schwarszchild-de Sitter with spinorial techniques). The asymptotic data of this class happen to be determined by the (conformal) geometry of $\scri$ as well as a conformal class of conformal Killing vector fields (CKVFs) of $\scri$. The latter is an equivalence class generated by a CKVF $\Y$ of $\scri$ up to (local) conformal diffeomorphisms of $\scri$ (cf. Appendix \ref{seclocconflat}). Interestingly, $\Y$ singles out a particular symmetry of the spacetime that  could be used to define asymptotic charges, which might lead to definitions of mass and angular momenta in the positive cosmological setting (see the arguments in \cite{marspeon21}).
 
 The Mars-Simon tensor is only defined in four spacetime dimensions, so it cannot be used for characterizations in higher dimensions. However, the asymptotic data of Kerr-de Sitter \cite{Kdslike,KdSnullinfty} have been proven to naturally extend to higher dimensions in \cite{marspeondata21}. The arbitrary dimensional Kerr-de Sitter metrics are due to Gibbons et al. \cite{Gibbons2005} and it is remarkable that their characterization in terms of asymptotic data in \cite{marspeondata21} gives a clear and strong argument why these metrics are indeed the natural extension of the Kerr-de Sitter metric to higher dimension,  beyond the original heuristic construction. Moreover, the asymptotic data of arbitrary dimensional Kerr-de Sitter happens to admit a further generalization. This is based on the fact that, just like in the four dimensional case, Kerr-de Sitter is asymptotically characterized by a conformally flat $\scri$ and a conformal class of CKVFs $[\Y]$, exclusively determined by the mass and rotation parameters of Kerr-de Sitter. By allowing $[\Y]$ to be an arbitrary conformal class, one generates a class of data which naturally extends to higher dimensions the Kerr-de Sitte-like class with conformally flat
 $\scri$.
 
 The starting point of this paper is the extension of the Kerr-de Sitter-like class with conformally flat $\scri$ outlined above, which has already been anticipated in \cite{marspeondata21}. The definition in terms of asymptotic data, however, lacks of an explicit spacetime metric. In this paper we shall construct all the metrics in this class by characterizing them as the so-called  {\it Kerr-Schild-de Sitter} metrics (cf. Definition \ref{defKSdS}), which are $(\Lambda >0)$-vacuum Kerr-Schild metrics which admit a locally conformally flat $\scri$ and satisfy a natural asymptotic decay condition. First, we prove that the Kerr-Schild-de Sitter metrics are contained in the Kerr-de Sitter-like class (with conformally flat $\scri$). For this part, we make essential use of the result in \cite{malek11} that all $\Lambda$-vacuum Kerr-Schild spacetimes are algebraically special is essential. We shall also prove the converse inclusion by taking advantage of the well-posedness of the asymptotic Cauchy problem. Namely, since the initial data of the Kerr-de Sitter-like class (with conformally flat $\scri$) are determined by a class of CKVFs $[\Y]$, we use the limits of classes of CKVFs studied in \cite{marspeon21} to determine limits of asymptotic data corresponding to Kerr-de Sitter, whose spacetime metric is known a priori. From the limits of data, we are able to infere spacetime limits, thus obtaining explicit metrics non-isometric to any member of the original Kerr-de Sitter family.
 It is worth to stress that the structure of limits in the even and odd spacetime dimensions is remarkably different. In the latter case, the data corresponding to Kerr-de Sitter is dense in the quotient topology (of the Lie algebra of CKVFs modulo the Lie group of conformal transformations), which allows to exhaust the Kerr-de Sitter-like class in terms of limits. When the spacetime dimension is even, there exists one family of spacetimes, different from Kerr-de Sitter, whose asymptotic data span an open domain of the quotient topology. This is constructed by an analytic extension of Kerr-de Sitter, which together with the limit spacetimes, exhaust the Kerr-de Sitter-like class.

 This paper is organized as follows. Section \ref{secpre} expands on the formal aspects of the problem presented above and gives a precise version of our results, including the definition of Kerr-Schild-de Sitter spacetimes, of Kerr-de Sitter-like class with conformally flat $\scri$ in all dimensions, our equivalence theorem between the two classes (Theorem \ref{theoKSKdS}) and the explicit form of the metrics (Theorem \ref{theoKSKdSmetrics}). Then, the basics of the Fefferman-Graham formalism and a list of some useful results in \cite{marspeondata21} are given in Section \ref{secFG}. The rest of the paper is devoted to the proof of both Theorems.  Section \ref{secKSsubKdS} proves the inclusion of Kerr-Schild-de Sitter metrics into the Kerr-de Sitter-like class. The proof of the converse inclusion along with the explicit construction of the metrics is provided in Section \ref{secKSsupKdS}.
 In Appendix \ref{seclocconflat} we address the problem of how to define conformal classes of CKVFs. The difficulty stems from the fact that we work locally on the manifold, so we need a suitable version of conformal equivalence defined in terms of local conformal diffeomorphisms. This difficulty is solved by showing that, in the locally conformally flat case, the local equivalence of CKVFs have a global correspondence in the $n$-sphere.

\section{General setting and main result}\label{secpre}
       
In this section we describe the general setting of the paper and state our main results (cf. Theorems \ref{theoKSKdS} and \ref{theoKSKdSmetrics}). Our analysis is based on the study of initial data at conformal infinity, so we begin by stating some fundamental facts on conformal geometry and conformal infinity.    
  
The spacetimes in this paper are $(n+1)$-dimensional with $n \geq 3$ and of Lorentzian signature $\{-1,+1,\cdots,+1 \}$. The metrics with tilde $\tilg$ denote smooth solutions of the $\Lambda$-positive Einstein equations 
\begin{equation}\label{EFE}
\widetilde R_{\alpha \beta}  = n \lambda \tilg_{\alpha \beta}, \quad\quad \lambda:= \frac{2 \Lambda}{n(n-1)}\quad(>0).
\end{equation}  
A spacetime manifold $(\man,g)$ with boundary $\partial \man$ is said to be a {\it conformal extension } of $(\widetilde \man, \tilg)$ if $\widetilde \man =  \mathrm{Int}(\man)$ and in addition there exists a function $\Om \in C^\infty (\man)$ positive on $\widetilde \man$ such that 
\begin{equation}
g = \Om^2 \tilg \quad\quad \mbox{and}\quad\quad \partial \man = \{\Om = 0 \cap \dif \Om \neq 0 \}.
\end{equation} 
In this section, $g$ need not to be smooth up to the boundary, it has however finite differentiability at $\partial \man$ (cf. equation \eqref{eqFGexpevsec}). Later on, we will assume that the induced geometry at $\partial \man$ is conformally flat, which, as we will discuss in section \ref{secFG}, implies that $g$ is smooth up to the boundary. 
The submanifold $\scri := (\partial \man, g\mid_{\partial \man})$ is called "conformal infinity" and it is well-known to be spacelike when $\tilg$ satisfies \eqref{EFE}. Since $\Omega$ can be multiplied by any smooth positive function of $\partial \man$ without affecting the definition, any smooth positive function $\omega$ of $\partial \man$ yields a different conformal extension such that $\scri = (\partial \man, \omega^2 g \mid_{\partial \man})$. Hence one usually considers $\scri$ as $\partial \man$ equipped with the whole conformal class of metrics $[g\mid_{\partial \man}] = \omega^2 g\mid_{\partial \man},~\forall \omega \in C^\infty(\partial \man),~\omega>0$. A metric admitting a conformal extension is said to be conformally extendable.

As mentioned in the introduction, the Friedrich conformal field equations \cite{Fried86geodcomp} allow to define a well-posed asymptotic Cauchy problem  for metrics $\tilg$ satisfying \eqref{EFE}. 
The initial data are a Riemannian $3$-manifold $(\Sigma,\gamma)$, which prescribes the (conformal) geometry of $\scri$, and a traceless and transverse (i.e. with zero divergence) symmetric two-tensor $D$, or TT tensor for short, which prescribes the electric part of the rescaled Weyl tensor at $\scri$. Recall that the electric part of the Weyl tensor and its rescaled version (for arbitrary $n$) are
\begin{equation} \label{eqelecresc}
 (C_\perp)_{\alpha \beta} := {C^\mu}_{\alpha \nu \beta} \nor_\mu \nor^\nu,\quad\quad (c_\perp)_{\alpha \beta} := \Om^{2-n} {C^\mu}_{\alpha \nu \beta} \nor_\mu \nor^\nu,
\end{equation}
where $\nor$ is a timelike unit vector defined in a open neighbourhood of $\partial \man$ and normal to $\scri$. Generically, $c_\perp$  diverges at $\scri$ in dimension $n  \geq 4$. However, as we shall discuss later, it is finite under suitable restrictions on the conformal class $[\gamma]$, in particular, whenever it is locally conformally flat. For such cases, we define its value at $\scri$
\begin{equation} \label{eqelecrescscri}
  D_{\alpha\beta}:= \restr{(c_\perp)_{\alpha\beta}}{\scri}.
\end{equation}

In higher dimensions, the asymptotic Cauchy problem of the Einstein equations is based on a different approach, namely, the Fefferman-Graham formalism \cite{FeffGrah85,ambientmetric} (see also Section \ref{secFG}). In the $n$ odd case, a theorem due to Anderson \cite{Anderson2005} associates a unique Einstein metric to every asymptotic data set $(\Sigma,\gamma, \widehat D)$, with $(\Sigma,\gamma)$ a Riemannian $n$-manifold $\gamma$ and $\widehat D$ a symmetric 2-tensor TT w.r.t. $\gamma$. Although the core idea
appears for the first time in \cite{Anderson2005}, neither this paper nor \cite{andersonchrusciel05}, which attempts to give
a detailed proof, are fully correct. The mistakes in those papers have recently been
identified in \cite{kaminski21}, where a complete proof of the existence results has been provided. The $n$ even case is more limited. The existing result by Kichenassamy \cite{kichenassamy03} associates a unique Einstein metric to every asymptotic data set $(\Sigma, \gamma, \widehat D)$, provided $(\Sigma,\gamma)$ is a Riemannian \underline{analytic} $n$-manifold and $\widehat D$ a symmetric \underline{analytic} 2-tensor, whose trace and divergence (generically non-zero) are determined exclusively by $\gamma$. The tensor $\widehat D$ does not in general prescribe the electric part of the rescaled Weyl tensor at $\scri$. However it does so whenever $\gamma$ is locally conformally flat (cf. \cite{marspeondata21,Holl05}). In such a case $\widehat D$ is obviously TT. For simplicity, we shall use ``conformally flat'' to mean ``locally conformally flat'' from now on. 

In the next theorem we summarize the existence and uniqueness theorems discussed so far. All the above asymptotic data sets are denoted generically $(\Sigma,\gamma,\widehat D)$ and we shall specify when it is known that $\widehat D$ prescribes $D$. In those cases we shall use $(\Sigma,\gamma,D)$ for the asymptotic data.

\begin{theorem}\label{theoexistencedata}
 Let $(\Sigma,\gamma,\widehat D)$ be an $n$-dimensional asymptotic data set. Then, if $n$ is odd \cite{Anderson2005,andersonchrusciel05,kaminski21}
 or if $\gamma,\widehat D$ are analytic \cite{kichenassamy03}, there exists a unique metric $\tilg$, which solves \eqref{EFE} and admits a conformal extension such that $\scri$ which can be identified with $(\Sigma,\gamma)$. Moreover, in $n=3$ \cite{Fried86geodcomp} or if $n>3$ and $(\Sigma,\gamma)$ is conformally flat \cite{marspeondata21}, the tensor $\widehat D$ is TT and prescribes $D$, the electric part of the rescaled Weyl tensor of $\tilg$ evaluated at $\scri$.
\end{theorem}

\begin{remark}\label{remarkequivdata}
It is clear that the asymptotic data must posses a large conformal gauge freedom, arising from the fact that the conformal factor $\Om$ which extends the metric to $\scri$ has no physical relevance. Whenever the data are $(\Sigma,\gamma,D)$, $D$ prescribing the electric part of the rescaled Weyl tensor at $\scri$, any other set of equivalent data  $(\Sigma,\gamma',D')$ must be (see Section \ref{secFG} for an explicit verification)
\begin{equation}\label{eqequivdata}
 \gamma' = \omega^2 \gamma 
 ,\quad\quad
 D' = \omega^{2-n} D 
 ,\quad\quad \forall (0<)\omega \in C^\infty(\Sigma). 
\end{equation}
If $\omega$ is a locally defined function, the equivalence holds restricting $\Sigma$ to the domain of definition of $\omega$. In all odd $n$ cases, where the data are $(\Sigma,\gamma,\widehat D)$ and $\widehat D$ is TT (generically different from $D$), one expects the same conformal transformation law for $\widehat D$ than in \eqref{eqequivdata}. For $n$ even, the transformation law is more involved because $\widehat D$ has non-zero trace and divergence. We can bypass this complication in the conformally flat $\scri$ case, by decomposing $\widehat D = \bar D + D$, where $\bar D$ is totally determined by $\gamma$ if $n=4$ or zero otherwise, and $D$ prescribes the electric part of the rescaled Weyl tensor

\end{remark}

 Theorem \ref{theoexistencedata} is used in \cite{marspeondata21} to give a geometric characterization of the Kerr-de Sitter metrics in all dimensions \cite{Gibbons2005} by means of their initial data, which consists of a conformally flat Riemannian metric $\gamma$ together with a TT tensor of the form 
\begin{equation}\label{eqTTDY}
  D =\kappa D_\Y, \quad \quad D_\Y := \frac{\kappa}{|\Y|^{n+2}_\gamma} \lr{\Yf \otimes \Yf - \frac{|\Y|_\gamma^2}{n}\gamma},\quad\quad \kappa \in \mathbb{R},
 \end{equation}
where $\Yf := \gamma (\Y,\cdot)$ for $\Y$ a specific CKVF\footnote{When using index-free notation, it will be convenient to distinguish the CKVFs of $\scri$ from their metricallly associated one-forms. We keep the notation employed here, namely, the one-forms are specified with boldface font. } of $\gamma$, that depends on the physical (i.e. mass and rotation) parameters of the Kerr-de Sitter metric. Nevertheless, the tensor $D_{\til\Y}$ is well defined (away from fixed points of $\til\Y$) for any
conformal Killing vector field $\Y$ on any $n$-dimensional manifold $(\Sigma,\gamma)$. After removing the fixed points from $\Sigma$, one easily checks that $D_{\til\Y}$ is a TT tensor (see the proof for $n=3$ in \cite{KdSnullinfty}, which readily extends to arbitrary $n$). When $\gamma$ is conformally flat, then both $\gamma$ and $D_{\til\Y}$ are analytic on $\Sigma$. Consequently, by Theorem 1.1 there exists a spacetime corresponding to the data $(\Sigma, \gamma, \kappa D_{\til\Y})$ for any constant $\kappa \in \mathbb{R}$ and regardless of the parity of $n$.
 The spacetime corresponding to data $(\Sigma, \gamma,\kappa D_{\til\Y})$ is not in general Kerr-de Sitter if $\til\Y \neq \Y$. In addition, combining the conformal equivalence of data in Remark \ref{remarkequivdata} and the diffeomorphism equivalence of data, it is proven in  \cite{marspeondata21} that two asymptotic data sets $(\Sigma, \gamma,\kappa D_{\Y})$ and 
$(\Sigma, \gamma,\kappa D_{\til\Y})$ are equivalent if and only if $\Y$ and $\til\Y$ are in the same conformal class (see Appendix \ref{seclocconflat} for a precise definition). We observe that since the conformal transformations $\phi \in \confloc(\Sigma,\gamma)$ are in general locally defined in open neighbourhoods $\neigh \subset \Sigma$, the equivalence of data between $(\Sigma, \gamma,\kappa D_{\Y})$ and 
$(\Sigma, \gamma,\kappa D_{\til\Y})$ with $\til\Y = \phi_\star(\Y)$ only makes sense restricted to non-empty sets of the form $\phi(\neigh) \cap \neigh$. We shall implicitly assume this restriction when dealing with local equivalence of data.

\begin{lemma}[\hspace{-0.025cm}\cite{marspeondata21}]\label{lemmpropKdSl} 
 Let $(\Sigma, \gamma, \kappa  D_\Y)$ be initial data with $\gamma$ conformally flat and $D_\Y$ of the form \eqref{eqTTDY}, with $\Y$ any CKVF of $\gamma$. Then for each CKVF $\til\Y$ in the conformal class of $\Y$ the set of data $(\Sigma, \gamma, \kappa  D_\Y)$ and $(\Sigma, \gamma, \kappa  D_{\til\Y})$ are equivalent.
\end{lemma}
 
Now consider the class of spacetimes corresponding to data of the form $(\Sigma,\gamma,\kappa D_\Y)$, with $\gamma$ conformally flat and $\Y$ an arbitrary CKVF of $\gamma$. By Lemma \ref{lemmpropKdSl}, there exist as many non-equivalent such spacetimes as conformal classes of CKVFs. Thus, the following definition is natural and extends to arbitrary dimension the corresponding notion in dimension $n=3$ introduced\footnote{The non-conformally flat $n=3$ case is defined in \cite{KdSnullinfty}.} in \cite{Kdslike}.

% Actually, this class of spacetimes was first described for the $n=3$ case in \cite{Kdslike} and named {\it Kerr-de Sitter-like class} of spacetimes with conformally flat $\scri$. In \cite{Kdslike}, the class is defined as the set of spacetimes solving the vacuum Einstein equations with positive cosmological constant, admitting a smooth conformally flat\footnote{The non-conformally flat $n=3$ case is similarly defined in \cite{KdSnullinfty}} $\scri$ as well as a Killing vector field $\zeta$, whose associated Mars-Simon tensor vanishes. This definition implies initial data at $\scri$ of the form $(\Sigma, \gamma,\kappa  D_\Y)$, whith $\Y$ is the restriction to $\scri$ of the Killing vector field $\zeta$. However, recall that no analogous to the Mars-Simon tensor is known in higher dimension, hence the extension of the definition of the Kerr-de Sitter-like class requires a different approach. By the above discussion, an obvious possibility is to give the definition directly in terms of its initial data $(\Sigma, \gamma,\kappa  D_\Y)$.

\begin{definition}\label{defKdSlike}
 The {\bf Kerr-de Sitter-like} class of spacetimes with conformally flat $\scri$  are conformally extendable metrics solving \eqref{EFE}, characterized by data $(\Sigma, [\gamma],[\kappa D_\Y])$, with $\gamma$ conformally flat and where $D_\Y$ is given by \eqref{eqTTDY} with $\Y$ a CKVF of $\gamma$. 
\end{definition}
In order to clarify the terminology, the word \underline{class} is used to denote a collection of \underline{families} of spacetimes, a family being a set of metrics, depending on a number of parameters and sharing certain properties. For example, the Kerr-de Sitter-like class with conformally flat $\scri$ and $n=3$ contains \cite{Kdslike}: the Kerr-de Sitter family, the Kottler families, a limit case of Kerr-de Sitter with infinite rotation parameter \cite{limitkds} and the Wick-rotated-Kerr-de Sitter spacetime \cite{wickrot54}. In this paper we extend the definition of these families to higher dimensions.

The main purpose of this work, is to prove that the Kerr-de Sitter-like class with conformally flat $\scri$ contains exactly all Kerr-Schild type spacetimes, solution of the $\Lambda >0$ vacuum Einstein equations and sharing $\scri$ with its background metric. In particular this requires that, as the background metric is de Sitter, $\scri$ is conformally flat. Recall that the Kerr-Schild spacetimes are of the form
\begin{equation}\label{eqKSmetrics}
 \tilg = \tilg_{dS} + \widetilde{\H} ~ \widetilde k \otimes \widetilde k
\end{equation}
where $\tilg_{dS}$ is de Sitter, $k$ is a field of lightlike one-forms (both w.r.t. $\tilg_{dS}$ and $\tilg$) and $\widetilde{\H}$ is a smooth function. It is convenient to give a name to the set of spacetimes we shall be dealing with.
\begin{definition}\label{defKSdS}
 The {\bf Kerr-Schild-de Sitter} spacetimes are of the form \eqref{eqKSmetrics}, solve the $\Lambda>0$ vacuum Einstein equations and admit a smooth conformally flat $\scri$ such that for some conformal extension $g  = \Om^2\tilg$, the tensor $\Om^2\til\H \til k \otimes \til k$ vanishes at $\scri$. 
\end{definition}
\begin{remark}\label{remarkdecayH1}
    Notice that asking the metric $\tilg$ to share $\scri$ with $\tilg_{dS}$, implies more than simply $\tilg$ to admit a conformally flat $\scri$. In particular, consider a conformal extension such that $\gamma = \Omega^2 \tilg|_\scri$ is conformally flat and assume that $\gamma_{dS} := \Omega^2 \tilg_{dS}|_\scri $ and $(\Om^2 \til \H \til k \otimes \til k)|_{\scri}$ are well-defined. Since $\gamma_{dS}$ is conformally flat, one could naively think that $\gamma = \gamma_{dS} + (\Om^2 \til \H \til k \otimes \til k)|_{\scri}$ implies $(\Om^2 \til \H \til k \otimes \til k)|_{\scri} = 0$, which would then imply the condition on $\Om^2 \til \H \til k \otimes \til k$ assumed in  Definition $\ref{defKSdS}$. However, there is still room, in principle, for conformally flat metrics of the form $\gamma_{dS} + \H_0 y \otimes y$ with $\H_0 \neq 0,~y\neq 0$. A simple example is any conformally flat graph in a flat $n$-dimensional space endowed with Cartesian coordinates $\{ x^i \}$, i.e.  a hypersurface defined by $x^n = f(x^i)$, such that the induced metric happens to be conformally flat. The induced metric takes precisely the form $\gamma_S = \gamma_{\mathbb{E}^{n-1}} + y \otimes y $, for a flat $(n-1)$-dimensional metric $\gamma_{\mathbb{E}^{n-1}}$ and $y := \dif f$ (as an explicit example one can take a hemisphere). 
    
    Thus, it may be possible that a Kerr-Schild metric, solving \eqref{EFE} and admitting a smooth conformally flat $\scri$ has a term $\Om^2 \til \H \til k \otimes \til k$ surviving at $\scri$. It would be interesting to settle whether any $\Lambda>0$-vacuum solution of this type can exist.

\end{remark}

With the above definitions \ref{defKdSlike} and \ref{defKSdS} we can now state the main result of this paper:

\begin{theorem}\label{theoKSKdS}
 A spacetime belongs to the Kerr-de Sitter-like class if and only if it is Kerr-Schild-de Sitter.
\end{theorem}

 The proof of Theorem \ref{theoKSKdS} involves two steps, which respectively we address in Sections \ref{secKSsubKdS} and \ref{secKSsupKdS} of this paper. In Section \ref{secKSsubKdS} we consider Kerr-Schild-de Sitter metrics and compute their initial data, which by Theorem \ref{theoexistencedata}, correspond to the conformal geometry of (conformally flat) $\scri$ and the electric part of the rescaled Weyl tensor $D$. The tensor $D$ is easily seen to have the form $D = \kappa D_\Y$, with $\kappa \in \mathbb{R}$ and $D_\Y$ given by \eqref{eqTTDY} with $\Y$ the projection of $k$ onto $\scri$. The main task of this section is to prove that $\Y$ is a CKVF of $\scri$. This is a consequence of the Kerr-Schild-de Sitter spacetimes being algebraically special (cf. Proposition \ref{propksspec}). This proves that every Kerr-Schild-de Sitter spacetime is contained in the Kerr-de Sitter-like class. 
 
 The reverse inclusion is proven in Section \ref{secKSsupKdS}. To do that we generate every spacetime in the Kerr-de Sitter-like class by taking advantage of the topological structure of the space of conformal classes of CKVFs. By Lemma  \ref{lemmpropKdSl}, one conformal class corresponds exactly to one spacetime in the class. Moreover, from the well-posedness of the Cauchy problem, limiting classes in the quotient space of CKVFs will generate limiting spacetimes. All the metrics one obtains are summarized in the next theorem. In order to simplify the statement, we modify slightly the notation with respect to Section \ref{secKSsupKdS}: all primes and hats are dropped and all rotation parameters are denoted by $a_i$.
 
 \begin{theorem}\label{theoKSKdSmetrics}
  Let be $(\mathcal \man,\tilg)$ be an $(n+1)$-dimensional manifold and set $ p := \lrbrkt{\frac{n+1}{2}} -1,$ and $q := \lrbrkt{\frac{n}{2}} $. Consider the functions $W$ and $\Xi$ of table \ref{tab2} and $\m_{p+1}$ obtained from the implicit equation in table \ref{tab2}, for a collection of real parameters  $\{a_i \}_{i=1}^{p+1}$ with $a_{p+1}=0$ if $n$ odd or in case $b)$. Then, in the coordinates 
  $\{\rho, t, \{\m_i\}_{i=1}^{p+1}, \{\phi_i\}_{i=1}^q\}$ taking values in $\phi_i \in [0,2 \pi)$ and the maximal domain where $W$ and $\Xi$ are positive and $\m_{p+1}$ is real, every Kerr-Schild-de Sitter metric 
 \begin{equation}
  \tilg = \tilg_{dS} + \til{\mathcal{H}} \til k \otimes \til k,\quad\quad\mbox{must have}\quad\til \H = \frac{2M \rho^{n-2}}{ \Xi \prod_{i=1}^q(1 + \rho^2 a_i^2) },\quad M \in \mathbb{R},
\end{equation}
 $k$ as given in table \ref{tab2} and the de Sitter metric $\tilg_{dS}$ in the corresponding following form:
\begin{enumerate}
 \item[a)] Kerr-de Sitter family,
\begin{align}
  \tilg_{dS} & =  - W \frac{(\rho^2 - \lambda )}{\rho^2} \dif t^2 + \frac{\Xi}{\rho^2 - \lambda } \frac{\dif \rho^2}{\rho^2} + \delta_{p,q} \frac{\dif \m_{p+1}^2}{\rho^2} \\ &  + \sum_{i=1}^q \frac{1 + \rho^2 a_i^2}{\rho^2} \lr{\dif \m_i^2 + \m_i^2 \dif \phi_i^2}     + \frac{ (\rho^2- \lambda )}{ \lambda W \rho^2 }  \frac{\dif W^2}{4}.
  \end{align}
 \item[b)] $\{a_i\rightarrow \infty \}$-limit-Kerr-de Sitter,
\begin{align}
  \hspace{-1.5cm}\tilg_{dS} & =    \frac{ \lambda \m_{p+1}^2 }{\rho^2} \dif t^2 - \frac{\Xi}{ \lambda } \frac{\dif \rho^2}{\rho^2} + \delta_{p+1,q}\frac{ \m_{p+1}^2  \dif \phi_{q}^2}{\rho^2 }  + \sum_{i=1}^p \frac{1 + \rho^2 a_i^2}{ \rho^2}{ \lr{\dif \m^2_i + \m^2_i \dif \phi_i^2} } \\ & 
  + \lr{\frac{1}{\lambda} + \frac{\sum_{i=1}^p \beta_i^2}{\rho^2 \hat \m_{p+1}^2}}\dif\m_{p+1}^2   -\frac{ 2  \dif\m_{p+1}}{\rho^2\m_{p+1}} \lr{ \sum_{i= 1}^{p}{\m_i \dif  \m_i}}.
  \end{align} 
 \item[c.1)] Wick-rotated-Kerr-de Sitter for $n$ even,
  \begin{align} 
  \tilg_{dS} & =   \frac{\lambda W}{\rho^2} \dif t^2 - \frac{\Xi}{ \lambda } \frac{\dif \rho^2}{\rho^2}  + \sum_{i=1}^q \frac{1 + \rho^2 a_i^2}{\rho^2}{ \lr{\dif  \m^2_i +  \m^2_i \dif \phi_i^2} }
  -  \frac{ 1}{  W \rho^2 }  \frac{\dif W^2}{4}.
\end{align}
 \item[c.2)] Wick-rotated-Kerr-de Sitter for $n$ odd,
 \begin{align}
  \tilg_{dS} & =    W \frac{(\rho^2 + \lambda )}{\rho^2} \dif t^2 - \frac{\Xi}{\rho^2 + \lambda } \frac{\dif \rho^2}{\rho^2} -  \frac{\dif \m^2_{p+1}}{\rho^2} + \sum_{i=1}^p \frac{1 +  \rho^2 a_i^2}{\rho^2} \lr{\dif  \m^2_i + \m^2_i \dif \phi_i^2}.
\end{align}
\end{enumerate}
 \end{theorem}

% \begin{table}[h!]
%  \begin{centering}
% \begin{tabular}{|c|c|c|c|c|}\hline
%  Case  & Constraint on $\{\alpha_i \}$ & $W$ & $\Xi$ &  $\til k$  \\[0.1cm]
%  \hline
%  a) & $\sum_{i=1}^{p+1} (1 + \lambda a_i^2) \alpha_i^2 = 1$ & $\sum_{i=1}^{p+1} \alpha_i^2 $ & $\sum\limits_{i=1}^{p+1} \frac{1+\lambda a_i^2}{1 + \rho^2 a_i^2} \m_i^2$ & 
%  $  W \dif t - \frac{\Xi}{\rho^2 - \lambda } \dif \rho - \sum\limits_{i=1}^q {a_i \m_i^2} \dif \phi_i$ \\
%  \hline
%   b) & $\m_{p+1}^2  + \sum_{i=1}^p \lambda a_i^2  \m^2_i= 1$ & $\m_{p+1}^2 $ & $\m_{p+1}^2 + \sum\limits_{i=1}^{p} \frac{\lambda a_i^2}{1 + \rho^2 a_i^2} \m_i^2$ & 
%  $W \dif t + \frac{\Xi}{ \lambda }  \dif \rho - \sum\limits_{i=1}^p { a_i \m^2_i} \dif \phi_i$ \\
%  \hline
%  c.1) & $\sum_{i=1}^{p+1} \lambda a_i^2 \alpha_i^2 = 1$ & $\sum_{i=1}^{p+1} \alpha_i^2 $ & $\sum\limits_{i=1}^{p+1} \frac{\lambda a_i^2}{1 + \rho^2 a_i^2} \m_i^2$ & 
%  $  \frac{\Xi}{ \lambda } \dif \rho - \sum\limits_{i=1}^q b_i {\m^2_i} \dif \phi_i$ \\
%  \hline
%   c.2) & $\alpha_{p+1}^2-\sum_{i=1}^{p} (1 - \lambda a_i^2) \alpha_i^2 = 1$ & $\m_{p+1}^2 - \sum_{i=1}^{p} \alpha_i^2 $ & $\m_{p+1}^2-\sum\limits_{i=1}^{p} \frac{1-\lambda a_i^2}{1 + \rho^2 a_i^2} \m_i^2$ & 
%  $ W \dif t + \frac{\Xi}{\rho^2 + \lambda } \dif \rho + \sum\limits_{i=1}^q {a_i \m_i^2}\dif \phi_i$ \\
%  \hline
% \end{tabular} 
% \caption{Functions defining Kerr-Schild-de Sitter families.}\label{tab2}  
% \end{centering} 
%   \end{table}  

\begin{center}
\makebox[0pt][c]{%
\begin{minipage}[b]{\paperwidth}
\centering 
\begin{tabular}{|c|c|c|c|c|}\hline
 Case  & Constraint on $\{\alpha_i \}$ & $W$ & $\Xi$ &
                  $\overset{~}{\til k}$                                                      
  \\ \hline
 a) & $\sum_{i=1}^{p+1} (1 + \lambda a_i^2) \alpha_i^2 = 1$ & $\sum_{i=1}^{p+1} \alpha_i^2 $ & $\sum\limits_{i=1}^{p+1} \frac{1+\lambda a_i^2}{1 + \rho^2 a_i^2} \m_i^2$ & 
 $  W \dif t - \frac{\Xi}{\rho^2 - \lambda } \dif \rho - \sum\limits_{i=1}^q {a_i \m_i^2} \dif \phi_i$ \\
 \hline
  b) & $\m_{p+1}^2  + \sum_{i=1}^p \lambda a_i^2  \m^2_i= 1$ & $\m_{p+1}^2 $ & $\m_{p+1}^2 + \sum\limits_{i=1}^{p} \frac{\lambda a_i^2}{1 + \rho^2 a_i^2} \m_i^2$ & 
 $W \dif t + \frac{\Xi}{ \lambda }  \dif \rho - \sum\limits_{i=1}^p { a_i \m^2_i} \dif \phi_i$ \\
 \hline
 c.1) & $\sum_{i=1}^{p+1} \lambda a_i^2 \alpha_i^2 = 1$ & $\sum_{i=1}^{p+1} \alpha_i^2 $ & $\sum\limits_{i=1}^{p+1} \frac{\lambda a_i^2}{1 + \rho^2 a_i^2} \m_i^2$ & 
 $  \frac{\Xi}{ \lambda } \dif \rho - \sum\limits_{i=1}^q b_i {\m^2_i} \dif \phi_i$ \\
 \hline
  c.2) & $\alpha_{p+1}^2-\sum_{i=1}^{p} (1 - \lambda a_i^2) \alpha_i^2 = 1$ & $\m_{p+1}^2 - \sum_{i=1}^{p} \alpha_i^2 $ & $\m_{p+1}^2-\sum\limits_{i=1}^{p} \frac{1-\lambda a_i^2}{1 + \rho^2 a_i^2} \m_i^2$ & 
 $ W \dif t + \frac{\Xi}{\rho^2 + \lambda } \dif \rho + \sum\limits_{i=1}^q {a_i \m_i^2}\dif \phi_i$ \\
 \hline
\end{tabular} \captionof{table}{\it \small Functions defining the Kerr-Schild-de Sitter families.}
\label{tab2}  
\end{minipage}
}
\end{center}

 \section{Fefferman-Graham formalism}\label{secFG}

The formalism which allows one to work with initial data at $\scri$ in spacetime dimensions higher than four is the Fefferman and Graham (FG) expansion for asymptotically Einstein metrics \cite{FeffGrah85, ambientmetric}. These are, roughly speaking, metrics satisfying the Einstein equations with non-zero cosmological constant to a certain order at $\scri$. The details of the Fefferman-Graham framework can be found in \cite{FeffGrah85, ambientmetric}. We only summarize the fundamental results that we shall need. We also restrict ourselves to the case of positive cosmological constant.

The FG expansion uses a particularly useful type of conformal extensions, called {\it geodesic}. Namely, those such that the vector field $\nabla \Om$ is geodesic w.r.t. the metric $g$. The following result concerning geodesic conformal extensions is standard in conformal geometry (e.g.  \cite{GrahamLee91}, \cite{ambientmetric} and particularly \cite{marspeondata21} for this specific formulation and proof).
\begin{lemma}\label{lemmanormgeod}
 Let $\tilg$ be a conformally extendable metric solving \eqref{EFE} and let $(\scri,[\gamma])$ be the corresponding conformal class of conformal infinity. Then, after restricting $\man$ if necessary, there exists a geodesic conformal extension for every boundary metric $\gamma \in [\gamma]$. In addition, a conformal extension $g = \Om^2 \tilg$ is geodesic if and only if the vector field $T_\mu:= \nabla_\mu \Om$ satisfies $ g^{\mu \nu}T_\mu T_\nu = -\lambda$. 
\end{lemma} 
Unless otherwise specified, we will assume that the conformal extensions are geodesic. In such case, it is often also convenient (possibly after restricting $\man$ further) to use geodesic Gaussian coordinates $\{\Om, x^i\}_{i=1}^n$ adapted the foliation of leaves $\Sigma_\Om = \{\Om = const.\}$.

Given a geodesic conformal extension inducing a boundary metric $\gamma$, the Einstein metric $\tilg$ and its conformally extended representative $g$ can be written in the so-called {\it normal form w.r.t. $\gamma$}
\begin{equation}\label{eqnormform}
 \tilg =\frac{1}{\Om^2} \lr{-\frac{\dif \Om^2}{\lambda} + g_\Om},\quad\quad g = \Om^2 \tilg = -\frac{\dif \Om^2}{\lambda} + g_\Om,
\end{equation} 
where $g_\Om$ is the metric induced by $g$ on the leaves $\Sigma_\Om = \{ \Om = const. \}$.
The FG expansion is an asymptotic series expansion of the metric $g_\Om$ of the form 
\begin{align}
 g_\Om  & \sim \sum\limits_{r = 0}^{(n-1)/2} g_{(2 r)}  \Om^{2r}  + \sum\limits_{r =n}^{\infty} g_{(r)}  \Om^r,\quad\quad &&\mbox{if $n$ is odd},\label{eqFGexpoddsec}\\
  g_\Om  & \sim \sum\limits_{r = 0}^{\infty} g_{(2 r)}  \Om^{2r}  + \sum\limits_{r = n}^{\infty} \sum\limits_{s = 1}^{m_r} \obs_{(r,s)} \Om^{r} (\log \Om)^s ,\quad\quad &&\mbox{if $n$ is even}\label{eqFGexpevsec},
\end{align}
where the coefficients  $g_{(r)}$ and $\obs_{(r,s)}$ are defined at $\scri$ and extended to $M$ as independent objects of $\Om$, in Gaussian coordinates. The logarithmic terms occur in the $n$ even case if the so-called obstruction tensor $\obs: = \obs_{(n,1)}$ is non-zero. Note that the presence of logarithmic terms makes the metric non-smooth at $\{\Om = 0\}$. It turns out that the metric is smooth if and only if $\obs =0$, because then all logarithmic terms vanish. The obstruction tensor is conformally invariant and completely determined by the boundary metric $\gamma$ and it vanishes identically when $\gamma$ is conformally flat.

In general, the FG expansion is only formal, thus not necessarily convergent, and all its terms can be generated by imposing that the Einstein equations are satisfied to infinite order at $\scri$, 
which are necessary (in general not sufficient) conditions for $\tilg$ to be Einstein\footnote{Note that although $\tilg$ is defined only in $\mathrm{Int}(\man)$, being Einstein means $\widetilde R_{\alpha \beta} - n \lambda \tilg_{\alpha \beta} = 0$, so this tensor and all its derivatives must extend as zero to $\scri$.}. The coefficients that need to be prescribed are the zero-th order $g_{(0)}$ and the $n$-th order $g_{(n)}$ ones, and the rest are recursively generated from them. Thus, for each pair of coefficients $(g_{(0)},g_{(n)})$ there exists a unique FG expansion \eqref{eqFGexpoddsec},\eqref{eqFGexpevsec}. Hence, for a given Einstein metric, one can associate a FG expansion generated by $g_{(0)} = \gamma$ (the induced metric at $\scri$) and $g_{(n)}$  (the $n$-th order derivative of $g_{\Om}$ a $\scri$).

The converse is more delicate, i.e. to associate an Einstein metric $\tilg$ to a formal series (equivalently to a pair of coefficients $(g_{(0)}, g_{(n)})$ ). As already mentioned in Section \ref{secpre}, this question was first addressed in the analytic case (see \cite{FeffGrah85},\cite{kichenassamy03},\cite{Rendall03}), where it was found that for analytic data $(g_{(0)}, g_{(n)})$, the FG expansion converges to a unique Einstein metric $\tilg$ defined in a sufficiently small neighbourhood of $\scri$. In the smooth setting, well-posedness of the asymptotic Cauchy problem with data $(g_0, g_{(n)})$ was proved by Anderson, Chruściel and Kamiński \cite{Anderson2005,andersonchrusciel05,kaminski21}) in the case of even spacetime dimension (i.e. $n$ odd).

% Irrespectively on whether an Einstein metric can be associated to a FG expansion, the coefficients of the FG expansion are recursively generated from seed data $(g_{(0)},g_{(n)})$, geometric identities and the vanishing of the Einstein equations and its derivatives at $\scri$ (see e.g. \cite{Anderson2004},\cite{marspeondata21}).

The fundamental properties of the FG expansion that we shall need are summarized in the next lemma. The proof can be found in \cite{ambientmetric},  and we refer to \cite{Anderson2004} (see also \cite{marspeondata21}) for a discussion specifically taylored for the asymptotic initial value problem of the Einstein equations in the $\lambda >0$ case.
\begin{lemma}[Properties of the FG expansion]\label{lemmapropFG}~
\begin{enumerate} 
\item Each coefficient $g_{(r)}$ with $0 <r < n$ depends on previous order coefficients up to order $g_{(r-2)}$ and tangential derivatives of them up to second order. This is also true for $n<r$ if $n$ odd or $n$ even with $\obs = 0$. If $n$ is even and $\obs \neq 0$, the terms $g_{(r)}$ and $\obs_{(r,s)}$ with $n<r$ depends on previous terms up to order $g_{(r-2)}$ and $\obs_{(r-2,l)}$.

 \item Up to order $n$, both expansions \eqref{eqFGexpoddsec}, \eqref{eqFGexpevsec} are even and all terms $g_{(r)}$ with $r<n$ or $r = n+1$ (but not $r= n$) are solely generated from $\gamma$. If $n$ is even, $\obs$ is also generated from $\gamma$ and in particular, $\obs = 0$ if $\gamma$ is conformally flat.
 
 \item The $n$-th order coefficient $g_{(n)}$ is independent on previous terms except for 
 \begin{equation}\label{eqTTgsec}
 \mathrm{Tr}_\gamma g_{(n)} = \a, \quad\quad \mathrm{div}_\gamma g_{(n)} = \b,
\end{equation}     
where $\a = 0,~\b = 0$ for $n$ odd and $\a$ is a scalar and $\b$ a one-form determined by $\gamma$ for $n$ even. 

\item For analytic data $(\gamma, g_{(n)})$, the FG expansion converges.  
\end{enumerate}

\end{lemma}

In the context of general relativity, a well-posed Cauchy problem can be used to characterize geometrically Einstein metrics $\tilg$ in terms of their asymptotic initial data. In particular, for $n$ odd or analytic data, one may use the data $(g_{(0)}, g_{(n)})$ of the FG expansion. The zero-th order coefficient is simply the boundary metric $\gamma$, but $g_{(n)}$ needs to be calculated as the $n$-th order derivative of $g_\Om$ in Gaussian coordinates. Finding these coordinates is not a necessarily easy calculation, so it is clearly a step forward to have a geometric expression for $g_{(n)}$. 
This was achieved in \cite{marspeondata21} (also in \cite{Holl05} for the negative $\Lambda $ case) for the case of conformally flat boundaries. To do that, one must first extract a TT part $\mathring g_{(n)}$, which we denote {\it free part}, from the $n$-th order coefficient $g_{(n)}$. By setting background data $(\gamma, \barg_{(n)})$ as those corresponding to de Sitter, one defines the free part as $\mathring g_{(n} := g_{(n)} - \barg_{(n)}$, in such a way that any pair of initial data  $(\gamma, g_{(n)})$ is equivalent to $(\gamma, \mathring g_{(n)})$. Then, it is proven that $\mathring g_{(n)}$ agrees, up to a constant, with the electric part of the rescaled Weyl tensor \eqref{eqelecresc},
where $\nor^\alpha = T^\alpha/|T|$ is the unit timelike, normal to $\Sigma_\Om$.
\begin{theorem}[{\hspace{-0.025cm}\cite{marspeondata21}}]\label{theodatascri}
    For every Einstein metric $\tilg$ of dimension  admitting a smooth conformally flat $\scri$, the data which determine the FG expansion $(g_{(0)},\mathring g_{(n)})$ can be geometrically identified with $(\gamma, D)$, where $\gamma = g_{(0)}$ is the boundary metric and $D = -\frac{\lambda}{2}n(n-2) \mathring g_{(n)}$.
\end{theorem}

It is clear from the symmetries of the Weyl tensor that $D_{\alpha \beta}$ has only non-trivial components along the tangential directions to $\Sigma_\Om$. 
Expression \eqref{eqelecresc} depends on the choice of conformal extension because although the Weyl tensor is conformally invariant, the vectors  $T= \nabla\Omega$ and $\nor = T/|T|$  are not. Indeed, for two conformal extensions $g = \Om^2 \tilg$ and $g' = \Om'^2 \tilg$ of a given physical metric $\tilg$, the conformal factors are related by a smooth positive function $\omega$ satisfying $\Om' = \omega \Om$.
Hence, the respective normal covectors $\nabla_\alpha \Om$ and $\nabla_\alpha  \Om'$ and their metrically associated vector fields satisfy
\begin{equation}\label{eqnormvecscri}
 \nabla_\alpha \Om' = {\Omega \nabla_\alpha\omega + \omega \nabla_\alpha\Om},
 \quad\Longrightarrow\quad
{g'}^{\alpha\beta} \nabla_{\alpha} \Omega' 
\nabla_{\beta} \Omega'  |_{\scri}  =
g^{\alpha\beta} \nabla_{\alpha} \Omega 
\nabla_{\beta} \Omega |_{\scri}.
\end{equation}
In particular, this means that the unit normals $u$ and $u'$ point in the same direction at $\scri$ and they satisfy $u'_\alpha|_\scri = \omega u_\alpha |_\scri$ and $u'^\alpha |_\scri = \omega^{-1} u^\alpha|_\scri$. Hence, their respective rescaled Weyl tensors $D_{\alpha \beta}$ and $D'_{\alpha \beta}$ at $\scri$ are related by
\begin{equation}
 D'_{\alpha\beta} =   \restr{(\Om'^{2-n} {C^\mu}_{\alpha \nu \beta} \widehat \nor_\mu \widehat\nor^\nu)}{\scri} = \restr{(\omega^{2-n}\Om^{2-n} {C^\mu}_{\alpha \nu \beta} \nor_\mu \nor^\nu)}{\scri} = \restr{\omega^{2-n} }{\scri} D_{\alpha \beta}.
\end{equation}
This, in particular, proves the equivalence of data given in Remark \ref{remarkequivdata}.

Once we have geometrically identified the initial data for spacetimes admitting a smooth conformally flat $\scri$, the next two Lemmas give formulae  to compute the electric part of the rescaled Weyl tensor. The proof can be also found in \cite{marspeondata21}.

 \begin{lemma}\label{lemmaWeyls}
 Let $\tilg$ be a conformally extendable metric satisfying \eqref{EFE} and $g = \Om^2 \tilg$ a geodesic conformal extension. Then, in Gaussian coordinates $\{\Om, x^i\}$, the electric part of the Weyl tensor w.r.t. $u = T/|T|_g$ reads
 \begin{equation}\label{eqCtnocov}
 (C_\perp)_{ij} =   {C^\mu}_{i \nu j} \nor_\mu \nor^\nu = \frac{\lambda}{2}  \lr{\frac{1}{2} \partial_{\Om} (g_{ik}) g^{kl} \partial_{\Om} (g_{lj})+  \frac{1}{\Om} \partial_\Om g_{ij} - \partial^{(2)}_\Om g_{ij}},
\end{equation}
where $g_{ij} = (g_\Om)_{ij}$ is the metric induced by $g$ on the leaves $ \{\Om = const. \}$.
\end{lemma}

In the next lemma the function $\Om$ has a priori nothing to do with a conformal factor, so the result can be applied to very general situations. We use this notation because the result will later be applied in situations where $\Omega$ is the conformal factor.

\begin{lemma}\label{lemmaWeyls2}
 Let $g$, $\widehat g$ be $(n+1)$-dimensional  metrics related by $g = \widehat g + \Om^m q$, for $m\geq 2$, with $g, \widehat{g}, q$  and $\Om$ at least $C^2$ in a neighbourhood of $\{\Om=0\}$. Assume that $\nabla \Om$ is nowhere null at $\Om=0$. Then their Weyl tensors satisfy the following expression
 \begin{equation}\label{eqexpweyl}
 {C^\mu}_{\nu\alpha\beta } = {\widehat C^\mu}_{~~\nu\alpha\beta } -  K_m(\Om) \frac{n-2}{n-1}(\nor^\mu \nor_{[\alpha} \tlt_{\beta]\nu} + {\tlt^\mu}_{~~ [\alpha}\nor_{\beta]}\nor_\nu) + \frac{\epsilon K_m(\Om)}{n-1}({\pi^\mu}_{[\alpha} \tlt_{\beta]\nu} + \tlt^\mu_{~~ [\alpha}  \pi_{\beta]\nu}) + o(\Om^{m-2})
\end{equation}
with 
\begin{equation}
 K_m(\Om) =  m(m-1) \Om^{m-2} F^2,
\end{equation}
and where $\nabla \Om = F \nor$, for $g(\nor, \nor) = \epsilon = \pm 1$, $\pi_{\alpha\beta}$ is the projector orthogonal to $\nor$, all indices are raised and lowered with $g$, $t_{\alpha\beta}= q_{\mu\nu}{\pi^\mu}_\alpha {\pi^\nu}_\beta$ while $t$ and $\tlt_{\alpha \beta}$ are its trace and traceless part respectively.
 \end{lemma} 

Lemmas \ref{lemmaWeyls} and \ref{lemmaWeyls2} are particularly suited to exploit the Kerr-Schild structure of the metrics \eqref{eqKSmetrics}. Indeed, the fact that the background metric determines the geometry of $\scri$, implies that the information of the remaining data $D$ must be encoded in $\H \til k \otimes \til k$. To extract this information (in Section \ref{secKSsubKdS}) we shall also need a few more results, which we quote next.

 The next lemma states the form of the FG expansion of a de Sitter metric. For a proof see \cite{marspeondata21} and also \cite{GrahamLee91}, \cite{skenderis} (the latter assume $\lambda<0$ but the proof is analogous when $\lambda >0$). 
 
    \begin{lemma}\label{lemmaFGdS}
     For every conformally flat boundary metric $\gamma$, let $\barg$ be the spacetime metric defined by
 \begin{equation}\label{eqpoindS}
 \barg := -\frac{\dif \Om^2}{\lambda} + \barg_{\Om}, \quad\quad \barg_{\Om} := \gamma + \frac{P}{\lambda}{\Om^2} + \frac{1}{4} \frac{P^2}{\lambda^2} \Om^4
\end{equation}
where $P$ is the Schouten tensor of $\gamma$
\begin{equation}\label{eqschoutens}
 P_{ij} :=\frac{1}{n-2}\lr{R_{ij} - \frac{R}{2(n-1)}\gamma_{ij}}, \quad\mbox{and}\quad  (P^2)_{ij} :=  P_{il} \gamma^{kl}  P_{lj}.  
\end{equation}
Then $\tilg_{dS} := \Om^{-2}\barg $ is locally isometric to the de Sitter metric.
    \end{lemma}
    
    Lemma \ref{lemmaFGdS} gives the FG expansion of metrics conformally isometric to de Sitter, but from property 2 of Lemma \ref{lemmapropFG}, it also determines the terms up to order $n$ of the FG expansion of any metric admitting a smooth conformally flat $\scri$. Consequently, for any such metric, the terms generated exclusively by the boundary metric $\gamma$ stop at fourth order. This implies that for $n=3$, a conformally flat $\gamma$  generates a term of order $n+1 = 4$ , which is not only independent on the $n$-th (i.e. third) order one by property 1 of Lemma \ref{lemmapropFG}, but actually must take the form $g_{(4)} = P^2/(4\lambda^2)$ by Lemma \ref{lemmaFGdS}. On the other hand, for $n>3$, the $n+1$ order term only depends on $\gamma$ by property 2 of Lemma \ref{lemmapropFG}. Hence, by Lemma \ref{lemmaFGdS} it must be zero.
    That is, if $\tilg$ is an Einstein metric admitting a smooth conformally flat $\scri$, then for every geodesic conformal extension $g = \Om^2 \tilg$, the FG expansion yields the following decomposition
    \begin{equation}\label{eqdecconflat}
   g = \barg + Q,
  \end{equation}
where $\barg$ is of the form \eqref{eqpoindS} (thus conformally isometric to de Sitter) and $Q$ is both $O(\Om^n)$ and has no term of order $\Om^{n+1}$ (when $n = 3$ this term exists in $g$ but it is included in $\barg$).
\begin{remark}\label{remarkTTg4}
  Note that by construction, the leading order term of $Q$ in decomposition \eqref{eqdecconflat} is precisely $\mathring{g}_{(n)}$, the free part of the $n$-th order coefficient. As mentioned above, this equals $g_{(n)}$ if $n$ odd, but not in general if $n$ even. However, in the conformally flat $\scri$ case and $n > 4$, it is also true that $\mathring{g}_{(n)} = g_{(n)}$. The argument is that the boundary metric determines the trace and divergence of $g_{(n)}$ (cf. property 3 of Lemma \ref{lemmapropFG}). If $n >4$ and $g$ is conformal to de Sitter there are no terms of order $n >4$, hence, a conformally flat boundary metric $\gamma$ cannot generate trace and divergence of the $n$-th order term.  On the contrary, for $n=4$ and conformally flat $\scri$ the $n$-th order term has trace and divergence determined by $\gamma$, which in decomposition  \eqref{eqdecconflat} we have included in $\barg$. In this case, we shall distinguish the two terms $g_{(4)} = \barg_{(4)} + \mathring{g}_{(4)}$, with $\barg_{(4)} = P^2/4$ (cf. equation \eqref{eqpoindS}).
\end{remark}

A similar converse statement also holds (see \cite{marspeondata21}). Namely, if $\tilg$ is a conformally extendable metric which for some geodesic conformal extension $g = \Om^2 \tilg$  admits a decomposition of the form 
\begin{equation}\label{eqdecconflat2}
g = \hatg + \widehat Q,
\end{equation}
  with $\hatg$ conformally isometric to de Sitter and $\widehat Q = O(\Om^n)$, then $\scri$ is conformally flat. One must be careful with the fact that $\hatg$ being conformally isometric to de Sitter does not mean that it takes the form \eqref{eqpoindS} for the conformal factor $\Om$ which is geodesic for $g$. Indeed, $\hatg$ does admit an expansion of the form $\eqref{eqpoindS}$ for some conformal factor $\widehat \Om$ geodesic w.r.t. $\hatg$, but in general this conformal factor is different to $\Om$. Thus, decomposition \eqref{eqdecconflat} is a very particular decomposition for metrics admitting a smooth conformally flat $\scri$, while decomposition \eqref{eqdecconflat2} is a sufficient condition for $g$ to admit a conformally flat $\scri$. Obviously, a metric which can be decomposed as in \eqref{eqdecconflat2} can also be decomposed as in \eqref{eqdecconflat}, but these decompositions do not in general coincide.

Both decompositions \eqref{eqdecconflat} and \eqref{eqdecconflat2} will be used in this section, so we summarize the above discussion in the following Proposition:

  \begin{proposition}\label{propconflatdec}
  Let $\tilg$ be an $n\geq 3$ dimensional conformally extendable metric satisfying \eqref{EFE} and let $g = \Om^2 \tilg$ be a geodesic conformal extension. Then 
  \begin{enumerate}
   \item[a)] If $\scri$ is conformally flat, then $g$ admits a decomposition of the form \eqref{eqdecconflat} with $\barg$ of the form \eqref{eqpoindS} and $Q= O(\Om^{n})$ with no terms in $\Om^{n+1}$.
   \item[b)] If $g$ admits a decomposition of the form \eqref{eqdecconflat2}, with $\hatg$ conformally isometric to de Sitter and $\widehat Q = O(\Om^{n})$, then $\scri$ is conformally flat.
  \end{enumerate} 
 \end{proposition}

  \section{Kerr-Schild-de Sitter $\subset$ Kerr-de Sitter-like class}\label{secKSsubKdS}

    In this section we prove the inclusion of the Kerr-Schild-de Sitter spacetimes in the Kerr-de Sitter-like class. This is done by direct calculation of the data at spacelike $\scri$ of the Kerr-Schild-de Sitter spacetimes and by showing that the  vector field  $\Y$ at $\scri$ that arises in the expression of $D_\Y$ is in fact a CKVF of $\gamma$.
    
    A key ingredient for this result is that all vacuum Kerr-Schild spacetimes are algebraically special in the Petrov classification. Recall that the Petrov classification is an algebraic classification of the Weyl tensor based on the vanishing of the components with certain boost weight, as we summarize next. In the case of arbitrary dimension this classification was developed in \cite{coley04,coley08,milson05} to which we refer for further details.   
    Consider a null frame of vectors $\{\til k,\til l,\til m_{(i)}\}$ for $i = 1,\cdots,n-1$ (whose indices are raised/lowered with $\tilg$), i.e.  a frame satisfying
    \begin{equation}\label{eqnullf}
     \til k^\alpha \til k_\alpha = \til l^\alpha \til l_\alpha = \til k^\alpha \til m_{(i) \alpha} = 0,\quad\quad \til k^\alpha \til l_\alpha = -1,\quad\quad {\til m_{(i)}}^\alpha \til m_{(j) \alpha} = \delta_{ij}.
    \end{equation}
This frame maintains its properties \eqref{eqnullf} under the following set of boost transformations
\begin{equation}
 \til k' = b \til k,\quad\quad \til l' = b^{-1} \til l, \quad\quad \til m_{(i)}' = \til m_{(i)},
\end{equation}
for every real non-zero parameter $b$. Thus, the components of the Weyl tensor $C$ expressed in this frame have ``boost weight'' depending on the number of contractions with $\til k$, $\til l$ and $\til m_{(i)}$. Namely, +1 for each contraction with $\til k$; $-1$ for each one with $\til l$; and $0$ for each one with $\til m_{(i)}$. From the symmetries of the Weyl tensor, the maximum boost weight of a component is $+2$ and the minimum is $-2$. 
The classification proceeds by looking for vectors $\til k$ such that the highest boost weight components vanish. One such $\til k$ (when it exists) is called a {\it Weyl aligned null direction} (WAND) and if the components of boost weight 1 or lower also vanish, $\til k$ is  called a {\it multiple} WAND. A spacetime which admits a multiple WAND is said to be {\it algebraically special}.

It turns out \cite{malek11} that all $\Lambda>0$-vacuum Kerr-Schild spacetimes are algebraically special. Hence, for this section, the following result will be key:

    \begin{proposition}[\hspace{-0.025cm}\cite{malek11}]\label{propksspec}
     Kerr-Schild-de Sitter spacetimes \eqref{eqKSmetrics} are algebraically special, with $\til k$ a multiple WAND satisfying
     \begin{equation}\label{eqmultWAND}
      {\widetilde C}_{\mu \alpha \nu \beta} \widetilde k^\mu \widetilde k^\nu \til m_{(i)}^\alpha \til m_{(j)}^\beta ={\til C}_{\mu \alpha \nu \beta} \til k^\mu \til k^\nu \til l^\alpha \til m_{(i)}^\beta = C_{\mu \alpha \nu \beta} \til k^\mu \til m_{(i)}^\alpha \til m_{(j)}^\nu \til m_{(k)}^\beta = 0, 
     \end{equation}
     for a suitable null frame $\{\til k,\til l,\til m_{(i)}\}$. Moreover, $\til k$ is geodesic, so after rescaling if necessary, it satisfies
      \begin{equation}\label{eqgeodtilk}
      \til k^\alpha \til \nabla_\alpha \til k_\beta = 0.
     \end{equation}
    \end{proposition} 
   
     We shall assume for now on that $\til k$ has been scaled so that \eqref{eqgeodtilk} holds.

 Let $\tilg$ be a Kerr-Schild-de Sitter spacetime and consider a geodesic conformal extension $g= \Om^2 \tilg$. Then, the conformal metric and its associated contravariant metric $g^\sharp$ are
 \begin{equation}\label{eqgeodKS}
 g_{\alpha \beta} = \Om^2 \tilg = \widehat g_{\alpha \beta} + \H k_\alpha ~k_\beta,\quad\quad g^{\alpha \beta} = \Om^{-2} \tilg^{\alpha \beta} = \widehat g^{\alpha \beta} - \H k^\alpha k^\beta,
\end{equation}
where $\widehat g := \Om^2 \tilg_{dS}$, $\H := \Omega^2 \widetilde\H$ and $k = \widetilde k$ is a field of one-forms whose metrically associated vector field $k^\alpha$ by $g$ has components $k^\alpha = g^{\alpha \beta} k_{\beta}= \Om^{-2} \tilg^{\alpha \beta} \widetilde k_\beta = \Om^{-2} \widetilde k^\alpha$, where $\widetilde k^\alpha$ is the vector field associated to $\widetilde k$ by $\tilg$.
Recall that the difference of connections $\nabla - \widetilde \nabla = Q$ is, for conformally related metrics $g = \Om^2 \tilg$, the tensor
\begin{equation}\label{eqQtens}
  {Q^\mu}_{\alpha \beta} = \frac{1}{\Om} \lr{T_\alpha {\delta^\mu}_\beta  
  + T_\beta {\delta^\mu}_\alpha - T^\mu g_{\alpha \beta} },\quad\quad T_\mu 
  := \nabla_\mu \Om, \quad T^\mu:=g^{\mu \nu} T_\nu.
\end{equation}
We recall the well-known property that $k_\alpha$ is geodesic w.r.t. $g$ if and only if $\widetilde k_\alpha$ is geodesic w.r.t. $\tilg$. Indeed
\begin{align}\label{eqkgeod}
 k^\alpha \nabla_\alpha k _\beta & = k^\alpha \widetilde \nabla_\alpha k _\beta - {Q^\mu}_{\alpha \beta} k^\alpha k_\mu = k^\alpha \widetilde \nabla_\alpha k _\beta 
 = \Om^{-2} \widetilde{k}^\alpha \widetilde \nabla_\alpha \widetilde k_\beta. 
\end{align}
Thus combining equation \eqref{eqkgeod} with Proposition \ref{propksspec}, $k$ must be geodesic w.r.t. $g$. In addition,
the conformal invariance of the Weyl tensor implies that $k$ is a multiple WAND for the Weyl tensor of $\tilg$ if and only if it is a WAND for the Weyl tensor of $g$. That is, by Proposition \ref{propksspec} and the above discussion, $k$ (as a one-form) is also a geodesic multiple WAND for $g$. In what follows, it will be useful to decompose $k$ in tangent and normal components to a timelike unit vector $u$. Specifically, given one such $u$, we write
\begin{equation}\label{eqprojkuy}
 k_\alpha  = s(\nor_\alpha + \y_\alpha), 
\end{equation}
which defines both the scalar $s$ and the spacelike unit vector $y$ perpendicular to $u$. Except in the trivial case that the Kerr-Schild metric is identical to the backgroud metric, it is clear that $\H k \otimes k$ cannot be identically zero. We let $U$ be a domain of the physical spacetime $\til \man$  where this quantity is not zero. We are only interested in the case where $\overline{U}$ intersects $\scri$ as otherwise  the free-data  $\mathring{g}_n$ is identically zero, and the Kerr-Schild metric would be identical to the background metric in some neighbourhood of $\scri$. Since $k$ is geodesic, affinely parametrized and nowhere zero in $(U,g)$, it must extend smoothly and nowhere zero to $\scri \cap \partial U$. This is because $g$-null geodesics starting sufficiently close to $\scri$ with non-zero tangent reach $\scri$ (smoothly). Since the tangent vector to the geodesic cannot vanish anywhere along the curve, we conclude that the covector $k$ is nowhere zero in $\scri \cap \partial U$. From now on we shall work on the manifold with boundary $\overline{U}$ so that its infinity  (still called $\scri$)  is such that $k$ is nowhere vanishing there.

In the next lemma, we summarize the important properties of $k$ w.r.t. to the conformal metric $g$
\begin{lemma}
 Let $\tilg$ be a Kerr-Schild-de Sitter metric and let $g = \Om^2 \tilg$ be a conformal extension. Assume that $\tilg$ is not identically equal to the background metric in some neighbourhood of $\scri$. Then, after restricting $\til\man$ if necessary, $k$ extends smoothly and nowhere zero to $\scri$ and it is both geodesic affinely parametrized w.r.t. to $g$ 
 \begin{equation} 
       k^\alpha  \nabla_\alpha  k_\beta = 0
     \end{equation}
    and a multiple WAND with
     \begin{equation}\label{eqmultWAND}
      { C}_{\mu \alpha \nu \beta}  k^\mu  k^\nu  m_{(i)}^\alpha  m_{(j)}^\beta ={ C}_{\mu \alpha \nu \beta}  k^\mu  k^\nu  l^\alpha  m_{(i)}^\beta = C_{\mu \alpha \nu \beta}  k^\mu  m_{i}^\alpha  m_{(j)}^\nu  m_{(k)}^\beta = 0, 
     \end{equation}
     for a suitable null frame $\{ k, l, m_{(i)}\}$ for $g$.
\end{lemma}

The Kerr-Schild ansatz gives a decomposition for the metrics \eqref{eqgeodKS} similar to the one in \eqref{eqdecconflat2}, where, however, $\widehat Q = \H k  \otimes k$ is in principle not necessarily $O(\Om^n)$. We now prove that Definition \ref{defKSdS} forces that necessarily $\H = O(\Om^n)$. In the following, we use the same name for a geometric object and its restriction to $\scri$ (we let the context clarify the meaning). This applies in particular to the vector $y$.

\begin{lemma}\label{lemmadecayH}
 Let $\tilg$ be a Kerr-Schild de Sitter spacetime and consider a geodesic conformal extension $g = \Om^2 \tilg$ as in \eqref{eqgeodKS}, inducing a (conformally flat) metric $\gamma$ at $\scri$. Then, $\H = O(\Om^n)$ and the electric part of the rescaled Weyl tensor at $\scri$ is 
 \begin{equation}\label{eqTTDA}
 D_{\alpha\beta} =  \Am \lr{\y_\alpha \y_\beta - \frac{1}{n} \gamma_{\alpha\beta}},
\end{equation}
where the function $\Am$ at $\scri$ is given by $(\Om^{-n}\H s^2 )|_\scri= -\frac{2 \Am  }{\lambda n (n-2)}$.
\end{lemma}
\begin{proof}
 By definition \ref{defKSdS}, $\H = \Om^2 \til\H$ must be $O(\Om^m)$ with $m \geq 1$. Assume first that $m = 1$.  By property 2 of Lemma \ref{lemmapropFG} the FG expansion of $g = -\dif \Om^2/\lambda + g_{\Om} $ is even up to order $n$, where $g_\Om$ is given by \eqref{eqFGexpoddsec} if $n$ odd or  \eqref{eqFGexpevsec} if $n$ even (with vanishing logarithmic terms because $\gamma$ is conformally flat).
Then, using the Kerr-Schild form $g  = \hatg + \H k \otimes k$ and expanding $ \hatg $ and $ \H k \otimes k$ in $\Om$, the non-zero terms of order $\Om$ of the tangent-tangent (i.e. tangent to $\Sigma_\Om = \{\Omega = const. \}$) components of $\hatg$ must  cancel out those of $\H k \otimes k$. To expand $\hatg$ in powers of $\Om$, consider a geodesic conformal factor\footnote{Notice that $\hatg = \widehat\Om^2 \tilg'_{dS}$, where $\tilg'_{dS}$ is locally de Sitter, isometric to the original one $\tilg_{dS}$, but not equal.} $\widehat\Om$ for $\hatg$, which induces the same boundary metric $\gamma$ at $\scri = \{ \Om = 0\} =  \{ \widehat \Om = 0\}$. The existence of such conformal factor follows by Lemma \ref{lemmanormgeod} and it must satisfy $\Om  = \widehat \Om \omega$, with $\restr{\omega}{\scri} = 1$.
By Lemma \ref{lemmaFGdS}, the FG expansion of $ \hatg$,  in Gaussian coordinates $\{\widehat \Om, \widehat x^i \}$ adapted to the foliation $\Sigma_{\widehat\Om} = \{ \widehat \Om = const. \}$, is given by \eqref{eqpoindS} 
  \begin{equation}\label{eqlem1}
 \hatg = -\frac{\dif \widehat \Om^2}{\lambda} + \hatg_{\widehat\Om}, \quad\quad \hatg_{\widehat\Om} = \gamma +\frac{P}{\lambda} {\widehat \Om^2} + \frac{1}{4}\frac{P^2}{\lambda^2}  \widehat \Om^4
\end{equation}
where $P$ is the Schouten tensor of $\gamma$. In order to compare with the expansion of $g_\Om$, one has to relate the conformal factors, but also the tangent directions. First, as $g|_\scri = \widehat g|_\scri = \gamma$ we can choose tangent coordinates satisfying $\widehat x^i = x^i + \Om z^i $, for a collection of functions $\{ z^i \}$ (still depending on $\Om$). We use now, as shown  before, that the vectors $\partial_\Om$ and $\partial_{\widehat \Om}$ are proportional at $\scri$
\begin{align}
 \partial_\Om|_{\Om =0} & = \lr{\partial_\Om \widehat x^i \partial_{\widehat x^i} + \partial_\Om \widehat \Om} \restr{\partial_{\widehat \Om}}{\Om = 0} = \lr{z^i + \Om \partial_\Om z^i } \restr{\partial_{\widehat x^j}}{\Om = 0} + \lr{ \omega + \Om \partial_{\widehat \Om} \omega }\restr{\partial_{\widehat \Om}}{\Om = 0}  =  \restr{\partial_{\widehat \Om}}{\Om = 0}.
\end{align}
Thus $z^i|_{\widehat \Om = 0} = 0$ so $z^i = O(\Om)$ and $\widehat x^i = x^i + O(\Om^2)$. This implies that when $\gamma$ (which recall is extended off $\scri$ as independent of $\widehat{\Om}$ in the Gaussian coordinates $\{\widehat\Omega, \widehat x^i \}$) is written in coordinates $x^i$, it does not add tangent-tangent terms ($\dif x^i \dif x^j$) of order $\Om$ and obviously neither they do the rest of terms in $\widehat g_{\widehat\Om}$ in \eqref{eqlem1}, because $\widehat \Om = \Om \omega$. On the other hand, $\dif \widehat \Om^2$ is 
\begin{equation}
 \dif \widehat \Om^2 = (\omega \dif \Om + \Om \dif \omega)^2 = \omega^2 \dif \Om + \Om^2 \dif \omega^2 + 2 \Om \dif \Om \dif \omega
\end{equation}
and the only tangent-tangent terms can only appear in $\Om^2 \dif \omega^2$, thus starting (at least) at order $\Om^2$.  Therefore the expansion of $\hatg $ in the conformal factor $\Om$ does not have first order terms, so neither it does $\H k \otimes k$ because the FG expansion of $g$ does not have such a term. This implies that $m \geq 2$.

Let us expand $\H$ as 
\begin{equation}
 \H =  -\frac{2 \Am  }{\lambda n (n-2)} (s^{-2}|_\scri) \Om^m + o(\Om^{m}),
\end{equation}
and note that $s$ that does not vanish anywhere (because $k$ has this property). 
By Lemma \ref{lemmaWeyls2}, the electric part of the Weyl tensor is straightforwardly calculated
 \begin{equation}\label{eqweylh}
  C_\perp = \Am \lr{y \otimes y - \frac{|y|^2}{n} g_\Om} \Om^{m-2} + o(\Om^{m-2})
 \end{equation}
where we have used that $\hatg$ is conformally flat, so that $\widehat C = 0$, and $\nabla \Om$ is geodesic, thus $F^2 = \lambda$, and $\epsilon = -1$ (cf. Lemma \ref{lemmanormgeod}). Now applying Theorem \ref{theodatascri}, scaling \eqref{eqweylh} by $\Om^{2-n}$ and evaluating at $\Om = 0$ must give the TT part of the $n$-th order coefficient of the FG expansion, so $m \geq n$. But $m > n$ gives $\mathring{g}_{(n)} = 0$, which by uniqueness of the FG expansion would imply that $\tilg$ is equal to its background metric, against hypothesis. Thus $m=n$ and the lemma follows after scaling \eqref{eqweylh} by $\Om^{2-n}$ and evaluating at $\scri$.
\end{proof}

In conclusion, the initial data for Kerr-Schild-de Sitter spacetimes are a conformally flat class of metrics $[\gamma]$ and a TT tensor of the form \eqref{eqTTDA}. The function $\Am$ cannot be identically zero at $\scri$ (as otherwise $\tilg$ would equal its background metric in a neighbourhood of $\scri$). After restricting $\man$ further we may therefore assume that $\Am$ is nowhere zero at $\scri$ and we may 
reparametrize it as $\Am =: \kappa/f^n$, with $f$ everywhere positive and $\kappa \in \mathbb{R}$ is a constant that carries the sign of $\Am$. For later convenience we do not normalize $\kappa$ to be $\pm 1$, which means that we keep an arbitrary (positive) scaling freedom in $f$. Then, the TT tensor $D$ of Lemma \ref{lemmadecayH} can be written as
 \begin{equation}\label{eqTTDY2}
 D = \kappa D_\Y,\quad \quad (D_\Y)_{\alpha\beta} :=  \frac{1}{f^{n+2}} \lr{\Y_\alpha \Y_\beta - \frac{f^2}{n} \gamma_{\alpha\beta}},
\end{equation}
with $\Y_\alpha := f \y_\alpha$.  Our next aim is to prove that $\Y$ it must be a CKVF of $\scri$. 
The strategy is to rewrite the conditions of being CKVF in terms of equations for $f$ and $y$ and then show that they are satisfied as a consequence of $k$ being a WAND. 

\bigskip

Recall the following standard decomposition of the covariant derivative of a unit vector field $\y_\alpha$ in terms parallel and orthogonal to itself
\begin{equation}\label{eqdecnaby}
 \nabscr_\alpha y_\beta = y_\alpha a_\beta + \Pi_{\alpha \beta} + \frac{\h_{\alpha \beta}}{n-1} \L +  \w_{\alpha \beta},\quad\quad L:= \nabscr_\alpha \y^\alpha
\end{equation}
where $\nabscr$ the Levi-civita connection of $\gamma$, $a_\beta$ is a covector, $h_{\alpha\beta} = \gamma_{\alpha\beta} - y_\alpha y_\beta$ (the ``projector'' onto $(\spn{\y})^\perp$) and $\Pi_{\alpha \beta}$ symmetric traceless and $\w_{\alpha \beta}$ skew-symmetric, i.e.
\begin{equation}
 \Pi_{(\alpha \beta)} = \Pi_{\alpha \beta},\quad\quad {\Pi^\alpha}_\alpha = 0, \quad\quad \w_{[\alpha \beta]} = \w_{\alpha \beta},
\end{equation}
satisfying
\begin{equation}
 \y^\alpha \Pi_{\alpha \beta} = \y^\alpha \h_{\alpha \beta} = \y^\alpha \w_{\alpha \beta} = 0,\quad\quad y^\alpha a_\alpha = 0.
\end{equation}
In what follows, it will be useful to express the metric $\gamma$ as
\begin{equation}\label{eqdecgamma}
 \gamma_{\alpha \beta} = y_\alpha y_\beta + \h_{\alpha \beta}.
\end{equation}

\begin{lemma}\label{lemmackvf}
 Let $\Y^\alpha = f y^\alpha$, with $\y^\alpha$ unit, be a vector field of a Riemannian $n$-manifold $(\Sigma, \gamma)$ and consider the decomposition of $\nabscr_\alpha y_\beta$ as in \eqref{eqdecnaby}. Then $\Y$ is a CKVF of $\gamma$ if and only if the following equations are satisfied
 \begin{equation}\label{eqckvf}
  \nabscr_\alpha f = \frac{f L}{n-1} y_\alpha - f a_\alpha,\quad\quad \Pi_{\alpha \beta} = 0.
 \end{equation}
\end{lemma}
\begin{proof}
 We rewrite the conformal Killing equation
 \begin{equation}\label{eqCKeq}
  \nabscr_\alpha \Y_\beta + \nabscr_\beta \Y_\alpha =
   \frac{2}{n} \nabscr_\mu \Y^\mu \gamma_{\alpha \beta}
 \end{equation}
 in terms of the kinematical quantities above. Since
\begin{align}
 \nabscr_\alpha \Y_\beta + \nabscr_\alpha \Y_\beta & =   (\nabscr_\alpha f) \y_\beta + (\nabscr_\beta f) \y_\alpha + f (\nabscr_\alpha \y_\beta + \nabscr_\beta \y_\alpha) \\
 & = (\nabscr_\alpha f) \y_\beta + (\nabscr_\beta f)\y_\alpha + f \lr{\y_\alpha a_\beta + \y_\beta a_\alpha + 2 \Pi_{\alpha \beta} + \frac{2h_{\alpha \beta}}{n-1}L }
\end{align}
and
\begin{equation}
  \frac{2}{n} \nabscr_\mu \Y^\mu \gamma_{\alpha \beta}
  = 
  \frac{2}{n} (\y^\mu \nabscr_\mu f + f L)(y_\alpha y_\beta + h_{\alpha \beta}),
\end{equation}
$\Y$ is a CKVF if and only if
\begin{equation}\label{eqequivCKVF1}
 (\nabscr_\alpha f) \y_\beta + (\nabscr_\beta f) \y_\alpha + f \lr{\y_\alpha a_\beta + \y_\beta a_\alpha + 2 \Pi_{\alpha \beta} + \frac{2h_{\alpha \beta}}{n-1}L } 
 = 
 \frac{2}{n} (\y^\mu \nabscr_\mu f + f L)(y_\alpha y_\beta + h_{\alpha \beta}).
\end{equation}
One contraction with $\y^\alpha$ gives  
\begin{equation}\label{eqequivCKVF2}
 (\y^\alpha \nabscr_\alpha f) \y_\beta + \nabscr_\beta f + f a_\beta = \frac{2}{n}(\y^\mu \nabscr_\mu f + f L)y_\beta
\end{equation}
and a second contraction with $y^\beta$
\begin{equation}\label{eqequivCKVF3}
\y^\alpha \nabscr_\alpha f + \y^\beta \nabscr_\beta f  = \frac{2}{n}(\y^\mu \nabscr_\mu f + f L)
~\Longleftrightarrow~
y^\alpha \nabscr_\alpha f = \frac{f L}{n-1}.
\end{equation}
Inserting \eqref{eqequivCKVF3} in \eqref{eqequivCKVF2} gives the first of equation \eqref{eqckvf}. 
Projecting \eqref{eqequivCKVF1} with ${\h^\alpha}_\mu {\h^\beta}_\nu$ gives
\begin{equation}
 2f \lr{\Pi_{\mu \nu} + \frac{h_{\mu \nu}}{n-1} L}= \frac{2}{n}(\y^\mu \nabscr_\mu f + f L)h_{\mu \nu}  
\end{equation}
which is equivalent to $\Pi=0$ after using \eqref{eqequivCKVF3}. This proves the result in one direction. The converse follows immediately because  \eqref{eqequivCKVF2} is identically satisfied when \eqref{eqckvf} hold.  
\end{proof}
  
Coming back to the data corresponding to Kerr-Schild de Sitter metrics, we prove that the first equation in \eqref{eqckvf} is satisfied just by imposing $D$ to be TT. The argument for the second equation is more subtle and will be addressed right after. 

\begin{lemma}\label{lemmaeqckvf1}
Let $\tilg$ be a Kerr-Schild de Sitter metric and $g = \Om^2 \tilg$ a geodesic conformal extension. Then 
 \begin{equation}\label{eqnabf}
  \nabscr_\alpha f = \frac{f L}{n-1} y_\alpha - f a_\alpha.
\end{equation}
\end{lemma}
\begin{proof}
Consider $D_\Y= f^{-n}\lr{y \otimes y - (1/n) \gamma}$, which by Lemma \ref{lemmadecayH} is, up to a constant, the electric part of the rescaled Weyl tensor of $\tilg$. Thus, it is a TT tensor and the vanishing of its divergence gives by \eqref{eqdecnaby}
\begin{equation}\label{eqaux1}
 \nabscr_\alpha {(D_\Y)^\alpha}_\beta =-\frac{n}{f^{n+1}} \lr{y^\alpha \nabscr_\alpha f y_\beta - \frac{\nabscr_\beta f}{n}} +
 \frac{1}{f^{n}} \lr{ L y_\beta + a_\beta } = 0.
\end{equation}
Contracting with $y^\beta$ one has
\begin{equation}
 y^\alpha \nabscr_\alpha f = \frac{f L}{n-1}
\end{equation}
and inserting back into \eqref{eqaux1} we get \eqref{eqnabf}. This condition, which is precisely the first in \eqref{eqckvf}, is  not only necessary for \eqref{eqaux1} but also sufficient. 
\end{proof}

\bigskip

We next show that $\Pi_{\alpha \beta} = 0$. First notice that $K_\Om$, the second fundamental form of the leaves $\Sigma_\Om = \{\Om = const. \}$, can be written
\begin{equation}
 K_\Om = \frac{1}{2}(\mathcal{L}_\nor g_\Om) = -\frac{\lambda^{1/2}}{2} (2 \Om g_{(2)} + \cdots + n \Om^{n-1} g_{(n)} + \cdots ),
\end{equation}
where $\mathcal{L}_\nor$ denotes the Lie derivative w.r.t. the unit vector $\nor^\alpha\partial_\alpha= \lambda^{-1/2} \nabla^\alpha\Om\partial_\alpha =  -\lambda^{1/2} \partial_\Om$. This tensor appears in the Codazzi equation 
\begin{equation}\label{eqfullcodaz}
 \lr{\nabla_k (K_\Om)_{ij}  - \nabla_i (K_\Om)_{kj}}   = {R^\mu}_{ jik }\nor_\mu,
\end{equation}
where $i,j,k$ denote tangent directions to $\Sigma_\Om$.
The strategy consists in analyzing the $\Om^{n-1}$ order terms of the following components of the Codazzi equation
\begin{equation}\label{eqcodaz}
 \lr{\nabla_\nu (K_\Om)_{\beta\alpha}  - \nabla_\beta (K_\Om)_{\nu \alpha}}  {\h^\alpha}_{(\lambda} {\h^\beta}_{\sigma)} y^\nu  = {R^\mu}_{ \alpha \nu \beta}\nor_\mu {\h^\alpha}_{(\lambda} {\h^\beta}_{\sigma)} y^\nu,
\end{equation}
where we extend $h$ away from $\scri$ as the projector orthogonal to $y$ and $u$, i.e. $h := g + \nor \otimes \nor - \y \otimes \y$. The proof that $\Pi_{\alpha\beta} = 0$ consists in two main steps. Firstly, we prove that the  $\Om^{n-1}$ order term of the LHS of \eqref{eqcodaz} only involves the free part $\mathring g_{(n)}$. This, by Theorem \ref{theodatascri}, coincides up to a constant with the electric part of the rescaled by tensor, which in turn, by Lemma \ref{lemmadecayH}, is given by equation \eqref{eqTTDY2}. From these facts it follows that the LHS of \eqref{eqcodaz} is (up to a non-zero factor) $\Pi_{\alpha \beta}$. The second step consist in analyzing the RHS of \eqref{eqcodaz}. From the algebraically special condition, it follows that the symmetric part of its $\Om^{n-1}$ order term is pure trace. Since $\Pi_{\alpha \beta}$ is traceless, it follows  $\Pi_{\alpha\beta}=0$.

{Before carrying out this program, we list some standard identities and definitions that will be required for the rest of this section. Let $g^{(1)}$ and $g^{(2)}$ denote two metrics and $\nabla^{(1)}$, $ \nabla^{(2)}$ their respective Levi-Civita connections. 
Firstly, the difference of connections $\Q = \nabla^{(1)} - \nabla^{(2)}$ is the tensor given by (e.g. \cite{Kroonbook})
 \begin{align}
  {\Q^\mu}_{\alpha \beta} & := \frac{1}{2} ({g^{(1)\sharp}})^{\mu \nu}(\nabla^{(2)}_\alpha {g^{(1)}}_{\beta \nu} +\nabla^{(2)}_\beta {g^{(1)}}_{\alpha \nu} - \nabla^{(2)}_\nu {g^{(1)}}_{ \alpha\beta })\label{eqdefS},
%   \\
%   & = -\frac{1}{2} {g^{(2)}}^{\mu \nu}(\nabla^{(1)}_\alpha {g^{(2)}}_{\beta \nu} + \til\nabla^{(1)}_\beta {g^{(2)}}_{\alpha \nu} - \nabla^{(1)}_\nu {g^{(2)}}_{ \alpha\beta }).
 \end{align}
 where $g^{(1)\sharp}$ denotes the contravariant metric associated to $g^{(1)}$. 
 From this relation between the connections, a formula for the difference of Riemann tensors follows
\begin{equation}\label{eqdiffriems}
    {({R^{(1)}})^\mu}_{\alpha\nu\beta} - {({R^{(2)}})^\mu}_{\alpha\nu\beta} = 2 \nabla^{(2)}_{[\nu} {\Q^\mu}_{\beta] \alpha} 
 - 
 2 {\Q^\kappa}_{[\nu | \alpha|}{{\Q^\mu}_{\beta] \kappa}}.
% 
%  = 2 \nabla^{(1)}_{[\nu} {\Q^\mu}_{\beta] \alpha} 
%  + 
%  2 {\Q^\kappa}_{[\nu | \alpha|}{{\Q^\mu}_{\beta] \kappa}}. 
\end{equation}
Now let us consider two arbitrary conformally related metrics $g = \Om^2 \tilg$, with $g$ extending to $\Om = 0$. Recall  the definition of the Weyl tensor 
  \begin{equation}\label{eqWeyl}
 {C^\mu}_{\alpha \nu \beta} = {R^\mu}_{\alpha \nu \beta}  - \frac{2}{n-1}({\delta^\mu}_{[\nu}R_{\beta]\alpha} - g_{\alpha[\nu}{R^\mu}_{|\beta]}) + \frac{2 R}{n(n-1)}  {\delta^\mu}_{[\nu}g_{\beta]\alpha}.
\end{equation}
As before let $u$ be unit normal along $\nabla \Om$ and $i,j,k$ denote orthogonal directions to $\spn{u}$. Then a straightforward calculation gives 
\begin{equation}\label{eqriemuyhh}
   {R^\mu}_{j i k}\nor_\mu  =  {C^\mu}_{jik} \nor_\mu    - \frac{2}{n-1} g_{j[k}{R_{i]\mu}}\nor^\mu .
\end{equation}
The  Ricci tensors of $g$ and $\tilg$ are related by
\begin{equation}\label{eqrelriccis}
  R_{\alpha \beta} - \widetilde R_{\alpha \beta} = - \frac{n-1}{\Om} \nabla_\alpha \nabla_\beta \Om - g_{\alpha \beta} \frac{\nabla_\mu \nabla^\mu \Om}{\Om} + g_{\alpha \beta}\frac{n}{\Om^2} \nabla_\mu \Om \nabla^\mu \Om,
 \end{equation} 
 where indices are raised with the contravariant metric of $g$. If $\tilg$ is Einstein and $\Om$ geodesic w.r.t $g$, \eqref{eqrelriccis} gives
 \begin{equation}\label{eqricciconfgeo}
  R_{\alpha \beta}  = - \frac{n-1}{\Om} \nabla_\alpha \nabla_\beta \Om - g_{\alpha \beta} \frac{\nabla_\mu \nabla^\mu \Om}{\Om}.
 \end{equation}
 Hence
 \begin{equation}
  R_{i \mu} \nor^\mu = 
%   - \frac{n-1}{\Om} \nor^\mu \nabla_i \nabla_\mu \Om =  
  -\lambda^{-1/2} \frac{n-1}{\Om}  (\nabla_i \nabla_\mu \Om) \nabla^\mu \Om  =  -\lambda^{-1/2} \frac{n-1}{2 \Om} \nabla_i(\nabla^\mu \Om \nabla_\mu \Om)  = 0
 \end{equation}
 and 
 \begin{equation}\label{eqriemijk}
   {R^\mu}_{j i k}\nor_\mu  =  {C^\mu}_{jik} \nor_\mu.
\end{equation}
 In particular
 \begin{equation}\label{eqriemuyhhgeodeins}
   {R^\mu}_{\alpha \nu \beta}\nor_\mu y^\nu {\h^\alpha}_\lambda {\h^\beta}_\sigma   =  {C^\mu}_{\alpha \nu \beta} \nor_\mu y^\nu {\h^\alpha}_\lambda {\h^\beta}_\sigma .
\end{equation}
  }

  \begin{lemma}\label{lemmacoddS}
   Let $\tilg_{dS}$ be the metric of de Sitter, $\barg = \Om^2 \tilg_{dS}$ a geodesic conformal extension and $\overline{K}_\Om$ the second fundamental form on the leaves $\Sigma_\Om = \{\Om = const.\}$. Then,  the Codazzi equation \eqref{eqfullcodaz} is 
   \begin{equation}
    \overline{\nabla}_k (\overline{K}_\Om)_{ij} - \overline{\nabla}_i (\overline{K}_\Om)_{kj} = 0,
   \end{equation}
  \end{lemma}
  \begin{proof}
   The lemma follows by simply applying the Codazzi equation \eqref{eqfullcodaz} to $\barg$ together with identity \eqref{eqriemijk}, where the Weyl tensor vanishes because $\barg$ is conformally flat.   
  \end{proof}

  \begin{proposition}\label{proplhscod}
   Let $\tilg$ be an Einstein metric admitting a smooth conformally flat $\scri$, $g = \Om^2 \tilg$ a geodesic conformal extension and ${K}_\Om$ the second fundamental form on the leaves $\Sigma_\Om = \{\Om = const.\}$. Then the leading order term of the LHS of the Codazzi equation \eqref{eqfullcodaz} is 
%     \begin{equation}
%     \nabla_k (\partial_\Om Q)_{ij} - \nabla_i (\partial_\Om Q)_{kj} + O( \Om^n),
%    \end{equation}
%  Moreover, its leading order term is  
\begin{equation}
   - \frac{\lambda^{1/2}}{2}(n-1)\Om^{n-1} \lr{ \nabscr_k (\mathring g_{(n)})_{ij} - \nabscr_i (\mathring g_{(n)})_{kj}},
   \end{equation}
   where $\gamma$ is extended off $\scri$ as independent of $\Om$ and $\nabscr$ denotes its Levi-Civita connection.
  \end{proposition}
  \begin{proof} 
  Consider the decomposition $a)$ of Proposition \ref{propconflatdec},  $g = \barg + Q$ with $\barg = -\dif \Om^2/\lambda + \barg_{\Om}$ conformal to de Sitter. Since $\barg^{\alpha\beta} \nabla_\alpha\Om \nabla_\beta \Om = -\lambda$, the conformal factor $\Om$ is  geodesic for both $g$ and $\barg$.  On the other hand, the second fundamental forms $K_\Om$ and $\overline{K}_\Om$, respectively induced by $g$ and $\barg$ on $\Sigma_\Om$, are related by
   \begin{equation}
    K_\Om = \frac{-\lambda^{1/2}}{2} \partial_\Om  g_\Om = \frac{-\lambda^{1/2}}{2} \partial_\Om (\barg_\Om + Q) = \overline{K}_\Om - \frac{\lambda^{1/2}}{2} (n-1) \Om^{n-1} \mathring g_{(n)} + O(\Om^{n+1}),
   \end{equation}
where we have used that by construction $Q = \Om^n \mathring g_{(n)} + O(\Om^{n+2})$. For every tensor $\mathcal{T}_{ij}$, tangent to $\Sigma_\Om$, it follows that its covariant derivatives w.r.t. $\nabla$ and $\overline{\nabla}$ satisfy (we use that the coordinates are Gaussian with respect to $g$)
\begin{equation}
 \nabla_k \mathcal{T}_{ij} =  \overline{\nabla}_k \mathcal{T}_{ij} - {\Q^l}_{ki} \mathcal T_{lj} - {\Q^l}_{kj} \mathcal T_{il}
\end{equation}
where the tangent components of $\Q$, given by \eqref{eqdefS} for $g^{(1)} = g $ and $g^{(2)} = \barg $, satisfy
\begin{equation}
 {\Q^l}_{ki} = \frac{1}{2}g^{lm}\lr{\overline{\nabla}_k g_{im}+\overline{\nabla}_i g_{km} -\overline{\nabla}_m g_{ki}} =\frac{1}{2} g^{lm}\lr{\overline{\nabla}_k Q_{im}+\overline{\nabla}_i Q_{km} -\overline{\nabla}_m Q_{ki}} = O(\Om^n).
\end{equation} 
Thus $\nabla_k \mathcal{T}_{ij} =  \overline{\nabla}_k \mathcal{T}_{ij}  + O(\Om^n)$. In particular, for $K_\Om$ 
\begin{align}
 {\nabla}_k ({K}_\Om)_{ij} & = {\nabla}_k (\overline{K}_\Om)_{ij} - \frac{\lambda^{1/2}}{2}(n-1)\Om^{n-1}\nabla_k(\mathring g_{(n)})_{ij} + O(\Om^{n+1})\\
 & =  \overline{\nabla}_k (\overline{K}_\Om)_{ij} - \frac{\lambda^{1/2}}{2}(n-1)\Om^{n-1}\nabla_k(\mathring g_{(n)})_{ij} + O(\Om^n),
\end{align}
and the LHS of the Codazzi equation \eqref{eqfullcodaz} for $K_\Om$ is
  \begin{align}
    {\nabla}_k ({K}_\Om)_{ij} - {\nabla}_i (K_\Om)_{kj}  = &   
    \overline{\nabla}_k (\overline{K}_\Om)_{ij} - \overline{\nabla}_i (\overline{K}_\Om)_{kj} 
 \\ & -  \frac{\lambda^{1/2}}{2}(n-1) \Om^{n-1}\lr{
    {\nabla}_k (\mathring g_{(n)})_{ij} - {\nabla}_i (\mathring g_{(n)})_{kj}} +  O(\Om^{n}) \\
     = & -  \frac{\lambda^{1/2}}{2}(n-1) \Om^{n-1}\lr{
    {\nabla}_k (\mathring g_{(n)})_{ij} - {\nabla}_i (\mathring g_{(n)})_{kj}} +  O(\Om^{n}),
   \end{align}
   where the second equality is a consequence of Lemma \ref{lemmacoddS}. Now, since $g_\Om = \gamma + O(\Om^2)$, the covariant derivatives ${\nabla}_k (\mathring g_{(n)})_{ij}$ and ${\nabla}_i (\mathring g_{(n)})_{kj}$ are, to lowest order in $\Om$, ${\nabscr}_k (\mathring g_{(n)})_{ij}$ and ${\nabscr}_i (\mathring g_{(n)})_{kj}$.
  \end{proof}

%   \carnote{The free indices in the proof where left greek because this is later contracted with $h$. Latin indices where introduced only when no contractions with $h$ or $y$ appear. Maybe I should keep this convention also in the proof?}

Therefore, for the particular case of Kerr-Schild-de Sitter metrics and the components of the Codazzi equation in \eqref{eqcodaz} we obtain:
\begin{corollary}\label{corolhscod}
 The $\Om^{n-1}$ order term of the LHS of \eqref{eqcodaz} is, up to a non-zero constant,
  \begin{equation}\label{eqLHScod}
    (\mathcal{LHS})_{\lambda \sigma} :=  -\frac{1}{f^n} \Pi_{\lambda \sigma}. 
   \end{equation}
\end{corollary}
\begin{proof} 
From Proposition \ref{proplhscod}, the term of order $\Om^{n-1}$ of \eqref{eqcodaz} only involves derivatives of $\mathring g_{(n)}$. By Theorem \ref{theodatascri}, $\mathring g_{(n)}$ is up to a constant the electric part of the rescaled Weyl tensor, which by Lemma \ref{lemmadecayH}, is given by expression \eqref{eqTTDY2}.  
Hence, substituting $\gamma_{\alpha\beta} = y_\alpha y_\beta + h_{\alpha\beta}$,  the $(n-1)$-th order of the LHS of \eqref{eqcodaz} is (up to a non-zero  constant)
\begin{align}
 \y^\nu (\nabscr_\nu (D_\Y)_{ \beta\alpha}  - \nabscr_\beta (D_\Y)_{\nu \alpha} )  & =
%  -\frac{n}{f^{n+1}} \y^\mu \nabscr_\mu f \lr{\y_\alpha \y_\beta - \frac{\y_\alpha \y_\beta}{n} - \frac{\h_{\alpha\beta}}{n}}+ \frac{1}{f^n}(a_\alpha y_\beta + a_\beta y_\alpha) \\
  - \frac{n}{f^{n+1}} \y^\nu \nabscr_\nu f
  \lr{ \y_\beta \y_\alpha - 
 \frac{\y_\beta \y_\alpha}{n} - \frac{\h_{\beta\alpha}}{n} } 
  + \frac{1}{f^n}( a_\beta y_\alpha+ a_\alpha y_\beta)\\
  & +  \frac{n}{f^{n+1}} \nabscr_\beta f \frac{n-1}{n}y_\alpha - \frac{1}{f^n}\nabscr_\beta y_\alpha.
\end{align}
 Inserting the decomposition \eqref{eqdecnaby} and using the first equation in \eqref{eqckvf}  
\begin{align}
 & \y^\nu  (\nabscr_\nu (D_\Y)_{ \beta\alpha}  - \nabscr_\beta (D_\Y)_{\nu \alpha} )   =
  - \frac{n}{f^{n+1}} \frac{f L}{n-1}
  \lr{ \frac{n-1}{n} \y_\beta\y_\alpha - \frac{\h_{\beta\alpha}}{n} } 
  + \frac{1}{f^n}(a_\beta y_\alpha + a_\alpha y_\beta)\\
  & +  \frac{n}{f^{n+1}} (\frac{f L}{n-1} y_\beta - f a_\beta) \frac{n-1}{n}y_\alpha - \frac{1}{f^n}(y_\beta a_\alpha + \Pi_{\beta\alpha } + \frac{L}{n-1}\h_{ \beta \alpha} + \w_{\beta \alpha}) \\
  & = -\frac{1}{f^n} \lr{ (n-2)a_\beta y_\alpha + \Pi_{\beta \alpha } + \w_{\beta\alpha}}.
\end{align}
 Contracting both indices with $\h$ and symmetrizing yields the following tensor
   \begin{equation}
    (\mathcal{LHS})_{\lambda \sigma} :=  \y^\nu (\nabscr_\nu (D_\Y)_{ \beta\alpha}  - \nabscr_\beta (D_\Y)_{\nu \alpha} ) {\h^\alpha}_{(\lambda} {\h^\beta}_{\sigma)} = -\frac{1}{f^n} \Pi_{\lambda \sigma}. 
   \end{equation}
   \end{proof}
   \bigskip
   In the remainder of this section, we ellaborate the RHS of \eqref{eqcodaz}. Applying identity \eqref{eqriemuyhhgeodeins} it follows
\begin{equation}\label{eqRHScod}
 (\mathcal{RHS})_{\sigma \lambda} := {R^\mu}_{ \alpha \nu \beta}\nor^\nu \y_\mu  {\h^\alpha}_{(\lambda} {\h^\beta}_{\sigma)}  = {C^\nu}_{ \beta \mu \alpha}\nor_\nu \y^\mu  {\h^\alpha}_{(\lambda }{\h^\beta}_{\sigma)}.
\end{equation}
Now we use the algebraic special condition to prove that the $\Om^{n-1}$ order components of the Weyl tensor in \eqref{eqRHScod} are pure trace. Recall the decomposition \eqref{eqprojkuy} of $k$. One can then define $l = 2s^{-1} (u-y)$ such that $l_\alpha k^\alpha = -1$ and complete to a null frame $\{ k,l,m_{(i)} \}$. Then, $\h$ is the projector onto $\spn{m_{(i)}}$. Thus, contracting ${C^\mu}_{ \alpha \nu \beta}$ with $k_\mu k^\nu {\h^\alpha}_{(\lambda} {\h^\beta}_{\sigma)}$ gives by Proposition \ref{propksspec}
\begin{align}
  ~~ & 0 = {C^\mu}_{ \alpha \nu \beta} k_\mu k^\nu {\h^\alpha}_{(\lambda} {\h^\beta}_{\sigma)} \\  ~\Longleftrightarrow~~ & 0 =  \lr{{C^\mu}_{ \alpha \nu \beta} \nor_\mu \nor^\nu + {C^\mu}_{ \alpha \nu \beta} \y_\mu \y^\nu + 2 {C^\mu}_{ (\alpha| \nu |\beta)} \nor_\mu \y^\nu}  {\h^\alpha}_{(\lambda} {\h^\beta}_{\sigma)} \\ 
  ~ \Longleftrightarrow ~~ & 2 {C^\mu}_{ (\alpha| \nu |\beta)} \nor^\nu \y_\mu  {\h^\alpha}_{(\lambda} {\h^\beta}_{\sigma)} = - {C^\mu}_{ \alpha \nu \beta} \nor_\mu \nor^\nu  {\h^\alpha}_{\lambda} {\h^\beta}_{\sigma}- {C^\mu}_{ \alpha \nu \beta} \y_\mu \y^\nu {\h^\alpha}_{\lambda} {\h^\beta}_{\sigma}.  
\end{align}
In addition
\begin{equation}
 g_{\alpha \beta} = -\nor_\alpha \nor_\beta + \y_\alpha \y_\beta + \h_{\alpha \beta},
\end{equation}
 and the traceless property of the Weyl tensor gives
\begin{equation}
 0 = {C^\mu}_{\alpha\mu\beta} = -{C^\mu}_{\alpha \nu \beta} \nor_\mu \nor^\nu + {C^\mu}_{\alpha \nu \beta} \y_\mu \y^\nu + {C^\mu}_{\alpha \nu \beta} {\h^\nu}_\mu ~\Longrightarrow ~ {C^\mu}_{\alpha \nu \beta} \y_\mu \y^\nu  = {C^\mu}_{\alpha \nu \beta} \nor_\mu \nor^\nu - {C^\mu}_{\alpha \nu \beta} {\h^\nu}_\mu.
\end{equation}
Therefore
\begin{equation}\label{eqtermsweyl}
 2 {C^\mu}_{ (\alpha| \nu |\beta)} \nor^\nu \y_\mu  {\h^\alpha}_{(\lambda} {\h^\beta}_{\sigma)} = -2 {C^\mu}_{ \alpha \nu \beta} \nor_\mu \nor^\nu {\h^\alpha}_{\lambda} {\h^\beta}_{\sigma} + {C^\mu}_{ \alpha \nu \beta} {\h^\nu}_\mu {\h^\alpha}_{\lambda} {\h^\beta}_{\sigma}.
\end{equation}
The first term in the RHS of \eqref{eqtermsweyl} only involves the electric part of the Weyl tensor. Using the previous results we next prove that, at order $\Om^{n-1}$, it can only contain trace terms. 
\begin{lemma}\label{lemmaCuu}
Let $\tilg$ be a conformally extendable metric admitting a smooth conformally flat $\scri$. Then, for every geodesic conformal extension $g = \Om^2 \tilg$, the electric part of Weyl tensor w.r.t. the normal vector ${C_\perp}$ has no terms in $\Om^{n-1}$. Moreover, if $\tilg$ is Kerr-Schild-de Sitter, the possible terms of order $\Om^{n-1}$ added by contracting twice with $h$, i.e. $({C_\perp)}_{\alpha \beta}{\h^\alpha}_\lambda {\h^\sigma}_\beta$, are pure trace.
\end{lemma}
\begin{proof}
First consider $g = -\dif \Om^2 + g_\Om$ in normal form w.r.t. a boundary metric $\gamma$. Since $\gamma$ is conformally flat, we can decompose $g_\Om$ as in statement $a)$ of Proposition \ref{propconflatdec}
 \begin{equation}\label{eqauxdecg}
 g_\Om = \barg_\Om + Q
\end{equation}
where $\barg = -\dif \Om^2 + \barg_\Om$ is conformally isometric to de Sitter, $\barg_\Om$ is given by \eqref{eqpoindS} and $Q=O(\Om^n)$ contains no terms of order $\Om^{n+1}$. We now insert this decomposition into formula \eqref{eqCtnocov}, which for simplicity we write using matrix notation as
 \begin{equation}\label{eqcperpmat}
  ({C_\perp}) =  \frac{\lambda}{2}  \lr{\frac{1}{2} \dot g_\Om g^{-1}_\Om \dot g_{\Om}+  \frac{1}{\Om} \dot g_{\Om} - \ddot g_{\Om}}.
\end{equation}
where $\dot~$ stands for derivative in $\Om$ and note, $g^{-1}_\Om$ must decompose as
\begin{equation}
 g^{-1}_\Om = \barg^{-1}_\Om + V
\end{equation}
with $V = O(\Om^n)$, because $g^{-1}_\Om g_\Om$ equals the identity and terms of order $m<n$ in $V$ could not be cancelled out.
We compute the terms in \eqref{eqcperpmat}. Firstly 
\begin{align}
 \dot g_\Om g^{-1} \dot g_{\Om} & = \dot{ \barg}_\Om g^{-1}_\Om \dot{ \barg}_\Om + \dot{ \barg}_\Om g^{-1}_\Om \dot Q + \dot Q g^{-1}_\Om \dot{ \barg}_\Om + \dot Q g^{-1}_\Om \dot Q \\
 & =\dot{\barg}_\Om \barg^{-1}_\Om \dot{ \barg}_\Om + \dot{ \barg}_\Om \barg^{-1}_\Om \dot Q + \dot Q \barg^{-1}_\Om \dot{ \barg}_\Om + \dot Q \barg^{-1}_\Om \dot Q + \dot{ \barg}_\Om V \dot{ \barg}_\Om + \dot{ \barg}_\Om V \dot Q + \dot Q V \dot{ \barg}_\Om + \dot Q V \dot Q,
\end{align}
 and second
\begin{equation}
 \frac{1}{\Om} \dot g_{\Om} - \ddot g_{\Om} = \frac{1}{\Om}  \dot{ \barg}_\Om -  \ddot{ \barg}_\Om + \frac{1}{\Om}  \dot Q -  \ddot Q.
\end{equation}
Adding them and taking into account that 
\begin{equation}
  \frac{\lambda}{2}\lr{ \frac{1}{2} \dot{\barg}_\Om \barg^{-1}_\Om \dot{ \barg}_\Om + \frac{1}{\Om}  \dot{ \barg}_\Om -  \ddot{ \barg}_\Om} = \overline{C}_\perp = 0
\end{equation}
where $(\overline{C}_\perp)$ is the electric part of the Weyl tensor of $\barg$, we are left with
\begin{align}
  \frac{2}{\lambda}({C_\perp}) & = \frac{1}{2}\lr{\dot{ \barg}_\Om \barg^{-1}_\Om \dot Q + \dot Q \barg^{-1}_\Om \dot{ \barg}_\Om + \dot Q \barg^{-1}_\Om \dot Q  + \dot{ \barg}_\Om V \dot{ \barg}_\Om + \dot{ \barg}_\Om V \dot Q + \dot Q V \dot{ \barg}_\Om + \dot Q V \dot Q} + \frac{1}{\Om}  \dot Q -  \ddot Q \\
  & =  \frac{1}{\Om}  \dot Q -  \ddot Q + O(\Om^n).\label{eqCelecdom}
\end{align}
Since $Q$ does not contain terms of order $\Om^{n+1}$, then \eqref{eqCelecdom} does not contain terms of order $\Om^{n-1}$. This proves the first part of the lemma. 

Combining this fact with equation \eqref{eqTTDY2}, we can write the leading order of $C_\perp$ and its tail order terms as
\begin{align}
 C_{\perp} & = \Om^{n-2} \frac{\kappa}{f^{n+2}} \lr{\Y_\alpha \Y_\beta - \frac{f^2}{n} \gamma_{\alpha\beta}} + O(\Om^n),
\end{align}
where $\gamma$ must be understood as the leading order term of $g_{\Om}$, i.e. the extension of $\gamma|_{\scri}$ to the spacetime as a tensor independent  of $\Om$ and similarly with $\Y$. Contracting this expression twice with $\h$ gives
 \begin{equation}
  (C_\perp)_{\alpha\beta}{\h^\alpha}_\mu {\h^\beta}_\nu = -\Om^{n-2} \frac{\kappa}{n} \frac{1}{f^n} \h_{\mu \nu} + O(\Om^n).
 \end{equation}
We cannot exclude that the presence of  $h_{\alpha\beta }$  in this expression introduces terms of order $\Om^{n-1}$, but if present, they are clearly trace terms, as claimed in the Lemma.
\end{proof}

We next look for a similar result for the components of the Weyl tensor ${C^\mu}_{\beta \nu \alpha} {\h^\nu}_\mu {\h^\alpha}_{(\lambda} {\h^\beta}_{\sigma)}$ which arise in \eqref{eqtermsweyl}. From the definition \eqref{eqWeyl} one has
\begin{align}
 {C^\mu}_{\alpha \nu \beta} {\h^\nu}_\mu  {\h^\alpha}_{\lambda} {\h^\beta}_{\sigma} & = {R^\mu}_{\alpha \nu \beta}{\h^\nu}_\mu  {\h^\alpha}_{\lambda} {\h^\beta}_{\sigma} \\ ~~ & + \lr{- \frac{2}{n-1}({\delta^\mu}_{[\nu}R_{\beta]\alpha} - g_{\alpha[\nu}{R^\mu}_{|\beta]}) + \frac{2 R}{n(n-1)}  {\delta^\mu}_{[\nu}g_{\beta]\alpha}} {\h^\nu}_\mu  {\h^\alpha}_{\lambda} {\h^\beta}_{\sigma} \\
%  & =  {R^\mu}_{\alpha \nu \beta}{\h^\nu}_\mu  {\h^\alpha}_{(\lambda} {\h^\beta}_{\sigma)} - R_{\alpha \beta} {\h^\alpha}_\lambda {\h^\beta}_\sigma + \frac{1}{n-1} R_{\alpha \beta} {\h^\alpha}_\lambda {\h^\beta}_\sigma \\
%  & + \frac{1}{n-1} R_{\alpha \beta} {\h^\alpha}_\lambda {\h^\beta}_\sigma 
%   -\frac{{R^\mu}_\nu {\h^\nu}_\mu}{n-1}\h_{\lambda\sigma} + \frac{R}{n(n-1)}\lr{(n-1) \h_{\lambda \sigma} - \h_{\lambda \sigma}} \\
  & = {R^\mu}_{\alpha \nu \beta}{\h^\nu}_\mu  {\h^\alpha}_{\lambda} {\h^\beta}_{\sigma} - \frac{n-3}{n-1} R_{\alpha \beta} {\h^\alpha}_\lambda {\h^\beta}_\sigma+ \lr{ -\frac{{R^\mu}_\nu {\h^\nu}_\mu}{n-1} + \frac{n-2}{n(n-1)}R} \h_{\lambda \sigma}\label{eqRhhh},
\end{align}
% For the components of the Ricci tensor we use the relation \eqref{eqrelriccis} for the Ricci tensor of two conformal metrics $g = \Om^2 \tilg$. Since $\nabla \Om$ is geodesic and $\tilg$ is Einstein $\til{Ric} = -\lambda n  \tilg = n|\nabla \Om|_g^2 g/\Om^2 $, this can be reduced to
%  \begin{equation}\label{eqrelriccis2}
%   R_{\alpha \beta} = - \frac{n-1}{\Om} \nabla_\alpha \nabla_\beta \Om - g_{\alpha \beta} \frac{\nabla_\mu \nabla^\mu \Om}{\Om}.
%  \end{equation}
 which using \eqref{eqricciconfgeo} gives  
 \begin{align} 
 {C^\mu}_{\alpha \nu \beta} {\h^\nu}_\mu  {\h^\alpha}_{\lambda} {\h^\beta}_{\sigma} & = {R^\mu}_{\alpha \nu \beta}{\h^\nu}_\mu  {\h^\alpha}_{\lambda} {\h^\beta}_{\sigma} +(n-3) \frac{\nabla_\alpha \nabla_\beta \Om}{\Om} {\h^\alpha}_\lambda {\h^\beta}_\sigma\\ 
 & + \lr{\frac{n-3}{n-1} \frac{\nabla_\mu \nabla^\mu \Om}{\Om} -\frac{{R^\mu}_\nu {\h^\nu}_\mu}{n-1} + \frac{n-2}{n(n-1)} R} \h_{\lambda \sigma}.\label{eqRhhh2}
\end{align}  
 The term containing ${\h^\alpha}_\lambda {\h^\beta}_\sigma \nabla_\alpha \nabla_\beta \Om $ will be left unaltered as it will cancel out after expanding the rest of terms.
Our next aim is to analyze the components of the Riemann tensor ${R^\mu}_{\alpha \nu \beta}{\h^\nu}_\mu  {\h^\alpha}_{(\lambda} {\h^\alpha}_{\sigma)}$, and relate them to the same components of the Riemann tensor of $\hatg$:
\begin{lemma}\label{lemmariems}
 The Riemann tensors of $g$ and $\hatg$ satisfy 
 \begin{align}
   {\widehat{R}^\mu}_{~~\alpha \nu \beta} {\h^\delta}_\mu {\h^\nu}_\gamma  {\h^\alpha}_{\lambda} {\h^\beta}_{\sigma} & = 
    {R^\mu}_{ \alpha \nu \beta} {\h^\delta}_\mu{\h^\nu}_\gamma  {\h^\alpha}_{\lambda} {\h^\beta}_{\sigma} \\ 
    & - 
    2 \H \h^{\delta \tau } {\h^\nu}_\gamma {\h^\alpha}_\lambda {\h^\beta}_\sigma \lr{\nabla_{[\nu|} k_{[\tau}\nabla_{\alpha]}k_{|\beta]} + \nabla_{[\nu} k_{\beta]} \nabla_{[\alpha} k_{\tau]}}.
  \end{align}
\end{lemma}
\begin{proof}
We apply the formula for the difference of Riemann tensors \eqref{eqdiffriems} with $g^{(1)} = \hatg$ and $g^{(2)} = g$. 
 Setting $g = \hatg + \H k \otimes k$, the tensor $\Q$ reads
 \begin{equation}\label{eqSKS}
  {\Q^\mu}_{\alpha \beta} =  -\frac{1}{2} (\hatg^\sharp)^{\mu \nu}\lr{\nabla_\alpha (\H k_\beta k_\nu) + \nabla_\beta (\H k_\alpha k_\nu) - \nabla_\nu (\H k_\alpha k_\beta )}.
 \end{equation}
 Hence, 
 \begin{align}
  {\Q^\kappa}_{\nu \alpha} {\h^\nu}_\gamma {\h^\alpha}_\lambda = -\frac{1}{2} \H k^\kappa {\h^\nu}_\gamma {\h^\alpha}_\lambda (\nabla_\nu k_\alpha + \nabla_\alpha k_\nu)
 \end{align}
and since (recall that $k$ is null geodesic $k^\kappa \nabla_\kappa k_\beta = k^\kappa \nabla_\beta k_\kappa = 0$)
\begin{equation}
 k^\kappa {\h^\beta}_\sigma {\Q^\mu}_{\beta \kappa} = -\frac{1}{2}k^\kappa {\h^\beta}_\sigma (\hatg^\sharp)^{\mu \tau}(\nabla_\beta (\H k_\kappa k_\tau) + \nabla_\kappa (\H k_\beta k_\tau) - \nabla_\tau (\H k_\kappa k_\beta )) = 0,
\end{equation}
it follows
\begin{equation}
 2 {\Q^\kappa}_{[\nu | \alpha|}{{\Q^\mu}_{\beta] \kappa}} {\h^\nu}_\gamma {\h^\alpha}_\lambda {\h^\beta}_\sigma = 0.
\end{equation}
On the other hand
\begin{align}
 \nabla_\nu {\Q^\mu}_{\alpha \beta}  = & - \frac{1}{2} \nabla_\nu (\hatg^\sharp)^{\mu \tau}
 \lr{\nabla_\alpha (\H k_\beta k_\tau) + \nabla_\beta (\H k_\alpha k_\tau) - \nabla_\tau (\H k_\alpha k_\beta )},   \\
 & - \frac{1}{2}  (\hatg^\sharp)^{\mu \tau}
 \nabla_\nu \lr{\nabla_\alpha (\H k_\beta k_\tau) + \nabla_\beta (\H k_\alpha k_\tau) - \nabla_\tau (\H k_\alpha k_\beta )}.  \label{eq2ndline}
\end{align}
The first three terms in \eqref{eq2ndline} vanish when contracted with ${\h^\alpha}_\lambda {\h^\beta}_\sigma$ because, taking into account \eqref{eqgeodKS} and that $k$ is null geodesic,
\begin{align}
 &  \frac{1}{2} \nabla_\nu (\hatg^\sharp)^{\mu \tau}
 \lr{\nabla_\alpha (\H k_\beta k_\tau) + \nabla_\beta (\H k_\alpha k_\tau) - \nabla_\tau (\H k_\alpha k_\beta )}{\h^\alpha}_\lambda {\h^\beta}_\sigma 
  \\  & =     \frac{1}{2} \nabla_\nu (\H k^\mu k^\tau) k_\tau \lr{\nabla_\alpha k_\beta + \nabla_\beta k_\alpha}{\h^\alpha}_\lambda {\h^\beta}_\sigma   = \frac{1}{2} \H^2 k^\mu (\nabla_\nu k^\tau)k_\tau \lr{\nabla_\alpha k_\beta + \nabla_\beta k_\alpha}{\h^\alpha}_\lambda {\h^\beta}_\sigma  = 0.
\end{align}
We calculate the contraction of the last three terms in \eqref{eq2ndline} with $h$ four times. The expansion of each term gives
\begin{align}
{\h^\delta}_\mu {\h^\nu}_\gamma {\h^\alpha}_\lambda {\h^\beta}_\sigma (\hatg^\sharp)^{\mu \tau}\nabla_\nu \nabla_\alpha (\H k_\beta k_\tau) & = \h^{\delta\tau} {\h^\nu}_\gamma {\h^\alpha}_\lambda {\h^\beta}_\sigma \H \lr{\nabla_\nu k_\tau \nabla_\alpha k_\beta + \nabla_\nu k_\beta \nabla_\alpha k_\tau},\label{eqterm1}\\
{\h^\delta}_\mu {\h^\nu}_\gamma {\h^\alpha}_\lambda {\h^\beta}_\sigma (\hatg^\sharp)^{\mu \tau}\nabla_\nu \nabla_\beta (\H k_\alpha k_\tau) & = \h^{\delta\tau} {\h^\nu}_\gamma {\h^\alpha}_\lambda {\h^\beta}_\sigma \H \lr{\nabla_\nu k_\tau \nabla_\beta k_\alpha + \nabla_\nu k_\alpha \nabla_\beta k_\tau},\label{eqterm2}\\
{\h^\delta}_\mu {\h^\nu}_\gamma {\h^\alpha}_\lambda {\h^\beta}_\sigma (\hatg^\sharp)^{\mu \tau}\nabla_\nu \nabla_\tau (\H k_\alpha k_\beta) & = \h^{\delta\tau} {\h^\nu}_\gamma {\h^\alpha}_\lambda {\h^\beta}_\sigma \H \lr{\nabla_\nu k_\alpha \nabla_\tau k_\beta + \nabla_\nu k_\beta \nabla_\tau k_\alpha}.\label{eqterm3}
\end{align} 
Then, rearraging terms, 
\begin{align} 
 2 {\h^\delta}_\mu {\h^\nu}_\gamma {\h^\alpha}_\lambda {\h^\beta}_\sigma \nabla_{[\nu} {\Q^\mu}_{ \beta] \alpha}  = -2 \H \h^{\delta \tau } {\h^\nu}_\gamma {\h^\alpha}_\lambda {\h^\beta}_\sigma \lr{\nabla_{[\nu|} k_{[\tau}\nabla_{\alpha]}k_{|\beta]} + \nabla_{[\nu} k_{\beta]} \nabla_{[\alpha} k_{\tau]}},
\end{align}
and the Lemma follows from the identity  \eqref{eqdiffriems}.
\end{proof}
\noindent
Specifically for our purposes, Lemma \ref{lemmariems} yields
 \begin{align}
    {R^\mu}_{ \alpha \nu \beta} {\h^\nu}_\mu  {\h^\alpha}_{\lambda} {\h^\beta}_{\sigma} = {\widehat{R}^\mu}_{~~\alpha \nu \beta}  {\h^\nu}_\mu  {\h^\alpha}_{\lambda} {\h^\beta}_{\sigma} + O(\Om^n), \label{eqriemhhhghatg} 
  \end{align}
  so we do not have to take into account the tail order terms.
 To calculate  ${\widehat{R}^\mu}_{ ~~\alpha \nu \beta} {\h^\nu}_\mu  {\h^\alpha}_{\lambda} {\h^\beta}_{\sigma}$, we use the definition of the Weyl tensor \eqref{eqWeyl}, which for $\hatg$ vanishes, and contractions with ${\h}$ give: 
 \begin{equation}\label{eqhatRhhh1}
 {\widehat{R}^\mu}_{~~\alpha \nu \beta}{\h^\nu}_\mu  {\h^\alpha}_{\lambda} {\h^\beta}_{\sigma} =  \frac{n-3}{n-1} \widehat{R}_{\alpha \beta} {\h^\alpha}_\lambda {\h^\beta}_\sigma- \lr{ -\frac{{\widehat{R}^\mu}_{~\nu} {\h^\nu}_\mu}{n-1} + \frac{n-2}{n(n-1)}\widehat{R}} \h_{\lambda \sigma}.
\end{equation} 
We finally relate the term $\widehat{R}_{\alpha \beta} {\h^\alpha}_\lambda {\h^\beta}_\sigma$ with the same components of the Ricci tensor of de Sitter. To do that, we use equation \eqref{eqrelriccis}, substituting $g$ by $\hatg$ and $\tilg$ by $\tilg_{dS}$
\begin{equation}\label{eqrelriccisds}
  \widehat R_{\alpha \beta} - \widetilde R^{dS}_{\alpha \beta} = - \frac{n-1}{\Om} \widehat \nabla_\alpha \widehat \nabla_\beta \Om - \hatg_{\alpha \beta} \frac{\widehat \nabla_\mu \widehat \nabla^\mu \Om}{\Om} + \hatg_{\alpha \beta}\frac{n}{\Om^2} \widehat \nabla_\mu \Om \widehat \nabla^\mu \Om.
 \end{equation}
 We may now use that $\tilg_{dS}$ is Einstein to cancel out terms, but $\Om$ is geodesic w.r.t. to $g$, which means 
\begin{align}
 \hatg_{\alpha \beta}\frac{n}{\Om^2} \widehat \nabla_\mu \Om \widehat \nabla^\mu \Om & = \hatg_{\alpha \beta}\frac{n}{\Om^2} \lr{g^{\mu \nu} + \H k^\mu k^\nu }\widehat \nabla_\mu \Om \widehat \nabla_\nu \Om = -\lambda n  \hatg_{\alpha \beta}\frac{n}{\Om^2} + \hatg_{\alpha \beta}\frac{n}{\Om^2} \H k^\mu k^\nu \nabla_\mu \Om  \nabla_\nu \Om \\
 & = -\lambda n \tilg^{dS}_{\alpha \beta} + \hatg_{\alpha \beta}\frac{n  s^2}{\lambda \Om^2} \H \label{eqnormnabOm}
\end{align}
where we have used that $g^{\mu \nu} \nabla_\mu \Om \nabla_\nu \Om = -\lambda$  and $s = - \lambda^{-1/2} k^\mu \nabla_\mu \Om$. Now, since the de Sitter metric is Einstein, equation \eqref{eqrelriccisds} with \eqref{eqnormnabOm} gives
\begin{align}
 \widehat{R}_{\alpha \beta} {\h^\alpha}_\lambda {\h^\beta}_\sigma & =  - \frac{n-1}{\Om} (\widehat \nabla_\alpha \widehat \nabla_\beta \Om)  {\h^\alpha}_\lambda {\h^\beta}_\sigma + \lr{ - \frac{\widehat \nabla_\mu \widehat \nabla^\mu \Om}{\Om} + \frac{n   s^2}{\lambda\Om^2} \H }\h_{\lambda \sigma}. 
\end{align} 
The tensor $\widehat \nabla_\alpha \widehat \nabla_\beta \Om$ can be related with $\nabla_\alpha \nabla_\beta \Om$ using the difference of connections 
\begin{equation}
 \widehat \nabla_\alpha \widehat \nabla_\beta \Om = \nabla_\alpha \nabla_\beta \Om - {\Q^\mu}_{\alpha \beta}\nabla_\mu \Om
\end{equation}
% with
% \begin{equation}
%   {Q^\mu}_{\alpha \beta} 
%  := \frac{1}{2} \hatg^{\mu \nu}(\nabla_\nu \hatg_{\alpha \beta} - \nabla_\alpha \hatg_{\beta \nu} - \nabla_\beta \hatg_{\alpha \nu}) = \frac{1}{2} \hatg^{\mu \nu}(-\nabla_\nu (\H k_\alpha k_\beta) + \nabla_\alpha (\H k_\beta k_\nu) + \nabla_\beta(\H k_\alpha k_\nu)).
% \end{equation}
with the tensor $\Q$ given in \eqref{eqSKS} and 
\begin{equation}
 {\Q^\mu}_{\alpha \beta} {\h^\alpha}_{\sigma} {\h^\beta}_{\sigma} = \frac{1}{2}\H k_\nu \hatg^{\mu \nu}(\nabla_\alpha k_\beta  + \nabla_\beta k_\alpha ){\h^\alpha}_\lambda {\h^\beta}_\sigma = O(\Om^n).
\end{equation}
Thus
\begin{align}
 \widehat{R}_{\alpha \beta} {\h^\alpha}_\lambda {\h^\beta}_\sigma & =  - \frac{n-1}{\Om} ( \nabla_\alpha  \nabla_\beta \Om)  {\h^\alpha}_\lambda {\h^\beta}_\sigma + \lr{ - \frac{\widehat \nabla_\mu \widehat \nabla^\mu \Om}{\Om} + \frac{n  s^2}{\lambda \Om^2} \H }\h_{\lambda \sigma} + O(\Om^n)
\end{align}
so that from equation \eqref{eqhatRhhh1} it follows 
\begin{align}
 {\widehat{R}^\mu}_{~~\alpha \nu \beta}{\h^\nu}_\mu  {\h^\alpha}_{(\lambda} {\h^\alpha}_{\sigma)} & =   - \frac{n-3}{\Om} ( \nabla_\alpha  \nabla_\beta \Om)  {\h^\alpha}_\lambda {\h^\beta}_\sigma \\ 
 & + \lr{  - \frac{n-3}{n-1} \frac{\widehat \nabla_\mu \widehat \nabla^\mu \Om}{\Om} + \frac{n-3}{n-1}\frac{n s^2}{\lambda \Om^2} \H   +\frac{{\widehat{R}^\mu}_{~\nu} {\h^\nu}_\mu}{n-1} - \frac{n-2}{n(n-1)}\widehat{R}} \h_{\lambda \sigma} + O(\Om^n).\label{eqhatRhhh2}
\end{align}
Combining equation \eqref{eqhatRhhh2} and \eqref{eqriemhhhghatg} and putting the result back in \eqref{eqRhhh2} , we have proven
\begin{align}\label{eqChhh} 
  {C^\mu}_{\alpha \nu \beta} {\h^\nu}_\mu  {\h^\alpha}_{\lambda} {\h^\alpha}_{\sigma} & =  \left( \frac{n-3}{n-1} \frac{\nabla_\mu \nabla^\mu \Om}{\Om} -\frac{{R^\mu}_\nu {\h^\nu}_\mu}{n-1} + \frac{n-2}{n(n-1)} R  \right. \\
   & \left. ~~ - \frac{n-3}{n-1} \frac{\widehat \nabla_\mu \widehat \nabla^\mu \Om}{\Om}   +\frac{{\widehat{R}^\mu}_{~\nu} {\h^\nu}_\mu}{n-1} - \frac{n-2}{n(n-1)}\widehat{R} + \frac{n-3}{n-1}\frac{n  s^2}{\lambda \Om^2} \H \right) \h_{\lambda \sigma} + O(\Om^n)
\end{align}
which is pure trace plus terms of order $n$. Now the following result is straightforward

\begin{proposition}
 Let $\tilg$ be a Kerr-Schild-de Sitter metric and $g = \Om^2 \tilg$ a geodesic conformal extension, with $\gamma = g|_\scri$ conformally flat by definition. Then the electric part of the rescaled Weyl tensor is
 \begin{equation}\label{eqTTDY3}
 D_{\alpha\beta} =  \frac{\kappa}{f^{n+2}} \lr{\Y_\alpha \Y_\beta - \frac{|\Y|^2_\gamma}{n} \gamma_{\alpha\beta}}
\end{equation}
where  $f$ is a function of $\scri$ defined by $(\Om^{-n}\H )|_\scri= \frac{2 f^{-n}}{\lambda n (n-2)}$ and $\Y = f\y$ is a CKVF of $\gamma$. Thus, the Kerr-Schild-de Sitter metrics are in the Kerr-de Sitter-like class.
\end{proposition}
\begin{proof}
By Lemma \ref{lemmadecayH}, we only have to prove that $\Y$ is a CKVF of $\gamma$.
 The RHS of the Codazzi equation \eqref{eqcodaz} is given by  \eqref{eqRHScod}. Combining equation \eqref{eqtermsweyl}, Lemma \ref{lemmaCuu} and equation \eqref{eqChhh}, the non-zero terms of order $\Om^{n-1}$ of ${C^\mu}_{ (\alpha| \nu |\beta)} \nor^\nu \y_\mu$ are pure trace. Thus, the traceless part of \eqref{eqRHScod} is identically zero. By Corollary \ref{corolhscod} this is precisely $0 = \Pi_{\alpha \beta}$. Now the Proposition follows from Lemma \ref{lemmackvf} and Lemma \ref{lemmaeqckvf1}.   
\end{proof}

\begin{remark}
Throughout this section we restricted $\scri$ to the set of point where $\H$ (and $k$) are not zero, because we assumed that $\kappa/f^n =\Am \neq 0$ to write down \eqref{eqTTDY2} (i.e. we assume that $f$ does not diverge). Now, we know that the vector $\Y$  is a CKVF of $\scri$, hence this vector is smooth everywhere. The set of points where it vanishes (i.e. where $f=0$) must be removed from $\scri$ as soon as the constant $\kappa$ in the data $D = \kappa D_{\Y}$ is not zero because the tensor $D_{\Y}$ is certainly singular at points where $\Y$ vanishes.
\end{remark}

\section{Kerr-Schild-de Sitter $\supset$ Kerr-de Sitter-like class}\label{secKSsupKdS}
 
 In this section we will prove the converse inclusion than in Section \ref{secKSsubKdS}, namely, that every spacetime in the Kerr-de Sitter-like class is Kerr-Schild-de Sitter. Our strategy is to explicitly construct every Kerr-de Sitter-like spacetime in Kerr-Schild form. To do that, we take advantage of the property that the data in the Kerr-de Sitter-like class depends solely on  the conformal class of the CKVF $\Y$ (Lemma \ref{lemmpropKdSl}) and a mutiplicative constant. Since the initial value problem is well-posed and each spacetime with data $(\Sigma,\gamma, \kappa D_{\Y} )$ is uniquely determined by $\kappa$ and the conformal class of $\Y$, we can infer all limits of spacetimes from the limits of data, which in turn are consequence of limits of conformal classes of CKVFs. The quotient space of conformal classes of CKVFs was studied in detail in \cite{marspeon21}. The next subsection is devoted to summarizing the results in \cite{marspeon21} that will be needed here.

 \subsection{Conformal classes of CKVFs}
 
As mentioned in Section \ref{secpre} (discussion after Theorem 1.1), conformal classes of CKVFs are of fundamental importance to understand the equivalences between data in the Kerr-de Sitter-like class. In this subsection we review briefly the particular case of conformally flat metrics and describe a method to determine whether two CKVF are in the same conformal class. We also describe several properties of the corresponding quotient space that will be needed later. Details for these results can be found in  \cite{Kdslike,marspeon20,marspeon21}.

It is well-known (e.g. \cite{IntroCFTschBook}) that the conformal diffeomorphisms of the $n$-sphere $\conf(\mathbb{S}^n)$,
% which amount to local conformal transformations of the conformal Euclidean $n$-space, $\confloc(\mathbb{E}^n)$,
can be constructed from the action of the orthochronous Lorentz group $ O^+(1,n+1)$ on the rays of the null cone of $\mink{1,n+1}$. From this, it follows the existence of a map $\phi: O^+(1,n+1) \rightarrow \confloc(\mathbb{E}^n),~ \Lambda \mapsto \phi_\Lambda$, which preserves the group law $\phi_\Lambda \circ \phi_{\Lambda'} = \phi_{\Lambda \cdot \Lambda '}$. At the level of Lie algebras, this implies a correspondence $\phi_\star : \skwend{\mink{1,n+1}} \rightarrow \ckill(\mathbb{E}^n)$, where $\skwend{\mink{1,n+1}}$ are skew-symmetric endomorphisms of $\mink{1,n+1}$, which are well-known to span the Lie algebra of $O^+(1,n+1)$, and $\ckill(\mathbb{E}^n)$ stands for the set of CKVFs of $\mathbb{E}^n$. The explicit form of the map between $\skwend{\mink{1,n+1}}$ and $\ckill(\mathbb{E}^n)$ may vary depending on various choices.
However, it can be proven (see \cite{Kdslike},\cite{marspeon21}) that for a fixed choice of flat metric  $\gamma_E \in [\gamma_E]$ 
and of Cartesian coordinates $\{ X^A\}_{A=1}^n$ of $\gamma_E$, as well as a basis $\{ e_\alpha \}_{\alpha = 0}^{n+1}$ of $\mink{1,n+1}$
with $e_0$ timelike, a map $\phi_{\star}$ can be constructed such that, when applied to an arbitrary skew-symmetric endomorphism $F$ written in the basis 
$e_{\alpha}$ as
\begin{equation}\label{skwmatrix}
 F = \begin{pmatrix}
  0 &  -\nu & -\ar^t + \br^t/2 \\
  -\nu & 0 & - \ar^t - \br^t/2 \\
  -\ar + \br/2 & \ar +\br/2 & -\pmb{\omega}
  \end{pmatrix}
\end{equation} 
gives the conformal Killing vector  $\phi_{\star}(F)$ on $\mathbb{E}^n$  given by
\begin{equation}  \label{eqCKVF}
\Yv  = \big(\br^A + \nu X^A + (\ar_B X^B)X^A - \frac{1}{2}(X_B X^B)\ar^A - {\omega^A}_B X^B \big) \partial_{X^A}.
\end{equation}
We shall also use the notation $F(\Yv)$ to indicate the skew-symmetric endomorphism which is associated to a CKVF $\Yv$ by the map $\phi_{\star}$.
% to which one associates the following skew-symmetric endomorphism, given in an orthonormal basis $\{e_\alpha \}_{\alpha = 0}^{n+1}$ of $\mathbb{M}^{1,n+1}$ with $e_0$ timelike,
% Then, referred to an orthonormal basis $\{e_\alpha \}_{\alpha = 0}^{n+1}$ of $\mathbb{M}^{1,n+1}$, the following map gives a Lie algebra (anti)homomorphism between the set of CKVFs of $[\gamma_E]$ and $\skwend{\mathbb{M}^{1,n+1}}$
The matrix \eqref{skwmatrix} is to be understood as follows. Letting $\alpha$ be column and $\beta$ row in this matrix, the endomorphism $F$ is given by
$F(e_{\beta} ) = F^{\alpha}{}_{\beta} e_{\alpha}$. In the matrix $\eqref{skwmatrix}$, $\ar,\br \in \mathbb{R}^n$ are column vectors with components $\ar^A, \br^A$ respectively, $t$ stands for their transpose (row vector), $\nu \in \mathbb{R}$ and ${\bm \omega}$ is $n\times n$ real skew-symmetric matrix of components $( \delta_{AC} {\omega^C}_A =:){\omega}_{AB} = -\omega_{BA}$.

 The map $\phi$ is a Lie algebra antihomomorphism, i.e. 
$[F(\Yv),F({\Yv}')] = -F([\Yv,{\Yv}'])$. In addition, for every $\Lambda \in O^+(1,n+1)$, it holds $\Lambda \cdot F(\Yv) = F (\phi_{\Lambda \star} (\Yv))$, where ``dot'' denotes the adjoint multiplication of matrices $\Lambda \cdot F(\Yv) = \Lambda F(\Yv) \Lambda^{-1}$. As a consequence, the classification of 
$\skwend{\mathbb{M}^{1,n+1}}$ up to $O^+(1,n+1)$ transformations is equivalent to the classification of CKVFs up to conformal transformations.

An invariant characterization of orbits in $\skwend{\mathbb{M}^{1,n+1}}/O^+(1,n+1)$ is achieved by giving the sufficient number of $O^+(1,n+1)$-invariant quantities. In \cite{Kdslike}, the invariants were chosen to be  the traces of even powers of $F(\Yv)$ and the rank of this endomorphism. 
In \cite{marspeon21}, an alternative characterization was given in terms of the eigenvalues of $F(\Yv)^2$ and the causal character of $\ker F(\Yv)$. Specifically, let $\mathcal{P}_{F^2}(-x)$ be the characteristic polynomial of $-F(\Yv)^2$ and define
\begin{equation}\label{defQF2}
\mathcal{Q}_{F^2}(x) := \lr{\mathcal{P}_{F^2}(-x)}^{1/2} \quad \mbox{($n$ even)}, \quad\quad \mathcal{Q}_{F^2}(x) := \lr{\frac{\mathcal{P}_{F^2}(-x)}{x}}^{1/2} \quad \mbox{($n$ odd)}.
\end{equation}
We introduce also the natural numbers
\begin{equation}\label{eqpq}
 p := \lrbrkt{\frac{n+1}{2}} -1,\quad\quad q := \lrbrkt{\frac{n}{2}},
\end{equation}
which are useful to provide unified expressions, independently on the parity of $n$. Note  that $q=p+1$ whenever $n$ is even, while $q = p$ when $n$ is odd. From the properties of $F^2(\Yv)$, it follows \cite{marspeon21} that 
$\mathcal{Q}_{F^2}(x)$ is a polynomial of degree $q+1$ with $q+1$ real roots counting multiplicity, with at most one of which negative. The classification result of equivalence classes of $F(\Yv)$ is given by the following proposition.
\begin{proposition}[\hspace{-0.025cm}\cite{marspeon21}]\label{defgammamu}
 Let $\mathrm{Roots}\lr{\mathcal{Q}_{F^2}}$ denote the set of roots of $\mathcal{Q}_{F^2}(x)$ repeated as many times as their multiplicity and arranged as follows
 \begin{enumerate}
  \item[a)] If $n$ odd, $\lrbrace{\sigma; \mu_1^2, \cdots, \mu_{p}^2} := \mathrm{Roots}\lr{\mathcal{Q}_{F^2}}$ sorted by $\sigma \geq \mu_1^2\geq \cdots \geq \mu_{p}^2$ if $\ker F(\Yv)$ is timelike, where in this case necessarily $\sigma > 0$. Otherwise $\mu_1^2\geq \cdots \geq \mu_{p}^2\geq 0 \geq \sigma$. 
  \item[b)] If $n$ even, 
  $\lrbrace{-\mu_t^2, \mu_s^2; \mu_1^2, \cdots, \mu_{p}^2} := \mathrm{Roots}\lr{\mathcal{Q}_{F^2}}$ sorted by $\mu_1^2\geq \cdots \geq \mu_{p}^2\geq\mu_s^2 = -\mu_t^2 = 0$ if $\ker F(\Yv)$ is degenerate. Otherwise  $\mu_s^2 \geq \mu_1^2\geq \cdots \geq \mu_{p}^2\geq 0 \geq  -\mu_t^2$, where either $\mu_s^2$ or $\mu_t^2$ are non-zero.  
 \end{enumerate}
 Then the parameters $\lrbrace{\sigma; \mu_1^2, \cdots, \mu_{p}^2}$ for $n$ odd and $\lrbrace{-\mu_t^2, \mu_s^2; \mu_1^2, \cdots, \mu_{p}^2}$  for $n$ even determine uniquely the class of $F(\Yv)$ up to $O^+(1,n+1)$ transformations and hence also the class of $\Yv$ up to conformal transformations.  
\end{proposition}

Note that the parameters introduced in Proposition \ref{defgammamu} are sorted in order to remove permutation ambiguities, so that we can fix a good space of parameters, i.e. such that each point represents a unique element in  $\ckill(\mathbb{E}^n)/\confloc(\mathbb{E}^n)$. Consider, for each $F \in \skwend{\mink{1,n+1}}$, the assignment of parameters $(\sigma,\mu_1^2,\cdots, \mu_p^2)$ for $n$ odd and $(-\mu_t^2, \mu_s^2,\mu_1^2,\cdots, \mu_p^2)$ for $n$ even given in Proposition \ref{defgammamu}. Then the space of parameters for $n$ odd is
 \begin{align}
\mathcal{A}^{(odd)}  := & \{  \lr{\cc, \mu_1^2, \cdots, \mu_{p}^2} \in  {\mathbb{R}}^{p+1} \mid \sigma \geq \mu_1^2 \geq \cdots \geq \mu_{p}^2,\quad\mbox{with $\sigma >0$}\}  \\
 \bigcup   & \{  \lr{ \sigma, \mu_1^2, \cdots, \mu_{p}^2} \in  {\mathbb{R}}^{p+1} \mid  \mu_1^2 \geq \cdots \geq \mu_{p}^2\geq 0 \geq \sigma \}
\end{align}
and $n$ even 
\begin{align}
\mathcal{A}^{(even)}  := & \{  \lr{-\mu_t^2, \mu_s^2, \mu_1^2, \cdots, \mu_{p}^2} \in  \mathbb{R}^{p+2} \mid \mu_s^2 \geq \mu_1^2 \geq \cdots \geq \mu_{p}^2 \geq 0 \geq -\mu_t^2, \quad\mbox{with $\mu_s^2$ or $\mu_t^2 \neq 0$} \}  \\
 \bigcup   &  \{  \lr{-\mu_t^2, \mu_s^2, \mu_1^2, \cdots, \mu_{p}^2} \in  \mathbb{R}^{p+2} \mid \ \mu_1^2 \geq \cdots \geq \mu_{p}^2 \geq 0 = \mu_s^2 = -\mu_t^2\}.
\end{align}
%Modification added: \mathbb{R}^+ is now \mathbb{R} and 
In order to describe limits in the quotient space, also studied in \cite{marspeon21}, we define the following subsets of $\mathcal{A}^{(odd)}$
\begin{align*}
 \mathcal{R}_+^{(n,m)} & := \lrbrace{(\sigma, \mu_1^2, \cdots, \mu_{p}^2) \in \mathcal{A}^{(odd)} \mid \cc \geq \mu_1^2 \geq \cdots >\mu_{p-m+1}^2 = \cdots = \mu_{p}^2 = 0}, \\
 \mathcal{R}_-^{(n,m)} & := \lrbrace{(\sigma, \mu_1^2, \cdots, \mu_{p}^2) \in \mathcal{A}^{(odd)} \mid \cc <0, \mu_1^2 \geq \cdots > \mu_{p-m+1}^2  = \cdots = \mu_{p}^2 = 0},\\
 \mathcal{R}_0^{(n,m)} & := \lrbrace{(\sigma, \mu_1^2, \cdots, \mu_{p}^2) \in \mathcal{A}^{(odd)} \mid \cc = 0, \mu_1^2 \geq \cdots >\mu_{p-m+1}^2 = \cdots = \mu_{p}^2 = 0},
\end{align*}
and of $\mathcal{A}^{(even)}$
\begin{align*}
 \mathcal{R}_+^{(n,m)} & := \lrbrace{(-\mu_t^2, \mu_s^2, \mu_1^2, \cdots, \mu_{p}^2) \in \mathcal{A}^{(even)} \mid  -\mu_t^2 = 0, \mu_s^2 \geq \mu_1^2 \geq \cdots > \mu_{p-m+1}^2 = \cdots = \mu_{p}^2 = 0}, \\
 \mathcal{R}_-^{(n,m)} & := \lrbrace{(-\mu_t^2, \mu_s^2, \mu_1^2, \cdots, \mu_{p}^2) \in \mathcal{A}^{(even)} \mid -\mu_t^2 < 0,  \mu_s^2 \geq \mu_1^2 \geq \cdots > \mu_{p-m+1}^2 = \cdots = \mu_{p}^2 = 0},\\
 \mathcal{R}_0^{(n,m)} & := \lrbrace{(-\mu_t^2, \mu_s^2, \mu_1^2, \cdots, \mu_{p}^2) \in \mathcal{A}^{(even)} \mid -\mu_t^2 = \mu_s^2 = 0, \mu_1^2 \geq \cdots > \mu_{p-m+1}^2 = \cdots = \mu_{p}^2 = 0}.
\end{align*}
The notation $\mathcal{R}_\epsilon^{(n,m)}$ refers to the dimension of the space, $n+2$, the number of last-vanishing parameters $\{\mu_i^2\}$, $m$, and the causal character of $\ker F$, $\epsilon \in \{0, \pm\}$: $0$ if degenerate, $+$ if timelike and $-$ if spacelike or zero. We note that $\epsilon$ is also given by the sign of $\sigma$ in the odd case and closely related to the sign structure of the first two entries of the point $s \in \mathcal{A}^{(even)}$ when $n$ is even.

 The space of skew-symmetric endomorphisms $\skwend{\mink{1,n+1}}$ (being a finite dimensional vector space) carries a canonical topology (see e.g. 
\cite{conwayfunctanals}). The quotient space inherits a natural topology, called ``quotient topology" which is the finest one that makes the proyection a continuous map. In this topology it is sufficient for a sequence of points $s_i$  to have a limit $s$ that  there is a sequence of endomorphisms $F_i$ converging to $F$ with $F_i$ belonging to the class $s_i$ and $F$ belonging to the class $s$. 

By constructing explicit sequences in the total space associated to sequences in the quotient, the following structure of limits arises  \cite{marspeon21}

\begin{proposition}\label{proplimts}
For $n$ odd, $\mathcal{R}_+^{(n,0)}$ and $\mathcal{R}_-^{(n,0)}$ are open in the quotient topology. Moreover there exists sequences in $\mathcal{R}_-^{(n,0)}$ taking limit at every point $ \mathcal{A}^{(odd)}\backslash  \mathcal{R}^{(n,0)}_+$. 

For $n$ even, $\mathcal{R}_-^{(n,0)}$ is open in the quotient topology. Moreover there exists sequences in $\mathcal{R}_-^{(n,0)}$  taking limit at every point $ \mathcal{A}^{(even)}$ (i.e. $\mathcal{R}_-^{(n,0)}$ is dense in the quotient topology). 
 \end{proposition} 
\begin{remark}
In \cite{marspeon21} it is not explicitly proven that $\mathcal{R}^{(n,0)}_+$ and $\mathcal{R}^{(n,0)}_-$ are 
open for $n$ odd and that $\mathcal{R}_-^{(n,0)}$ is open for $n$ even. We provide an argument:

\bigskip

Let $F \in \skwend{\mink{1,n+1}}$, $[F] \in \skwend{\mink{1,n+1}}/O^+(1,n+1)$ its class in the quotient and $\pi$ the canonical projection map $\pi: F \mapsto [F]$. The independent term of the characteristic polynomial is an invariant of the class $[F]$. Let $c_0$ be the function that maps $F$ into the independent term of its characteristic polynomial $c_0(F)$. This map is clearly continuous. Let also $[c_0]$ be the induced map in the quotient, i.e. the map satisfying $c_0 = [c_0] \circ \pi$. Then $[c_0]$ is also continuous (e.g. \cite{willardtopology}). Moreover \cite{marspeon21}, for $n$ even,  $c_0(F) = -\mu_t^2 \mu_s^2 \mu_1^2 \cdots \mu_p^2$ and, for $n$ odd,  $c_0(F) = \sigma \mu_1^2 \cdots \mu_p^2$. Thus when $n$ is odd $\mathcal{R}_+^{(n,0)}$ and $\mathcal{R}_-^{(n,0)}$ are open in $\skwend{\mink{1,n+1}}/O^+(1,n+1)$ as they are the preimage by $[c_0]$ of the open intervals $(0,\infty)$ and  $(-\infty,0)$ respectively. 
When $n$ is even $\mathcal{R}^{(n,0)}_+$ is also open because it is the preimage of the open interval $(0, \infty)$.

\end{remark}

 \subsection{Kerr-de Sitter and its limits at $\scri$}\label{seckdsandlims}

 The explicit form of the  metrics in the full Kerr-de Sitter-like class will be obtained via either limits or analytic extensions of the Kerr-de Sitter family of metrics in all dimensions as presented in \cite{Gibbons2005}. However, we introduce modifications in the coordinates of \cite{Gibbons2005} which make our analysis more direct. Namely, as the limits will be inferred from its data at $\scri$, it is convenient to give the metrics in coordinates such that, in the conformally extended space, the conformal factor vanishes at a finite value of the coordinates. We will also absorb some constants depending on the rotation parameters into the coordinates. This will allow us to perform several limits at once. Moreover, we give the metric already in Kerr-Schild form \eqref{eqKSmetrics}. 
 This will be useful to show that the limits also belong to the Kerr-Schild-de Sitter class.

 Consider an $(n+1)$-dimensional Lorentzian manifold and let $\{\m_i \}_{i=1}^{p+1}$ be a set of functions satisfying
 \begin{equation}\label{DEFmu}
  \sum\limits_{i=1}^{p+1}  \m_i^2 = 1,
\end{equation}
which define $p$ independent quantities (by default we choose the $p$ first ones $\{ \m_i\}_{i=1}^p$) which we will use as coordinates). From the definitions of $p$ and $q$ \eqref{eqpq} it follows that $n = p+q +1$ irrespectively of whether $n$ is even or odd. So, we must supplement the $p$ coordinates $\alpha_i$ with $q+2$ additional ones. The full set of coordinates will be denoted by $\{\rho, t, \{\m_i\}_{i=1}^{p+1}, \{\phi_i\}_{i=1}^q\}$ and it is subject to the constraint \eqref{DEFmu}.  
The $\m_i$s and $\phi_i$s are related to polar and azimuthal angles of the sphere respectively and they take values in $0 \leq \m_i \leq 1$ and $0 \leq \phi_i < 2 \pi$ for $i = 1, \cdots, q$ and (only when $n$ odd) $-1 \leq \m_{p+1} \leq 1$.  The remaining $\rho$ and $t$  lie in $0 \leq \rho < \lambda^{1/2}$ and $t \in \mathbb{R}$. The domain of definition of $\rho$ (which is the inverse of a ``radial'' coordinate) can be extended (across the Killing horizon) to $\rho > \lambda^{1/2}$, but this is unnecessary in this work since we are interested in regions near $\rho = 0$. In addition, we also introduce $q$ real rotation parameters $\{ a_i\}_{i=1}^q$, each one associated to a $\phi_i$. For notational reasons, it is useful to define a trivial parameter
$a_{p+1} =0$ in the case of $n$ odd.

Next, we give the explicit expressions of the Kerr-de Sitter family of metrics. In order for the reader to compare with the original form given in \cite{Gibbons2005}, our notation corresponds to $r := \rho^{-1}$ and we have renamed the $\mu_i$ in \cite{Gibbons2005} as $\alpha_i$.

In terms of the functions
\begin{equation}\label{eqdefWXi}
 W := \sum\limits_{i=1}^{p+1} \frac{\m_i^2}{1 + \lambda a_i^2}\quad \quad  \Xi := \sum\limits_{i=1}^{p+1} \frac{\m_i^2}{1 + \rho^2 a_i^2},  
\end{equation}
and 
\begin{equation}\label{eqPi}
 \Pi :=  \prod\limits_{j=1}^q (1 + \rho^2 a_j^2),
\end{equation}
the Kerr-de Sitter family of metrics \cite{Gibbons2005} is the Kerr-Schild metric
\begin{equation}\label{eqtilH}
 \tilg = \tilg_{dS} + \til\H \til k \otimes \til k,\quad\quad  \til \H = \frac{2M \rho^{n-2}}{ \Pi \Xi}\quad M\in \mathbb{R} 
\end{equation}
with
\begin{align}
  \tilg_{dS} & :=  - W \frac{(\rho^2 - \lambda )}{\rho^2} \dif t^2 + \frac{\Xi}{\rho^2 - \lambda } \frac{\dif \rho^2}{\rho^2} + \delta_{p,q} \frac{\dif \m_{p+1}^2}{\rho^2} + \sum_{i=1}^q \frac{1 + \rho^2 a_i^2}{\rho^2}\frac{ \lr{\dif \m_i^2 + \m_i^2 \dif \phi_i^2} }{1 + \lambda a_i^2} 
  \\
 & 
  + \frac{\lambda}{W \rho^2 (\rho^2- \lambda )}  \lr{\sum_{i= 1}^{p+1} \frac{\lr{1 + \rho^2 a_i^2} \m_i \dif \m_i}{1 + \lambda a_i^2}}^2  , \\
 \til k & :=  W \dif t - \frac{\Xi}{\rho^2 - \lambda } \dif \rho - \sum\limits_{i=1}^q \frac{a_i \m_i^2}{1 + \lambda a_i^2} \dif \phi_i.
\end{align}
The term $\delta_{p,q}$ only appears when $q = p$, i.e. when $n$ is odd. In the case of even $n$, all terms multiplying $\delta_{p,q}$ simply go away. Hence, the above equations  provide unified expressions for $(n+1)$-dimensional Kerr-de Sitter family of metrics for any parity of $n$.  Note that the introduction of the spurious rotation parameter $a_{p+1} \equiv 0$ is also necessary to achieve single form for the $n$ odd and $n$ even cases.

The above expressions simplify using the coordinates 
\begin{equation}\label{eqconsthatalpha}
 \hat \alpha_i :=  \frac{\alpha_i}{(1+\lambda a_i^2)^{1/2}} \quad\Longrightarrow \quad \sum_{i=1}^{p+1}  \alpha_i^2 = \sum_{i=1}^{p+1} (1 + \lambda a_i^2) \hat \alpha_i^2 = 1
\end{equation}
so that 
\begin{equation}\label{eqWXihats}
 W = \sum\limits_{i=1}^{p+1} \hat \m_i^2,\quad \quad  \Xi = \sum\limits_{i=1}^{p+1} \frac{1+\lambda a_i^2}{1 + \rho^2 a_i^2}\hat \m_i^2,
\end{equation}
and
\begin{align}
   \tilg_{dS}  & =  - W \frac{(\rho^2 - \lambda )}{\rho^2} \dif t^2 + \frac{\Xi}{\rho^2 - \lambda } \frac{\dif \rho^2}{\rho^2} + \delta_{p,q} \frac{\dif \hat\m_{p+1}^2}{\rho^2} + \sum_{i=1}^q \frac{1 + \rho^2 a_i^2}{\rho^2} \lr{\dif \hat\m_i^2 + \hat\m_i^2 \dif \phi_i^2}    \\ & ~~ + \frac{ (\rho^2- \lambda )}{ \lambda W \rho^2 }  \frac{\dif W^2}{4} , \label{eqgdSodd}\\
 \til k & =  W \dif t - \frac{\Xi}{\rho^2 - \lambda } \dif \rho - \sum\limits_{i=1}^q {a_i \hat\m_i^2} \dif \phi_i, \label{eqkodd}
\end{align}
where for \eqref{eqgdSodd} we have used the differential of \eqref{eqconsthatalpha}
\begin{align}
  & ~~ \sum_{i=1}^{p+1} (1 + \lambda a_i^2) \hat \alpha_i \dif\hat \alpha_i = 0 \\ \Longrightarrow & ~~ \sum_{i=1}^{p+1}  \lambda a_i^2 \hat \alpha_i \dif\hat \alpha_i = 
  - \sum_{i=1}^{p+1} \hat \alpha_i \dif\hat \alpha_i = -\frac{\dif W}{2}\\
  \Longrightarrow & ~~ \lr{\sum_{i= 1}^{p+1} \frac{\lr{1 + \rho^2 a_i^2} \m_i \dif \m_i}{1 + \lambda a_i^2}}^2 = \lr{\sum_{i= 1}^{p+1} {\lr{1 + \rho^2 a_i^2} \hat \m_i \dif \hat \m_i}}^2 = \lr{\frac{\rho^2-\lambda}{\lambda}}^2 \frac{\dif W^2}{4}.
\end{align}

The initial data of the Kerr-de Sitter family were calculated in \cite{marspeondata21} to be of the form $(\Sigma,\gamma,\kappa D_{\Y})$. so by definition (cf. Definition \ref{defKdSlike}) they belong to the (larger) class of Kerr-de Sitter-like initial data. The parameter $\kappa$ is given in terms of the mass parameter $M$ by 
\begin{equation}\label{eqkappa}
 \kappa := -n(n-2)M \lambda^ {-\frac{n}{2}}\in \mathbb{R}.
\end{equation}
The conformally flat boundary metric takes the explicit form
\begin{equation}
 \gamma = \tilg_{dS}\mid_\scri  = \lambda W  \dif t^2 + \delta_{p,q} {\dif \m_{p+1}^2} + \sum_{i=1}^q  \lr{\dif \hat \m_i^2 + \hat \m_i^2 \dif \phi_i^2} 
  - \frac{1}{W }  \frac{\dif W^2}{4} , \label{eqgamma}
\end{equation}
and the conformal Killing vector $\Yv$ is 
\begin{align}\label{eqYflat}
 \Yv =\frac{1}{\lambda} \partial_t - \sum^{q}_{i=1} a_i \partial_{\phi_i}. 
\end{align}
The conformal class of $\Y$ was determined in \cite{marspeon21} by finding explicitly a set of Cartesian coordinates for a flat metric $\gamma_E$ in the conformal class of the metric \eqref{eqgamma}. With the notation of Proposition \ref{defgammamu}, and after a suitable reordering of the rotational parameters $\{ a_i\}$, this conformal class is determined by the parameters $\{ \sigma = -\lambda^{-1}, \mu_i^2 = a_i^2\}$ for $n$ odd and $\{ -\mu_t^2 = -\lambda^{-1},~\mu_s^2 = a_1^2 , \mu_i^2 = a_{i+1}^2 \}$ for $n$ even. Observe that $\lambda$ is one of the parameters which determines the conformal class of $\Y$. This is a priori fixed by the Einstein equations, so it is not a freely specifiable parameter of the metric. However, under scalings of $\Y$,   $\sigma$ is also scaled with the same factor. From the structure of $D_\Y$ in \eqref{eqTTDY2}, we have the freedom of scaling $\Y$ and leave the data $\kappa D_\Y$ unaltered if we absorb the inverse (squared) scaling factor in $\con$, which is essentially the mass parameter of Kerr-de Sitter, therefore freely specifiable. In this way, we may cover the full domain defining the family $\mathcal{R}_-^{(n,0)}$.

Recall that from Lemma \ref{lemmpropKdSl}, each metric in the Kerr-de Sitter-like class is determined by the parameter $\kappa$ and
the conformal class of $\Y$. Thus, for a fixed value of $\kappa$, one can associate exactly one metric in the Kerr-de Sitter-like class to each point in $
\ckill(\mathbb{E}^n)/\confloc(\mathbb{E}^n)$. Moreover, the limits of regions in $
\ckill(\mathbb{E}^n)/\confloc(\mathbb{E}^n)$, must induce limits of data $(\Sigma, \gamma, \kappa D_\Y)$ which in turn, from the well-posedness of the Cauchy problem, also induce limit of spacetimes corresponding to such data. In this way, we can endow the space of metrics in the Kerr-de Sitter-like class with the topology of $
\ckill(\mathbb{E}^n)/\confloc(\mathbb{E}^n)$.
Now, the following result is immediate after Proposition \ref{proplimts} 
\begin{proposition}\label{propKdsopen}
 The conformal class of the Kerr-de Sitter family  with $m$ vanishing rotation parameters belongs to the region $\mathcal{R}_-^{(n,m)}$ with  $\sigma := -\lambda^{-1}$ and $\mu_i^2 := a_i^2$ for $n$ odd and $-\mu_t^2 := -\lambda^{-1},~\mu_s^2 := a_1^2$ and $\mu_i^2 := a_{i+1}^2$ for $n$ even. Thus, the Kerr-de Sitter family of metrics with all non-zero rotation parameters covers the whole $\mathcal{R}_-^{(n,0)}$. For $n$ even, Kerr-de Sitter family data and its limits cover all data in the Kerr-de Sitter-like class. 
\end{proposition}
In the rest of the paper, we will construct all spacetime metrics in the Kerr-de Sitter-like class taking advantage of the topological structure given in Proposition \ref{proplimts}. 
The region $\mathcal{R}^{(n,0)}_+$ cannot be obtained as a limit from points in $\mathcal{R}^{(n,0)}_-$. Thus, the metrics corresponding to such data cannot be obtained as a limit of the Kerr-de Sitter family.
This family will be obtained by analytic extension of Kerr-de Sitter.

The spacetime limits will be inferred from limits of data as follows. Start with data corresponding to Kerr-de Sitter $(\Sigma,\gamma,\kappa D_\Y)$ in $\mathcal{R}^{(n,m)}_-$, and consider the uniparametric set of equivalent data $(\Sigma,\gamma_\param := \param^{-2} \gamma, \param^{n-2}\kappa D_{\Y})$ for a constant parameter $\param \in \mathbb{R}$. Scaling the following quantities as
\begin{equation}
 M_\param := M \param^{n},\quad\quad\Yv_\param := \param\Yv\quad\quad \Yf_\param:= \gamma_\param(\Y_\param,\cdot) = \param^{-1} \Yf,
\end{equation}
where recall $\Yf = \gamma(\Yv,\cdot)$, we have
\begin{equation}
 \param^{n-2}\kappa D_{\Y} =- \lambda^{\frac{2-n}{2}} \frac{n(n-2)}{|\Y|_\gamma^{n+2}} M \param^{n-2} \lr{\Yf \otimes \Yf - \frac{|\Y|_\gamma}{n} \gamma} = 
 -\lambda^{\frac{2-n}{2}} \frac{n(n-2)}{|\Y_\param|_{\gamma_\param}^{n+2}} M_\param  
 \lr{\Yf_\param \otimes \Yf_\param - \frac{|\Y_\param|^2_{\gamma_\param}}{n} \gamma_\param}.
\end{equation}
Thus, we obtain the uniparametric family of data $(\Sigma, \gamma_\param, \kappa_\param D_{\Y_\param})$, where  $\kappa_\param$ is given by \eqref{eqkappa} with mass parameter $M_\param$. 
As we shall describe, after a suitable rescaling of the coordinates and the rotation parameters, the data
$(\Sigma, \gamma_\param, \kappa_\param D_{\Y_\param})$ admits regular limits as $\param \rightarrow 0$, which are no longer equivalent to the original family, but are still in the Kerr-de Sitter-like class.

By Lemma \ref{lemmpropKdSl}, the limit data are uniquely determined by the limit mass $M' := \lim_{\param \rightarrow 0} M_\param$ and the conformal class of $\limY :=\lim_{\param \rightarrow 0} \Y_\param$.
In all cases, the scaling of the rotation parameters will be of the form
$a_i = \param^{-1} b_i$, where we still allow $b_i$ to smoothly depend on $\param$ (so they may vanish in the limit $\param \rightarrow 0$). For the CKVF itself, in the following subsections we distinguish the limits of the vector field 
\begin{equation}\label{eqxibeta}
 \Yv_\param = {\param} \lr{\frac{1}{\lambda} \partial_t - \sum^{q}_{i=1} \param^{-1} b_i \partial_{\phi_i} }
\end{equation}
 as $\param \rightarrow 0$ into two types, depending on whether or not the parameter $\param$ is absorbed in the $t$ coordinate by means of the change $t = \param t'$.  The limit performed with the coordinate $t'$ will be proven to correspond to the region $\mathcal{R}_0^{(n,m)}$, where $m$ is given by the number of vanishing $b_i$. 
  The limits with the $t$ coordinate unchanged will only be calculated in the $n$ even case and will be proven to lie in the region $\mathcal{R}_+^{(n,m)}$, where $m$ is given by the number of vanishing $b_i$. The reason why we calculate them only for $n$ even is because only in this case we may attain every point in every region $\mathcal{R}_+^{(n,m)}$ from $\mathcal{R}_-^{(n,0)}$ (cf. Proposition \ref{proplimts}). For the $n$ odd case we need to perform an analytic extension to obtain the spacetimes with data in $\mathcal{R}_+^{(n,m)}$.
%  When the $t$ coordinate is left unchanged, the limit ( which will only be calculated in the $n$ even case) will be proven to lie in the region $\mathcal{R}_+^{(n,m)}$
% %  . 
 For any limit data at $\scri$, there is one corresponding spacetime, which from the well-possedness of the Cauchy problem, must be a limit of Kerr-de Sitter. In general, these limit spacetimes are obtained with the same changes than those performed at $\scri$ plus the redefinition $\rho ' = \param
\rho$, as we shall also explicitly demonstrate.

In all the situations, the term $\tilg_{dS}$ takes a well-defined limit independently of the term $\H k \otimes k$. Morever, we will show that, in all cases, $\tilg_{dS}$ and its derivatives up to second order depend continuously on $\param$. Consequently, the Riemann tensor of the limit metric $\tilg_{dS}' = \lim_{\param \rightarrow 0} \tilg_{dS}$ is the limit of the Riemann tensor of $\tilg$, i.e.
 \begin{align}\label{eqriemconstc}
 R'_{\alpha\beta\mu\nu} = \lim_{\param \rightarrow 0}  R_{\alpha\beta\mu\nu} & = \lambda \lim_{\param\rightarrow 0} \lr{ (\tilg_{dS})_{\alpha\mu}(\tilg_{dS})_{\beta\nu} - (\tilg_{dS})_{\alpha\nu} (\tilg_{dS})_{\beta\mu}}
 \\
 & =\lambda \lr{(\tilg'_{dS})_{\alpha\mu}(\tilg'_{dS})_{\beta\nu} - (\tilg'_{dS})_{\alpha\nu} (\tilg'_{dS})_{\beta\mu}}.
 \end{align}
Thus the background limit metric is still Einstein of constant curvature, therefore locally isometric to de Sitter.

As already mentioned, in the $n$ even case all spacetimes in the Kerr-de Sitter-like class are limits of the Kerr-de Sitter family. In the $n$ odd case,  the spacetimes corresponding to the regions $\mathcal{R}_+^{(n,m)}$ will be constructed by analytic continuation, and the rest of them as limits of Kerr-de Sitter. For given data, the corresponding spacetimes  will be assigned to a family depending on the region $\mathcal{R}_\epsilon^{(n,m)}$ to which the defining CKVF at $\scri$ belongs. In analogy with the $n=3$ case \cite{Kdslike}, these families will be called  {\it generalized $ \{ a_i  \rightarrow \infty\}$-limit Kerr-de Sitter} if $\Yv$ lies in $\mathcal{R}_0^ {(n,m)}$ (extending the definition \cite{limitkds}), or {\it generalized Wick-rotated Kerr-de Sitter} if $\Yv$ lies in $\mathcal{R}_+^ {(n,m+1)}$ (also by analogy with \cite{Kdslike}).

 \subsection{Limits $n$-even}\label{secnevenlims}

 We start by determining all limits of Kerr-de Sitter family in the $n$ even case.  In principle the limits can be performed in multiple ways. However, by the classification of conformal classes of CKVF described above it suffices to exhibit one limit for each case. 
 To obtain the spacetimes whose CKVF class at $\scri$ lies in $\mathcal{R}_+^{(n,m)}$, we will assume that the starting family has all its rotation parameters different from zero, i.e. that it belongs to the region $\mathcal{R}^{(n,0)}_-$. Similarly, to obtain those whose CKVF class lies in $\mathcal{R}^{(n,m)}_0$ we shall start from Kerr-de Sitter with exactly one rotation parameter equal to zero, i.e. whose CKVF is in $\mathcal{R}^{(n,1)}_-$. Obviously, all spacetimes in $\mathcal{R}^{(n,m)}_-$ are simply obtained by setting $m$ rotation parameters $a_i$ to zero, so there is no need to explicitly calculate any limit.

 \subsubsection{Generalized Wick-rotated}\label{subsecWReven}

 In this subsection we shall not absorb $\param$ in the coordinate $t$.  As mentioned in subsection \ref{seckdsandlims}, we will obtain in this way all spacetimes whose corresponding CKVF at $\scri$ lies in $\mathbb{R}_+^{(n,m)}$. We will call these {\it Wick-rotated-Kerr-de Sitter} family of spacetimes because in the $n$ odd case (cf. subsection \ref{secwicknodd}) they will actually be obtained by a Wick-rotation of Kerr-de Sitter.

 We start with a metric in the Kerr-de Sitter family, with every rotation parameter being non-zero and apply the redefinitions
  \begin{equation}
  \rho = \param \rho', \quad\quad \hat\m_i = \param \mlim_i \quad\quad a_i = \param^{-1} b_i, \quad \quad M = M' \param^{n}. 
 \end{equation}
 Observe that if any of the rotation parameters were zero, say $a_i=0$, then the scaling of $\hat{\alpha}_i = \param \mlim_i$ would not be allowed because \eqref{eqconsthatalpha} would imply that ${\mlim}_i$ is divergent in the limit $\param \rightarrow 0$. 
 The parameters $b_i$ are still allowed to depend smoothly\footnote{Sufficient differentiability is necessary in order to make sure that the background metric is de Sitter in the limit. W.l.o.g. we can assumme smooth dependence on $\param$ as we only want to allow vanishing values in the limit.} on $\param$, so that their limit at $\param$ may take the value zero. For notational simplicity we shall not include the dependence on $\param$. In particular, the limit at $\param \rightarrow 0$ will still be called $b_i$. The context will make clear the  intended meaning.

 In the limit $\param \rightarrow 0$, by \eqref{eqconsthatalpha} the coordinates $\{ \mlim_i \}$ satisfy
 \begin{equation}
 \sum_{i=1}^{p+1} \lambda b_i^2  \mlim^2_i = 1,
\end{equation} 
thus, at least one $b_i$ must be non-zero. Note that if all were zero, the limit vector field $\Y'  = \lim_{\param \rightarrow 0} \Y_\param$ would be identically zero, and we would at best fall outside the Kerr-de Sitter-like class.

 The function W goes to zero as $\param^2$ while $\Xi$ and $\Pi$ take finite and smooth limits (cf. \eqref{eqPi} and \eqref{eqWXihats}). We therefore introduce the following limit quantities
 \begin{equation}   
  W' := \lim_{\param \rightarrow 0}   \param^{-2} W = \sum\limits_{i=1}^{p+1} {{\mlim}_i^2},\quad \quad  \Xi' := \lim_{\param \rightarrow 0} \Xi = \sum\limits_{i=1}^{p+1} \frac{\lambda b_i^2\mlim^2_i}{1 + \rho'^2 b_i^2},\quad\quad  \Pi' = \lim_{\param \rightarrow 0} \Pi = \prod_{j=1}^q (1 + \rho'^2 b_j^2). 
 \end{equation}
On the other hand, by \eqref{eqkodd}, the terms of $\til k$ in $\dif \rho$ and $\dif \phi_i$ tend to zero with $\param$, while the term in $\dif t$ goes with $\param^{2}$. Hence we set
 \begin{equation}
  \til k' := \lim_{\param \rightarrow 0} \param^{-1} \til k =   \frac{\Xi'}{ \lambda } \dif \rho' - \sum\limits_{i=1}^q b_i {\mlim^2_i} \dif \phi_i,
 \end{equation}
and the redefinition of mass $M' = \param^{n} M$ absorbs the zero of $\til k \otimes \til k$ and that of $\rho^{n-2} = \param^{n-2} \rho'$ in $\til \H ~\tilde{k} \otimes \tilde{ k}$ (cf. \eqref{eqtilH}). Thus, the limit metric has the Kerr-Schild form  
\begin{equation}\label{eqWRKSform}
  \tilg' = \tilg'_{dS} + \til\H' \til k' \otimes \til k',\quad\quad \til\H' = \frac{2M' \rho'^{n-2}}{ \Pi' \Xi'}, \quad M'\in \mathbb{R}  
\end{equation}
with
\begin{align} \label{eqgdslimWR}
  \tilg'_{dS} & =   \frac{\lambda W'}{\rho'^2} \dif t^2 - \frac{\Xi'}{ \lambda } \frac{\dif \rho'^2}{\rho'^2}  + \sum_{i=1}^q \frac{1 + \rho'^2 b_i^2}{\rho'^2}{ \lr{\dif  \mlim^2_i +  \mlim^2_i \dif \phi_i^2} }
  -  \frac{ 1}{  W' \rho'^2 }  \frac{\dif W'^2}{4}.
\end{align}  
 One can easily check that the original de Sitter metric $\tilg_{dS}$ in \eqref{eqgdSodd}, written in primed coordinates is ${C}^2$ in $\param$. Hence, by the above argument $\tilg'_{dS}$ (cf. \eqref{eqriemconstc}), the limit metric \eqref{eqgdslimWR}, is (locally) isometric to de Sitter.

  Consider the conformal extension $g' = \rho'^2 \tilg'$. The boundary metric induced by $g'$ coincides with the one induced by $\tilg'_{dS}$, which is 
 \begin{equation}
  \gamma' = \rho'^2 \tilg_{dS}'|_\scri =\lambda W' \dif t^2  + \sum_{i=1}^q { \lr{\dif  \mlim^2_i +  \mlim^2_i \dif \phi_i^2} }{ } 
  - \frac{1}{W'  }  \frac{\dif W'^2}{4},
 \end{equation}
and coincides with the limit $\lim_{\param \rightarrow 0} \param^{-2} \gamma$ of $\gamma$ given by  \eqref{eqgamma} using the coordinates $\{ \mlim_i \}$. As $\tilg_{dS}$ is locally isometric to de Sitter, $\gamma'$ must be locally conformally flat. 

To calculate the electric part of the rescaled Weyl tensor, we use formula \eqref{eqexpweyl}, after which it follows
\begin{align}
 D = \rho'^{n-2} C_\perp |_\scri =  -\lambda n(n-2) M' \lr{\Y' \otimes \Y' - \frac{|\Y'|^2_{\gamma'}}{n} \gamma'},
\end{align}
where $\Y'$ is the projection of $\til k'$ onto $\scri$
\begin{align}
 \Yf' =  -\sum\limits_{i=1}^q {b_i   \mlim^2_i} \dif \phi_i ~\Longrightarrow~ \Yv' = -\sum\limits_{i=1}^q b_i \partial_{\phi_i}.
\end{align}
This is obviously a (conformal) Killing vector field of $\gamma'$. 
Therefore, the metric \eqref{eqWRKSform} is in the Kerr-de Sitter-like class. To calculate the conformal class of $\Yv$, we find an explicitly flat representative in $[\gamma']$. It is a matter of direct computation to check that the coordinate change\footnote{This coordinate change is inspired from the Kerr-de Sitter case, worked out in \cite{marspeondata21}.}
\begin{align}
 x_i = \frac{e^{\sqrt{\lambda} t}}{\sqrt{W'}}{\mlim_i} \cos \phi_i,
 \quad \quad 
 y_i = \frac{e^{\sqrt{\lambda} t}}{\sqrt{W'}} {\mlim_i} \sin \phi_i ,
\end{align}
brings the metric $\gamma'$ into the form
\begin{equation}
\gamma' = \frac{W'}{e^{2 \sqrt{\lambda} t}} \sum_{i=1}^q \lr{\dif x_i^2 + \dif y_i^2}.
\end{equation}
Hence $ \gamma_E :=e^{2 \sqrt{\lambda} t} W^{-1} \gamma$
is flat and $\Y'$ is in Cartesian coordinates $\{ x_i,y_i\}$:
\begin{equation}
 \Y' = -\sum\limits_{i=1}^q b_i (x_i \partial_{y_i} - y_i \partial_{x_i}). 
\end{equation}
Thus, $\Yv$ is the sum of generators of rotations within $q$ different orthogonal planes. Its corresponding skew-symmetric endomorphism of $\mathbb{M}^{1,n+1}$, with respect to an orthogonal unit basis $\{e_\alpha\}_{\alpha = 0}^{n+1}$ with $e_0$ timelike, can be directly calculated from \eqref{skwmatrix}:
\begin{align}
 F(\Yv) & = \left(
 \begin{array}{cc}
  0 & 0 \\ 0 & 0
 \end{array}
 \right)
  \bigoplus_{i=1}^{q} 
  \left(
 \begin{array}{cc}
  0 & -b_i \\ b_i & 0
 \end{array}
 \right) . \label{eqFWReven}
\end{align}
   The orthogonal sum of two-dimensional blocks is adapted to the decomposition 
   \begin{equation}
    \mathbb{M}^{1,n+1} = \Pi_0 \bigoplus_{i=1}^{q} \Pi_i
   \end{equation}
where $\Pi_0 = \spn{e_0,e_1}$ and $\Pi_i = \spn{e_{2i},e_{2i+1}}$ are $F$-invariant planes. The causal character of $\ker F(\Yv)$ is evidently timelike because $e_0 \in \ker F(\Yv)$ and the polynomial $\mathcal{Q}_{F^2}$ is also straightforwardly computable from the block form \eqref{eqFWReven}
\begin{equation}
\mathcal{Q}_{F^2}(x) = \prod_{i=1}^q (x - b_i^2).
\end{equation}
Then, permuting the indices $i$ so that the rotation parameters $b_i^2$ appear in decreasing order $b_1^2 \geq \cdots \geq b_q^2$, and applying Proposition \ref{defgammamu}, the conformal class of $\Yv$ is defined by the parameters 
\begin{equation}
 \{-\mu_t^2 = 0, \mu_s^2 = b_1^2; \mu_1^2 = b_2^2, \cdots, \mu_p^2 = b_q^2 \}.
\end{equation}
In consequence, for $b_i$s taking arbitrary values, this family covers every point in every region $\mathcal{R}^{(n,m)}_+$ of the quotient $\ckill(\mathbb{E}^n)/\confloc(\mathbb{E}^n)$, where $m$  is the number of vanishing $b_i$s.

\bigskip

\subsubsection{Generalized $\{a_i \rightarrow \infty \}$-limit Kerr-de Sitter.}\label{secainftyneven}
 
 In this subsection we perform the limits that cover the regions $\mathcal{R}^{(n,m)}_0$ of the quotient $\ckill(\mathbb{E}^n) / \newline \confloc(\mathbb{E}^n)$. In this case, the limits are achieved by absorbing $\param$ in the $t$ coordinate, i.e. defining $t' = \param^{-1} t$, so that the limit vector field $\Y' = \lim_{\param \rightarrow  0}\Yv_\param$ has a non-zero term in $\partial_{t'}$. It turns out that these limits lie in the Kerr-de Sitter-like class provided that the Kerr-de Sitter metric from which they are calculated have one rotation parameter vanishing. Otherwise the limit of the boundary metric is degenerate. Thus we will assume that $a_q = 0$. We name the limit spacetimes  obtained in this way $\{ a_i \rightarrow \infty\}$-{\it limit-Kerr-de Sitter} because the conformal class that characterizes them is similar to the $n=3$ case \cite{limitkds}.

 Consider the de Sitter metric \eqref{eqgdSodd} with the change of coordinates
 \begin{equation}
  \rho = \param \rho', \quad \quad t = \param t',\quad\quad \phi_{q} = \param \Phi, \quad\quad \hat\m_i = \param \mlim_i  \quad\quad  (i=1,\cdots,p),
 \end{equation}
  where note that the coordinate $\hat\m_q$ and the angles $\phi_i$ $(i=1, ...,p)$ remain unaltered. In addition, let us redefine the parameters 
  \begin{equation}
    M = M' \param^{n},\quad \quad a_i = \param^{-1} b_i\quad\quad(i=1,\cdots,p).
 \end{equation} 
 By \eqref{eqconsthatalpha}, the coordinates $\{\mlim_i, \hat\m_q \}$ satisfy in the limit $\param \rightarrow 0$: 
  \begin{equation}\label{eqredefalphaainfty}
 \hat\m_q^2  + \sum_{i=1}^p \lambda b_i^2  \mlim^2_i= 1.
\end{equation} 
The limits of $\Pi$ and $W$, $\Xi$ are obtained immediately from \eqref{eqPi} and  \eqref{eqWXihats} respectively. They are
 \begin{equation}\label{eqdefWXiainfty1}
 \Pi' = \lim_{\param \rightarrow 0} \Pi =\prod\limits_{j=1}^p (1 + \rho'^2 b_j^2),
\quad\quad
W' =\lim_{\param \rightarrow 0} W = \hat\m_{q}^2 ,\quad \quad  \Xi' = \lim_{\param \rightarrow 0} \Xi = \hat\alpha_{q}^2 + \sum\limits_{i=1}^{p} \frac{\lambda b_i^2 }{1 + \rho'^2 b_i^2}\mlim^2_i.
\end{equation}
In addition from \eqref{eqtilH} and \eqref{eqkodd} and the redefinitions above it follows
\begin{equation}\label{eqHkkainfty1}
 \til\H' \til k' \otimes \til k' :=\lim_{\param \rightarrow 0} \H \til k \otimes \til k = \underbrace{\frac{2 M'}{\Pi' \Xi'} \rho'^{n-2}}_{=:\til\H'} \Big( \underbrace{W' \dif t' + \frac{\Xi'}{ \lambda }  \dif \rho' - \sum\limits_{i=1}^p { b_i \mlim^2_i} \dif \phi_i}_{=:\til k'}\Big)^2.
\end{equation}
 
 Before taking the limit, we rewrite the de Sitter metric \eqref{eqgdSodd} in the new coordinates and separate
the terms multiplying $\dif \hat\m_{q}$
\begin{align}
  \tilg_{dS} & =  - W \frac{(\param^{2}\rho'^2 - \lambda )}{\rho'^2} \dif t'^2 + \frac{\Xi}{\param^{2}\rho'^2 - \lambda } \frac{\dif \rho'^2}{\rho'^2} + \frac{1}{ \param^{2} \rho'^2} \lr{\dif \hat\alpha_{q}^2 + \hat\alpha_q^2  \param^{2} \dif \Phi^2} + \sum_{i=1}^p \frac{1 + \rho'^2 b_i^2}{ \rho'^2}{ \lr{\dif \mlim^2_i + \mlim^2_i \dif \phi_i^2} }
   \\
  & 
   + \frac{(\param^{2}\rho'^2- \lambda )}{\lambda W \param^{2} \rho'^2 } \lr{\hat\m_q \dif\hat\m_{q} + \param^2 \sum_{i= 1}^{p}{ \mlim_i \dif  \mlim_i}}^2 , \label{eqgdSainfty1}
\end{align}
 with
 \begin{equation}\label{eqdefWXiainfty1}
 W = \hat\m_{q}^2 + \param^2\sum\limits_{i=1}^{p} {\mlim^2_i},\quad \quad  \Xi = \hat\alpha_{q}^2 + \sum\limits_{i=1}^{p} \frac{\param^2 + \lambda b_i^2 }{1 + \rho'^2 b_i^2}\mlim^2_i.
\end{equation}
%  The one-form $\til k$, the functions $\Pi$ and $\H$ are
%  \begin{align}
%    \til k & =  W \param\dif t' - \frac{\Xi}{\param^{2}\rho^2 - \lambda } \param \dif \rho' - \sum\limits_{i=1}^p {\param b_i \m'^2_i} \phi_i, \quad\quad
%   \Pi  =  \prod\limits_{j=1}^p (1 + \rho'^2 b_j^2), \quad \til{\mathcal{H}} =  \frac{2M \param^{n-2}\rho'^{n-2}}{ \Pi \Xi}, \quad M\in \mathbb{R}  \label{eqkpiainfty1}.
%  \end{align}
%  In the limit $\param\rightarrow 0$, the constraint for the $\{\m'_i, \hat\m_q \}$ coordinates \eqref{eqconsthatalpha} becomes 
%  \begin{equation}\label{eqredefalphaainfty}
%  \hat\m_p  + \sum_{i=1}^p \lambda b_i^2 \hat \m'^2_i= 1.
% \end{equation} 
%  Moreover, it is obvious that $W$, $\Xi$ and $\Pi$ have a finite limit
%  \begin{equation}
%   W' := \lim_{\param \rightarrow 0}  W = \alpha_q^2,\quad \quad  \Xi' := \lim_{\param \rightarrow 0} \Xi = \hat\m_q^2 + \sum\limits_{i=1}^{p} \frac{\lambda b_i^2 \m'^2_i}{1 + \rho'^2 b_i^2} \quad\quad \Pi' := \lim_{\param \rightarrow 0}  \Pi = \prod\limits_{j=1}^p (1 + \rho'^2 b_j^2),
%  \end{equation}
% and with the redinition of mass $M = M' \param^{n}$, it also follows
% \begin{equation}\label{eqHkkainfty1}
%  \til\H' \til k' \otimes \til k' :=\lim_{\param \rightarrow 0} \H \til k \otimes \til k = \frac{2 M'}{\Pi' \Xi'} \rho'^{n-2} \lr{ W' \dif t' + \frac{\Xi'}{ \lambda }  \dif \rho' - \sum\limits_{i=1}^p { b_i \m'^2_i} \dif \phi_i}^2.
% \end{equation}
Only the terms involving $\dif \hat \m_q$ are troublesome in the limit $\param \rightarrow 0$. Let us gather them to get 
\begin{align}
  g_{(\alpha_q)} & :=  \frac{1}{\rho'^2 \param^{2}}\dif \hat\m_{q}^2 +  \frac{(\param^{2}\rho'^2- \lambda )}{\lambda W \param^{2} \rho'^2 }\lr{\hat\m_q \dif\hat\m_{q} + \param^2 \sum_{i= 1}^{p}{ \mlim_i \dif \mlim_i}}^2 \\ 
   & = \frac{1}{\param^2\rho'^2} \lr{1 + \frac{ (\param^{2}\rho'^2- \lambda )\hat\m_q^2}{\lambda W } }\dif \hat\m_q^2\\ 
   & +  \frac{ (\param^{2}\rho'^2- \lambda ) }{\lambda W \param^2  \rho'^2 }\lr{ \param^4 \lr{ \sum_{i= 1}^{p}{\mlim_i \dif  \mlim_i}}^2 + 2 \param^2 \hat\m_q \dif \hat\m_q \lr{ \sum_{i= 1}^{p}{\mlim_i \dif \mlim_i}} }.\label{eqline2}
   \end{align}
Writting $W$ in coordinates $\{\beta_i,\hat \m_q \}$, the term in $\dif \hat\m_q^2$ takes the limit
\begin{align}
 \lim_{\param \rightarrow 0}\frac{1}{\param^2\rho'^2} \lr{1 + \frac{ (\param^{2}\rho'^2- \lambda )\hat\m_q^2}{\lambda W } }\dif \hat\m_q^2 = \lim_{\param \rightarrow 0} \frac{\lambda \param^2 \lr{\sum_{i=1}^p \beta_i^2} + \param^2 \rho'^2 \hat\m_q^2}{\param^2 \rho'^2 \lambda(\hat\m_q^2 + \param^2 \sum_{i=1}^p \beta_i^2)} \dif \hat \m_q^2 = \frac{1}{\lambda} + \frac{\sum_{i=1}^p \beta_i^2}{\rho'^2 \hat \m_q^2},
\end{align}
while the limit of the last two terms is direct
\begin{equation}
 \lim_{\param\rightarrow 0} \frac{ (\param^{2}\rho'^2- \lambda ) }{\lambda W \param^2  \rho'^2 }\lr{ \param^4 \lr{ \sum_{i= 1}^{p}{\mlim_i \dif  \mlim_i}}^2 + 2 \param^2 \hat\m_q \dif \hat\m_q \lr{ \sum_{i= 1}^{p}{\mlim_i \dif \mlim_i}} } = -\frac{ 2  \dif\hat\m_q}{\rho'^2 \hat\m_q} \lr{ \sum_{i= 1}^{p}{\mlim_i \dif  \mlim_i}}.
\end{equation}
Thus, the limit $\param \rightarrow 0$ of \eqref{eqgdSainfty1} is
\begin{align}
  \tilg'_{dS} & =    \frac{ \lambda\hat \m_q^2 }{\rho'^2} \dif t'^2 - \frac{\Xi'}{ \lambda } \frac{\dif \rho'^2}{\rho'^2} + \frac{\hat \m_q^2  \dif \Phi^2}{\rho'^2 }  + \sum_{i=1}^p \frac{1 + \rho'^2 b_i^2}{ \rho'^2}{ \lr{\dif \mlim^2_i + \mlim^2_i \dif \phi_i^2} }
  + \lr{\frac{1}{\lambda} + \frac{\sum_{i=1}^p \beta_i^2}{\rho'^2 \hat \m_q^2}} \dif\hat\m_q^2    -\frac{ 2  \dif\hat\m_q}{\rho'^2 \hat\m_q} \lr{ \sum_{i= 1}^{p}{\mlim_i \dif  \mlim_i}},
  \end{align}
where we have already substituted $W' = \hat\m_q ^2$. From the argument above (cf. \eqref{eqriemconstc}) $ \tilg'_{dS}$ is locally isometric to de Sitter. Thus, we have all the ingredients to build up the limit Kerr-Schild metric, namely
\begin{equation}
  \tilg' = \tilg'_{dS} + \til\H' \til k' \otimes \til k',\quad\quad \til\H' = \frac{2M' \rho'^{n-2}}{ \Pi' \Xi'}, \quad M'\in \mathbb{R}.
\end{equation}

\bigskip

We now calculate the asymptotic structure and verify that indeed, these spacetimes correspond to the regions $\mathcal{R}^{(n,m)}_0$ in the space of orbits. The boundary metric is 
\begin{align}
 \gamma' & = \rho'^2 \tilg'|_{\scri} = { \lambda\hat\m_q^2 } \dif t'^2  + {\hat\m_q^2  \dif \Phi^2}  + \sum_{i=1}^p { \lr{\dif \mlim^2_i + \mlim^2_i \dif \phi_i^2} } 
   +  \lr{ \sum_{i=1}^p  {\mlim^2_i}} \frac{\dif \hat\m_q^2 }{\hat\m_q^2}- \frac{2  \dif\hat\m_q} {\hat\m_q}  \lr{ \sum_{i= 1}^{p}{\mlim_i \dif  \mlim_i}}. \label{eqgammaainfty}
\end{align}
As usual,  the TT tensor $D_{\Y'}$ is directly calculated with equation  \eqref{eqexpweyl}
\begin{align}
 D = \rho'^{n-2} C_\perp |_\scri =  -\lambda n(n-2) M' \lr{\Y' \otimes \Y' - \frac{|\Y'|^2_{\gamma'}}{n} \gamma'},
\end{align}
where $\Y'$ is the projection of $\til k'$ onto $\scri$
\begin{align}
 \Yf' =   \hat\m_q^2 \dif t' - \sum\limits_{i=1}^p { b_i \mlim^2_i} \dif \phi_i ~\Longrightarrow~ \Y' =  \frac{1}{\lambda}\partial_{t'} - \sum\limits_{i=1}^p { b_i } \partial_{\phi_i}.
\end{align}
To calculate the conformal class of $\Y'$, we look for a flat representative in $[\gamma']$ written in Cartesian cordinates. It turns out to be useful to scale the coordinates $\{ \mlim_i \}_{i=1}^p$ as
\begin{equation}
 \til\mlim_i = \frac{\mlim_i}{\hat\alpha_q}.
\end{equation}
Replacing $\mlim_i = \til\mlim_i \hat\alpha_q$ and $\dif  \mlim_i = \hat\m_q \dif \til\mlim_i + \til\mlim_i \dif \hat\m_q$ in equation \eqref{eqgammaainfty}, all terms in $\dif \hat\m_q$ cancel out and we are left with the expression
\begin{align}
 \gamma'  = \hat\m^2_q \lr{ \lambda \dif t'^2  + { \dif \Phi^2}  + \sum_{i=1}^p { \lr{\dif {\til\mlim}_i^2 + {\til\mlim}_i^2 \dif \phi_i^2} } 
   }. \label{eqgammaainftyconflat}
\end{align}
This determines a flat representative $ \gamma_E := \hat\m^{-2}_q \gamma'$, where by  \eqref{eqredefalphaainfty}, $\hat\m_q$ is written explictly in terms of $\{\tilde \beta_i\}$ as $\hat\m_q^2 = (1 + \lambda \sum_{i=1}^p b_i^2 \tilde\beta_i^2)^{-1}$. The set $\{ \tcar := \lambda^{1/2} t', \Phi, x_i:=\til\mlim_i \cos \phi_i, y_i := \til\mlim_i \sin \phi_i \}$ define Cartesian coordinates for $\gamma'$, into  which vector field $\Y'$ reads
\begin{equation}
 \Y' =  \frac{1}{\lambda^{1/2}}\partial_\tcar - \sum\limits_{i=1}^p { b_i } (x_i \partial_{y_i} - y_i \partial_{x_i}),
\end{equation}
i.e. is the sum of translation along the coordinate $\tcar$ plus the sum of $p$ independent orthogonal rotations. Its correspoding skew-symmetric endomorphism  of $\mathbb{M}^{1,n+1}$ is by \eqref{skwmatrix} 
\begin{align}
 F(\Yv) & = \left(
 \begin{array}{cccc}
  0 & 0 & \frac{\lambda^{-1/2}}{2} & 0 \\ 
  0 & 0 & -\frac{\lambda^{-1/2}}{2} & 0 \\
   \frac{\lambda^{-1/2}}{2} & \frac{\lambda^{-1/2}}{2} & 0  & 0 \\
   0 & 0 & 0 & 0 
 \end{array}
 \right)
  \bigoplus_{i=1}^{p} 
  \left(
 \begin{array}{cc}
  0 & -b_i \\ b_i & 0
 \end{array}
 \right) \label{eqFainftyeven}
\end{align}
 in an orthogonal unit basis $\{e_\alpha\}_{\alpha = 0}^{n+1}$ with $e_0$ timelike. Similar to subsection \ref{subsecWReven}, the direct sum \eqref{eqFainftyeven} is adapted to the decomposition 
   \begin{equation}
    \mathbb{M}^{1,n+1} = \mink{1,3} \bigoplus_{i=1}^{p} \Pi_i,
   \end{equation}
where $\mink{1,3} = \spn{e_0,e_1,e_2,e_3}$ and $\Pi_i = \spn{e_{2(i+1)},e_{2(i+1)+1}}$ are $F$-invariant subspaces. The causal character of  $\ker F(\Y')$ is determined by the causal character of $\ker F(\Y')|_{\mink{1,3}}$, because every non-spacelike vector $v \in \ker F(\Y')$ must have non-zero projection $v_0 \in \mink{1,3}$ with $v_0 \in \ker F(\Y')|_{\mink{1,3}}$. It is immediate to calculate $\ker F(\Y')|_{\mink{1,3}} = \spn{e_0-e_1,e_3}$, where $e_0-e_1$ is a null vector in $\ker F(\Y')$, thus $\ker F(\Y')$ is degenerate. The polynomial $\mathcal{Q}_{F^2}$ is by direct calculation
\begin{equation}
 \mathcal{Q}_{F^2} = x^2 \prod_{i = 1}^p (x-b_i^2).
\end{equation}
This, by Proposition \ref{defgammamu}, gives the parameters for the conformal class of $\Y'$ 
\begin{equation}
 \{-\mu_t^2 = 0, \mu_s^2 = 0 ; \mu_1^2 = b_1^2, \cdots, \mu_p^2 = b_p^2 \}.
\end{equation}
This collection of conformal classes covers every point in every region $\mathcal{R}^{(n,m)}_0$, where $m$ is the number of zero $b_i$ parameters.

 \subsection{Limits $n$-odd}
 
  One major difference between the $n$ odd and even cases is that, only when $n$ is even the region 
 $\mathcal{R}_-^{(n,0)}$ (namely the portion corresponding to Kerr-de Sitter with none of the rotation parameters vanishing) 
admits limit in the whole of $\ckill(\mathbb{E}^n)/\confloc(\mathbb{E}^n)$ (cf. Proposition \ref{proplimts}). This is what allowed us to construct all spacetimes in the Kerr-de Sitter-like class directly as limits of Kerr-de Sitter in subsection \ref{secnevenlims}. In the $n$ odd case, no sequence in $\mathcal{R}^{(n,0)}_-$ takes limit at $\mathcal{R}^{(n,0)}_+$ and viceversa, because they are disjoint and open subspaces by Proposition \ref{proplimts}. In subsection \ref{secwicknodd} we deal with this issue by constructing, using analytic continuation of Kerr-de Sitter, the set spacetimes whose CKVF class  corresponds to $\mathcal{R}^{(n,0)}_+$. To do this, we define a Wick rotation in arbitrary $n+1$ even dimensions (generalizing the transformation in \cite{wickrot98}). We name the resulting family Wick-rotated-Kerr-de Sitter, in analogy with the $n=3$ case in \cite{Kdslike}. From these, all spacetimes in $\mathcal{R}^{(n,m)}_+$ can be obtained easily. Subection \ref{secainftynodd} is devoted to finding the spacetimes whose CKVF class corresponds to $\mathcal{R}^{(n,m)}_0$. These are obtained by performing limits to Kerr-de Sitter, similar to those in subsection \ref{secainftyneven}.

 \subsubsection{Generalized Wick-rotated}\label{secwicknodd}
 
 Let now $n$ be odd and
 let us consider the Kerr-de Sitter metric with none of the rotation parameters $a_i$ equal to zero. The generalization of the Wick rotation is given by the following complex coordinate transformation
 \begin{align}\label{eqwickrotcoord}
  t = i  t',\quad\quad \rho = i  \rho' \quad\quad \hat\m_i = i  \mlim_i \quad i= 1, \cdots, p,
 \end{align}
with $t', \rho', \mlim_i \in \mathbb{R}$, and the redefinition of parameters 
\begin{equation}\label{eqwickrotparam}
 a_i = - i b_i \quad \quad M =(-1)^{\frac{n+1}{2}} i M',\quad\quad M' \in  \mathbb{R}.
\end{equation}
Note that the only the first $p$ $\hat\m_i$ coordinates have been ``rotated''. Introducing  $\mlim_{p+1} := \hat\m_{p+1}$, \eqref{eqconsthatalpha} gives:
\begin{equation}\label{eqalphapWRodd}
 \mlim_{p+1}^2 - \sum_{i=1}^{p} (1 - \lambda b_i^2) \mlim^2_i = 1.
\end{equation}
By performing the Wick rotation \eqref{eqwickrotcoord}, the functions $W$, $\Xi$ in \eqref{eqWXihats} and $\Pi$ in \eqref{eqPi} are now redefined
 \begin{equation}\label{eqdefWXiWRodd}
  W' := \mlim^2_{p+1} - \sum\limits_{i=1}^{p}  \mlim^2_i \quad \quad \Xi' :=\mlim^2_{p+1} - \sum\limits_{i=1}^{p} \frac{1 - \lambda b_i^2}{1 +  \rho'^2 b_i^2}\mlim^2_i,
  \quad\quad \Pi'  =  \prod\limits_{j=1}^p (1 +  \rho'^2 b_j^2).
\end{equation}
 The spacetime metric is given by  
\begin{equation}\label{eqwickKSodd}
 \tilg' = \tilg'_{dS} + \til{\H}' \til k' \otimes \til k,\quad\quad \til{\H}' =\frac{2 M' \rho'^{n-2}}{\Pi' \Xi'}\quad M' \in \mathbb{R}
\end{equation}
with 
\begin{align}
  \tilg'_{dS} & =    W' \frac{(\rho'^2 + \lambda )}{\rho'^2} \dif t'^2 - \frac{\Xi'}{\rho'^2 + \lambda } \frac{\dif \rho'^2}{\rho'^2} -  \frac{\dif \mlim^2_{p+1}}{\rho'^2} \\ &  + \sum_{i=1}^p \frac{1 +  \rho'^2 b_i^2}{\rho'^2} \lr{\dif  \mlim^2_i + \mlim^2_i \dif \phi_i^2} 
  + \frac{ (\rho'^2+ \lambda )}{ \lambda W' \rho'^2 }  \frac{\dif W'^2}{4}, \label{eqgdSWRoddfin}\\
  \til k & : =    \lr{ W' \dif \til t' + \frac{\Xi'}{\rho'^2 + \lambda } \dif \rho' - \sum\limits_{i=1}^p {b_i \mlim^2_i} \dif \phi_i}.
\end{align}  
The domain of definition of the coordinates is $ t', \rho' \in \mathbb{R}$, and the $\phi_i \in [0, 2 \pi)$ are still angles. Moreover, $\beta^2_{p+1} >0$, and $\{\mlim_i \}_{i=1}^p$ are restricted to a sufficiently small neighbourhood of $\{ \mlim_i = 0 \}_{i=1}^p$ so that $ W', \Xi'$ are positive (see \eqref{eqdefWXiWRodd}). With this restriction of the coordinates, vanishing values of the $b_i$ parameters are allowed. 

\bigskip

 The signature is not necessarily preserved after a Wick rotation, so we still need to prove that the  Wick-rotated Kerr-de Sitter metrics are Lorentzian. They are obviously $\Lambda>0$-vacuum Einstein because we have only performed a (complex) change of coordinates. From the Einstein equations and positivity of the cosmological constant, it follows that the boundary metric is positive definite if and only if the spacetime metric is Lorentzian in a neighbourhood of $\scri$. In addition, note that the boundary metric induced by $\tilg'$ is the same as the one induced by $\tilg'_{dS}$. Moreover, $\tilg'_{dS}$ is clearly Einstein of constant curvature. Thus, proving that $\gamma'$ is positive definite, in turn, also proves that $\tilg'_{dS}$ is Lorentzian and therefore locally isometric to de Sitter.

The metric induced at $\scri$ is, directly from \eqref{eqgdSWRoddfin},
 \begin{align}
\gamma' 
  &  =   W'  \lambda  \dif t'^2  -  {\dif {\mlim}_{p+1}^2} + \sum_{i=1}^p { \lr{\dif  {\mlim}_i^2 +  {\mlim}_i^2 \dif \phi_i^2} }
  + \frac{1}{W'} \frac{\dif W'^2}{4}. \label{eqhatgamwr} 
 \end{align}
 The explicitly conformally flat form is obtained under the change of coordinates 
 \begin{equation}\label{eqtilal}
  \til\mlim_i = \frac{\mlim_i}{W'^{1/2}},\quad\quad i = 1, \cdots, p + 1
 \end{equation}
Observe that by redefining all the $p+1$ coordinates we now have
\begin{equation}
 W' = \mlim^2_{p+1} - \sum_{i=1}^p \mlim^2_i  = W' (\til\mlim^2_{p+1} - \sum_{i=1}^p \til\mlim^2_i ) ~ \Longrightarrow ~ \til\mlim^2_{p+1} - \sum_{i=1}^p \til\mlim^2_i = 1  
\end{equation}
and 
\begin{align}
 W' \til\mlim_{p+1}^2 & = \mlim^2_{p+1}  =   1 + \sum_{i=1}^p (1 - \lambda b_i^2) {\mlim}_i^2 =  1 + W' \sum_{i=1}^p (1 - \lambda b_i^2) {\til\mlim}_i^2 \\
  & \Longrightarrow ~ W' = \frac{1}{1 + \sum_{i=1}^p \lambda b_i^2 \til\mlim_i^2}.
\end{align}
Inserting the coordinate change \eqref{eqtilal} into \eqref{eqhatgamwr} gives
\begin{equation}
 \gamma' =  W' \lr{\lambda  \dif t'^2  -  {\dif {\til\mlim}_{p+1}^2} + \sum_{i=1}^p { \lr{\dif  {\til\mlim}_i^2 +  {\til\mlim}_i^2 \dif \phi_i^2} }}.
\end{equation}
From this expression it already follows that $\gamma'$ is Riemannian, because the restriction $\{t' = const.\}$ defines a spacelike hyperboloid in a $(p+1)$- dimensional Minkowski space. More specifically, let us introduce the parametrization 
\begin{equation}
 \til\mlim_{p+1} = \cosh\chi, \quad\quad \til\mlim_{i} = \nu_i \sinh\chi,\quad i = 1,\cdots, p,\quad\mbox{with}\quad \sum_{i=1}^p \nu_i^2 = 1,
\end{equation}
so that 
\begin{equation}
  \gamma' =  W' \lr{\lambda  \dif t'^2  +  {\dif \chi^2} + \sinh^2\chi~\gamma_{\mathbb{S}^{n-2}}},
\end{equation}
where
\begin{equation}
 \gamma_{\mathbb{S}^{n-2}} := \sum_{i=1}^p { \lr{\dif  {\nu}_i^2 +  {\nu}_i^2 \dif \phi_i^2} }|_{\sum_{i=1}^p \nu_i^2 = 1}
\end{equation}
is an $(n-2)$-dimensional spherical metric. 
 Finally, defining the coordinates
 \begin{equation}
  z := \frac{\sin\sqrt{\lambda} t'}{\cos\sqrt{\lambda} t' + \cosh\chi} \quad\quad x_i:=  \frac{\sinh\chi}{\cos\sqrt{\lambda} t' + \cosh\chi} \nu_i \cos\phi_i,\quad\quad y_i:=  \frac{\sinh\chi}{\cos\sqrt{\lambda} t' + \cosh\chi} \nu_i \sin\phi_i
 \end{equation}
 for $i = 1,\cdots, p$, one has
\begin{equation}
 \gamma_E := \frac{1}{ W'(\cos\sqrt{\lambda} t' + \cosh\chi)^2} \gamma' = \dif z^2 + \sum_{i=1}^p \lr{\dif x_i^2 + \dif y_i^2}  .
\end{equation}
Thus $\gamma_E$ is a flat representative $\gamma_E \in [\gamma']$ and $\{z, x_i,y_i \}$ are Cartesian coordinates of $\gamma_E$.

We continue by calculating the electric part of the rescaled Weyl tensor at $\scri$. As usual, the expression follows
from formula \eqref{eqexpweyl}. We give it first in coordinates $\{t',\rho', \mlim_i, \phi_i \}$ : 
\begin{align} 
 D_{\Y'} = \rho'^{n-2} C_\perp |_\scri =  - \lambda n(n-2)  M' \lr{ \Y' \otimes  \Y' - \frac{|\Y'|^2_{\gamma'}}{n}\gamma'},
\end{align}
where  $\Y'$ is the projection of $\til k'$ onto $\scri$
\begin{align}
  \Yf' =  W' \dif  t'  - \sum\limits_{i=1}^p {b_i  {\mlim}_i^2} \dif \phi_i ~\Longrightarrow~
  \Y' = \frac{1}{\lambda} \partial_{t'} - \sum\limits_{i=1}^q b_i \partial_{\phi_i}.
\end{align}
To express $\Y'$ in Cartesian coordinates $\{z, \{ x_i,y_i\}_{i=1}^p \}$, firstly observe
\begin{align}
 \frac{\partial z}{\partial t'} & = \sqrt{\lambda} ~\frac{\cos\sqrt{\lambda} t'(\cos\sqrt{\lambda} t'+ \cosh\chi) + \sin^2\sqrt{\lambda} t'}{(\cos\sqrt{\lambda} t'+ \cosh\chi)^2} = \sqrt{\lambda}\lr{\frac{1}{2} - \frac{1}{2} +  \frac{1 + \cos\sqrt{\lambda} t' \cosh\chi}{(\cos\sqrt{\lambda} t'+ \cosh\chi)^2}} \\
 & = 
%  \frac{1}{2} +  \frac{2- \cos^2\sqrt{\lambda} t' - \cosh^2\chi}{2(\cos\sqrt{\lambda} t'+ \cosh\chi)^2} =
 \frac{\sqrt{\lambda}}{2} + \frac{\sqrt{\lambda}}{2}\lr{ z^2 - \sum_{i= 1}^p (x_i^2 + y_i^2)},
\end{align}
and it is also straightforward that
\begin{equation}
 \frac{\partial x_i}{\partial t'} = \sqrt{\lambda} z x_i , \quad \quad\frac{\partial y_i}{\partial t'} =\sqrt{\lambda} z y_i .
\end{equation}
 Then 
\begin{align}
 \partial_t = \frac{\partial z}{\partial  t} \partial_z + \sum_{i=1}^p \lr{\frac{\partial x_i}{\partial  t} \partial_{x_i} + \frac{\partial y_i}{\partial  t} \partial_{y_i}} = \frac{\sqrt{\lambda}}{2}\lr{ 1 + z^2 - \sum_{i= 1}^p (x_i^2 + y_i^2)} \partial_z + \sqrt{\lambda} z \sum_{i=1}^p (x_i \partial_{x_i} + y_i \partial_{y_i})
\end{align}
and on the other hand
\begin{equation}
 \partial_{\phi_i} = \frac{\partial x_i}{\partial \phi_i} \partial_{x_i} + \frac{\partial y_i}{\partial \phi_i} \partial_{y_i} = x_i \partial_{y_i} - y_i \partial_{x_i}.
\end{equation}
Therefore 
\begin{equation}\label{eqCKVFWRodd}
 \Yv = \frac{1}{2\sqrt{\lambda}}\lr{ 1 + z^2 - \sum_{i= 1}^p (x_i^2 + y_i^2)} \partial_z +  \frac{z}{\sqrt{\lambda}} \sum_{i=1}^p (x_i \partial_{x_i} + y_i \partial_{y_i}) - \sum_{i=1}^p b_i(x_i \partial_{y_i} - y_i \partial_{x_i}).
\end{equation}
Denoting the coordinates as $\{ X^A \}_{A =1}^n := \{z, \{ x_i,y_i\}_{i=1}^p \}$, $\Y'$ is a CKVF with $\ar^A ={\delta^A}_1 \lambda^{-1/2} $, $\br^A =\ar^A/2$, plus a sum of orthogonal rotations with parameters $b_i$. The associated skew-symmetric endomorphism of $\mathbb{M}^{1,n+1}$ is directly computable from expression \eqref{eqCKVFWRodd}
and \eqref{skwmatrix}
\begin{align}
 F(\Yv) & = \left(
 \begin{array}{ccc}
   0 & 0 & -\frac{3\lambda^{-1/2}}{4}  \\
   0 & 0 & -\frac{5\lambda^{-1/2}}{4} \\
    -\frac{3\lambda^{-1/2}}{4}& \frac{5\lambda^{-1/2}}{4} & 0  
 \end{array}
 \right)
  \bigoplus_{i=1}^{p} 
  \left(
 \begin{array}{cc}
  0 & -b_i \\ b_i & 0
 \end{array}
 \right) . \label{eqFWRodd}
\end{align}
$F(\Yv)$ is referred to an orthogonal unit basis $\{e_\alpha\}_{\alpha = 0}^{n+1}$ with $e_0$ timelike and as in the previous sections the direct sum \eqref{eqFWRodd} is adapted to the decomposition 
   \begin{equation}
    \mathbb{M}^{1,n+1} = \mink{1,2} \bigoplus_{i=1}^{p} \Pi_i
   \end{equation}
where $\mink{1,2} = \spn{e_0,e_1,e_2}$ and $\Pi_i = \spn{e_{2 i+1},e_{2(i+1)}}$ are $F$-invariant subspaces. The causal character of $\ker F(\Yv)$ is straightforwardly determined by checking that $v : =  5 e_0 + 3 e_1 $ is timelike and that it belongs to $\ker F(\Yv)$. Thus $\ker F(\Yv)$ is timelike.  

On the other hand, the polynomial $\mathcal{Q}_{F^2}$ is 
\begin{equation}
 \mathcal{Q}_{F^2}(x) = (x- \frac{1}{\lambda}) \prod_{i = 1}^p (x-b_i^2) = \prod_{i = 0}^p (x-b_i^2)
\end{equation}
where for the last equality we have set $b_{0}^{2} := 1/\lambda$. Now let $\{ \til b_i\}_{i = 0}^{p}$ the parameters $b_i$ sorted in decreasing order $\til b_0^2 \geq \cdots \geq \til b_{p}^2$. 
Then  by Proposition \ref{defgammamu}, the conformal class of $\Yv$ is given by
$\{\sigma = \til b_0^2; \mu_1^2 = \til b_1^2, \cdots, \mu_p^2 = \til b_p^2 \}$. Note that the value of one of the parameters is $1/\sqrt{\lambda}$, so it is a priori fixed. To cover the whole the space of parameters $\mathcal{R}^{(n,m)}_+$, we must consider the scaling freedom of $\Y$, just like we explained in the case of Kerr-de Sitter. Taking this into account, this family of metrics covers every point in all the regions $\mathcal{R}^{(n,m)}_+$ in the space of conformal classes.

 \subsubsection{Generalized $\{a_i \rightarrow \infty \}$-limit Kerr-de Sitter.}\label{secainftynodd}
 
 In this subsection we calculate the remaining family of metrics which completes the Kerr-de Sitter-like class for $n$ odd, i.e. those corresponding to the regions $\mathcal{R}^{(n,m)}_0$ in the space of conformal classes.   Analogously  to the case of $n$ even (cf. subsection \ref{secainftyneven}), these are called generalized $\{a_i \rightarrow \infty \}$-limit Kerr-de Sitter, also extending the definition in \cite{limitkds}.

 Contrary to the $n$ even case, if $n$ is odd we obtain a good limit from Kerr-de Sitter with none of the rotation parameters initially vanishing. The reason is that having only $p$ non-vanishing rotation parameters $a_i = \param^{-1} b_i~(i=1,\cdots,p)$ (recall that $a_{p+1} = 0$ was defined for notational reasons) the function $W$ remains finite in the limit $\lim_{\param\rightarrow 0 } W = \m_{p+1}^2$ if we scale the first $p$ coordinates $\hat\alpha_i = \param \mlim_i$. Thus, $\gamma_\param = \param^{-2} \gamma$ and $\Y_\param$ both admit a finite limit $\param \rightarrow 0 $, as soon as the coordinate $t$ is rescaled to $t = \param t'$ (see subsection \ref{secainftyneven} for comparison).
 
  Consider the de Sitter metric \eqref{eqgdSodd} with the change of coordinates
 \begin{equation}\label{eqscale1}
  \rho = \param \rho', \quad \quad t = \param t', \quad\quad \hat\m_i = \param \mlim_i\quad\quad(i=1,\cdots,p),
 \end{equation}
 where notice that $\hat\m_{p+1}$ has not been scaled. Also consider the 
  redefinition of parameters 
  \begin{equation}\label{eqscale2}
    M = M' \param^{n},\quad \quad a_i = \param^{-1} b_i\quad\quad(i=1,\cdots,p).
 \end{equation}
Unlike in the $n$ even case, no $\phi$ angle is associated to $\hat\m_{p+1}$, so there is no need the rescale any of the $\phi_i$ coordinates. All calculations are analogous to those in subsection \ref{secainftyneven}, so we provide here less detail. 

First, the scaled coordinates $\{ \mlim_i \}_{i=1}^p$ and $\hat\m_{p+1}$ satisfy when $\param \rightarrow 0$
\begin{equation}\label{eqalphapWRodd2}
  \hat\m^2_{p + 1} + \sum_{i=1}^{p}  \lambda b_i^2 \mlim^2_i = 1.
\end{equation}
The functions $W$, $\Xi$ (cf. \eqref{eqWXihats}) and $\Pi$ (cf. \eqref{eqPi}) take the limit 
\begin{equation}
 W' := \lim_{\param \rightarrow 0} W
 =  \hat\m^2_{p+1},\quad \quad \Xi' := \lim_{\param \rightarrow 0 } \Xi :=\hat\m_{p+1}^2 + \sum\limits_{i=1}^{p} \frac{\lambda b_i^2 \mlim^2_i}{1 + \rho'^2 b_i^2},\quad \quad \Pi'  =  \prod\limits_{j=1}^p (1 +  \rho'^2 b_j^2). 
\end{equation}
The limit of the term $\til \H~ \til k \otimes \til k$ present no difficulties since the scalings defined in \eqref{eqscale1} and \eqref{eqscale2} compensate each other so that no divergences appear. Then
\begin{equation}\label{eqHkkainfty2}  
 \til\H' \til k' \otimes \til k' :=\lim_{\param \rightarrow 0} \H \til k \otimes \til k = \underbrace{\frac{2 M'}{\Pi' \Xi'} \rho'^{n-2}}_{=:\til\H'} \Big(\underbrace{ W' \dif t' + \frac{\Xi'}{ \lambda }  \dif \rho' - \sum\limits_{i=1}^p { b_i \mlim^2_i} \dif \phi_i}_{=:\til k'}\Big)^2.
\end{equation}
For the de Sitter background \eqref{eqgdSodd} a computation analogous to the case of $n$ even shows that the terms in $\dif \hat{\alpha}_{p+1}$ do not diverge. In fact, the limit of de Sitter as $\param \rightarrow 0$ is 
\begin{align}
 \tilg'_{dS} & =    \frac{ \lambda\hat \m_{p+1}^2 }{\rho'^2} \dif t'^2 - \frac{\Xi'}{ \lambda } \frac{\dif \rho'^2}{\rho'^2}  + \sum_{i=1}^p \frac{1 + \rho'^2 b_i^2}{ \rho'^2}{ \lr{\dif \mlim^2_i + \mlim^2_i \dif \phi_i^2} }
   + \lr{\frac{1}{\lambda} + \frac{\sum_{i=1}^p \beta_i^2}{\rho'^2 \hat \m_{p+1}^2}} \dif\hat\m_{p+1}^2     -\frac{2}{\rho'^2}\frac{  \dif\hat\m_{p+1}}{\hat\m_{p+1}}  \lr{ \sum_{i= 1}^{p}{\mlim_i \dif  \mlim_i}}. \label{eqgdSainfty2}
\end{align}
The limit metric is thus
\begin{equation}
  \tilg' = \tilg'_{dS} + \til\H' \til k' \otimes \til k',\quad\quad \til\H' = \frac{2M' \rho'^{n-2}}{ \Pi' \Xi'}, \quad M'\in \mathbb{R}. 
\end{equation}
In addition,  $\tilg_{dS}'$ must be locally isometric to de Sitter, because the metric $\tilg_{dS}$ is ${C}^2$ in $\param$ (up to and including $\zeta=0$) when written in the primed  coordinates.

\bigskip

We next analyze the asymptotic structure. First, the boundary metric
\begin{align}
 \gamma' &  = { \lambda\hat\m_{p+1}^2 } \dif t'^2    + \sum_{i=1}^p { \lr{\dif \mlim^2_i +  \mlim^2_i \dif \phi_i^2} } 
   +  \lr{ \sum_{i=1}^p  {\mlim_i^2}} \frac{\dif \hat\alpha_{p+1}^2 }{\hat\alpha_{p+1}^2}  -2\frac{\dif\hat\m_{p+1}}{\hat\m_{p+1}}  \lr{ \sum_{i= 1}^{p}{\mlim_i \dif  \mlim_i}} , \label{eqgammaainftyodd}
\end{align}
which is explicitly conformally flat in coordinates 
\begin{equation}
  \til\mlim_i = \frac{\mlim_i}{\hat \alpha_{p+1}},\quad\quad i = 1, \cdots, p 
\end{equation}
because
\begin{align}
 \gamma'   = \hat\alpha_{p+1}^2 \lr{ \lambda \dif t'^2    + \sum_{i=1}^p { \lr{\dif {\til\mlim}_i^2 + {\til\mlim}_i^2 \dif \phi_i^2} } 
   }. \label{eqgammaainftyconflat}
\end{align}
This also determines a flat representative $ \gamma_E := \hat\m_{p+1}^{-2} \gamma'$ with Cartesian coordinates $\{ \tcar := \sqrt{\lambda} t', x_i:=\til\mlim_i \cos \phi_i, y_i := \til\mlim_i \sin \phi_i \}$.

The electric part of the rescaled Weyl tensor $D$ follows from equation  \eqref{eqexpweyl}
\begin{align}
 D_{\Y'} = \rho^{n-2} C_\perp |_\scri =  -\lambda n(n-2) M' \lr{\Yf' \otimes \Yf' - \frac{|\Y'|^2_{\gamma'}}{n} \gamma'},
\end{align}
where
\begin{align}
 \Yf' =   \hat\m_{p+1}^2 \dif t' - \sum\limits_{i=1}^p { b_i \mlim_i^2} \dif \phi_i ~\Longrightarrow~ \Y' =  \frac{1}{\lambda}\partial_{t'} - \sum\limits_{i=1}^p { b_i } \partial_{\phi_i},
\end{align}
which in Cartesian coordinates is simply
\begin{equation}
 \Y' =  \frac{1}{\lambda^{1/2}}\partial_{\tcar} - \sum\limits_{i=1}^p { b_i } (x_i \partial_{y_i} - y_i \partial_{x_i}).
 \end{equation}
 Letting $\{ X^A\}_{A=1}^n := \{ \tcar, \{x_i,y_i\}_{i=1}^p \}$,  the skew-symmetric endomorphism of $\mathbb{M}^{1,n+1}$ associated to $\Y'$ is by \eqref{skwmatrix} 
\begin{align}
 F(\Yv) & = \left(
 \begin{array}{ccc}
  0 & 0 & \frac{\lambda^{-1/2}}{2}  \\ 
  0 & 0 & -\frac{\lambda^{-1/2}}{2}  \\
   \frac{\lambda^{-1/2}}{2} & \frac{\lambda^{-1/2}}{2} & 0 
 \end{array}
 \right)
  \bigoplus_{i=1}^{p} 
  \left(
 \begin{array}{cc}
  0 & -b_i \\ b_i & 0
 \end{array}
 \right) \label{eqFainftyodd}
\end{align}
 referred to an orthogonal unit basis $\{e_\alpha\}_{\alpha = 0}^{n+1}$ with $e_0$ timelike. The direct sum \eqref{eqFainftyodd} is adapted to the decomposition 
   \begin{equation}
    \mathbb{M}^{1,n+1} = \mink{1,2} \bigoplus_{i=1}^{p} \Pi_i
   \end{equation}
where $\mink{1,2} = \spn{e_0,e_1,e_2}$ and $\Pi_i = \spn{e_{2i+1},e_{2(i+1)}}$ are $F$-invariant subspaces. For analogous reasons than in the $n$ even case, $\ker F(\Y')$ is degenerate. The polynomial $\mathcal{Q}_{F^2}$ is 
\begin{equation}
 \mathcal{Q}_{F^2} =  x \prod_{i = 1}^p (x-b_i^2),
\end{equation}
and by Proposition \ref{defgammamu}, the parameters determining the conformal class of $\Y'$ are
\begin{equation}
 \{\sigma = 0 ; \mu_1^2 = b_1^2, \cdots, \mu_p^2 = b_p^2 \}.
\end{equation}
Hence, this set of conformal classes covers every point in every region $\mathcal{R}^{(n,m)}_0$.

 \newpage

 \appendix

 \section{Local conformal transformations and local conformal flatness}\label{seclocconflat}

 The conformal diffeomorphisms of a manifold $(\Sigma,\gamma)$ need not to be globally defined. This is well-known in the case of $\mathbb{E}^n$ (e.g. \cite{IntroCFTschBook}), where for every M\"obius transformation two points must be removed. This raises a difficulty for establishing conformal equivalences of global objects, such as global vector fields, because if $\phi$ is only defined in an open neighbourhood $\phi: \neigh \rightarrow \Sigma$, any conformal relation between vector fields must be restricted to $\neigh$ and $\phi(\neigh)$. 
 In the particular case of locally conformally flat manifolds we can use the conformal sphere as a reference to make these relations global.

 Following \cite{zhu94}, we define:
\begin{definition}\label{defconflat}
  A Riemannian $n$-manifold $(\Sigma,\gamma)$ is {\bf locally conformally flat} if there exists an open cover $\{ \neighp_\ia \}$ of $\Sigma$ and a collection of conformal maps $\{ \pc_\ia \}$ from $\neighp_\ia$ to the $n$-sphere, $\pc_ \ia:  \neighp_\ia \rightarrow \mathbb{S}^n$. The set of pairs $\{ \neighp_\ia , \pc_\ia \}$ is called a {\bf conformal cover}. A conformal cover  $\{ \neighp_\ia , \pc_\ia \}$ is said to be {\bf maximal} if every possible conformal map
   $\pc_b : \neighp_b \rightarrow \mathbb{S}^n$ of a domain $\neighp_b \subset \Sigma$ is contained in $\{ \neighp_\ia , \pc_\ia \}$.
\end{definition}

Observe that a maximal conformal cover $\{\neighp_a,\pc_a\}$ of $(\Sigma,\gamma)$ can always be constructed as the union of every conformal cover. It is also clear that the maximal cover is unique. Ee next prove that the maximal conformal cover provides a cover of the sphere:

\begin{lemma}\label{lemmaxcover}
Given the maximal conformal cover $\{\neighp_a,\pc_a\}$ of a locally conformally flat manifold $(\Sigma,\gamma)$, the images $\{\mathcal{W}_a := \neighp_a(\pc_a)\}$  are a cover of $\mathbb{S}^n$.
\end{lemma} 
\begin{proof}
  The group of diffeomorphisms $\conf(\mathbb{S}^n)$ acts transitively on the sphere (note that it contains $SO(n)$). As a consequence,
  given any  $(\neighp_b,\pc_b ) \in\{\neighp_a,\pc_a \}$
  the set of all neighourboods $(\psi \circ \pc_b)(\neighp_b)$ generated with every $\psi \in \conf(\mathbb{S}^n)$ 
  covers $\mathbb{S}^n$. Now, since every $\pc_b' := \psi \circ \pc_b$ is a conformal map from $\neighp_b$ to $\mathbb{S}^n$, it must be contained in the maximal cover and the lemma follows.
\end{proof}

From now on, we shall assume that every locally conformally flat manifold $(\Sigma,\gamma)$ is endowed with its maximal conformal cover. Next, we define the local conformal transformations of $(\Sigma,\gamma)$ as follows
\begin{definition}\label{defconfloc}
  A map $\phi: \neigh \rightarrow \Sigma$, where $\neigh \subset \Sigma$
  is an open set,  is called a {\bf local diffeomorphism of $\Sigma$} if $\phi$ is a diffeomorphism of $\neigh$
 onto its image.
 The set $\bm{\confloc(\Sigma,\gamma)}$ is the set of local diffeomorphisms such that $\phi^\star(\restr{\gamma}{\phi(\neigh)}) = \restr{\omega^2 \gamma}{\neigh}$, for a positive smooth function on $\neigh$.  The set $\bm{\conf(\Sigma,\gamma)}$ is the set of global diffeomorphisms such that $\phi^\star(\restr{\gamma}{\phi(\Sigma)}) = \restr{\omega^2 \gamma}{\Sigma}$, for a positive smooth function on $\Sigma$
\end{definition}

\begin{remark}\label{remarkconflocS}
    In the following discussion, global extendability of the conformal transformations and CKVFs of the $n$-sphere will be key. This property is true for every conformal transformation and CKVF  of $\mathbb{S}^n$ and dimension $n>2$ \cite{blairconformal}. For $n=2$, $\mathbb{S}^2$ admits non-global conformal transformations, as an indirect consequence of its complex structure. In a locally conformally flat $2$-manifold $(\Sigma,\gamma)$, the non-global conformal transformations of $\mathbb{S}^2$ as well as the global conformal transformations $\conf(\mathbb{S}^2)$, induce transformations of $\confloc(\Sigma,\gamma)$ which are not a priori distinguishable. Nevertheless, this is a unique feature of $\mathbb{S}^2$ \cite{borelserre53} (see also \cite{blairconformal}), so to avoid this difficulty, we shall restrict ourselves to the $n > 2$ case here. 
 \end{remark}

\begin{figure}
  \begin{center}
   \psfrag{Cb}{$\chi_b$}
    \psfrag{Cbm}{$\chi_c^{-1}$}
     \psfrag{F}{$\phi$}
      \psfrag{S}{$\psi$}
    \includegraphics[scale=1]{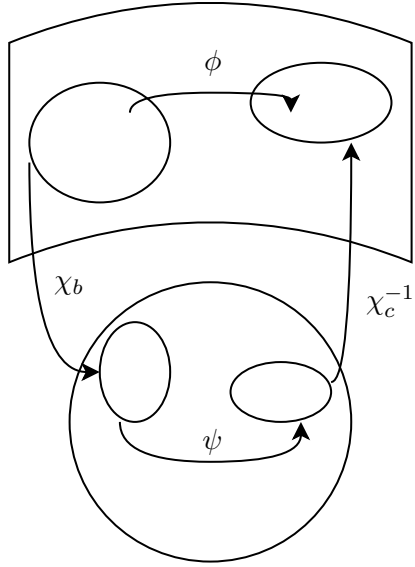}
    \caption{{\small \it Relation between elements $\phi \in \confloc(\Sigma,\gamma)$ and $\psi \in \conf(\mathbb{S}^n)$.}}
        \label{figconfloc}
  \end{center} 
  \end{figure}

Let $(\Sigma,\gamma)$ be a locally conformally flat manifold. We want to
establish a relationship between $\confloc(\Sigma,\gamma)$ and $\conf(\mathbb{S}^n)$.  We start by showing that 
to each transformation $\psi \in \conf(\mathbb{S}^n)$ one can associate maps $\phi \in \confloc(\Sigma,\gamma)$. Choose a conformal map $\pc_b: \neighp_b \rightarrow \mathbb{S}^n$. As a consequence of Lemma \ref{lemmaxcover} and restricting $\neighp_b$ if necessary, 
the image $\psf(\pc_b(\neighp_b))$ lies in the image of some map $\pc_c$ in the maximal cover. Then $\phi := \pc_c^{-1} \circ \psf \circ \pc_b$ is clearly an element of $\confloc(\Sigma,\gamma)$ (Figure \ref{figconfloc}). One can construct as many elements of $\confloc(\Sigma,\gamma)$ as conformal maps $\pc_c$ exist in the maximal cover satisfying the required condition. Also, observe that the transitivity property of $\conf(\mathbb{S}^n)$ induces a transitivity property  in $\confloc(\Sigma,\gamma)$ in the sense that the map $\phi$ can always be constructed so that $\phi(p) = q$ for any two given points $p,q \in \Sigma$. Indeed, such  $\phi$ can be constructed  from any $\psf  \in \conf(\mathbb{S}^n)$ satsifying $\psf(\pc_b(p)) = \pc_c(q)$.

Conversely, to each $\phi \in \confloc(\Sigma,\gamma)$ defined in a neighbourhood $\neigh \subset \Sigma$, one can locally associate a map $\psf$. Let $(\neighp_b,\pc_b)$ and $(\neighp_c,\pc_c)$ belong to the maximal conformal cover $\{\neighp_a,\pc_a \}$ of $\Sigma$ and satisfy that the intersections
$\neigh \cap \neighp_b$
and $\phi(\neigh) \cap \neighp_c$ are non-empty. The map $\psf := \pc_c \circ \phi \circ \pc_b^{-1}$ is well-defined 
on $\chi_b (\neigh \cap \neighp_b) \subset \mathbb{S}^n$ and it is obviously a conformal map. It is a fundamental property of the conformal group of the sphere \cite{IntroCFTschBook}, that there always exists a unique element
$\psf \in \conf(\mathbb{S}^n)$ extending the previous map to the whole sphere.
As before, the assingment of a given element $\phi \in \confloc(\Sigma,\gamma)$ to an element of $\conf(\mathbb{S}^n)$ is highly non-unique. Thus,  there is no one-to-one correspondence between
$\confloc(\Sigma,\gamma)$ and $\conf(\mathbb{S}^n)$. However, as we show next
this correspondence provides a useful notion of conformal class for (local)
conformal vector fields in $(\Sigma,\gamma)$.

Before doing this, let us discuss the case of  $\confloc(\mathbb{E}^n)$. Recall that a map $\phi \in \confloc(\mathbb{E}^n)$, constructed from a $\psi \in \conf(\mathbb{S}^n)$, defines a diffeomorphism in $\mathbb{E}^n$ minus two points (e.g. \cite{IntroCFTschBook}). This follows by relating $\mathbb{E}^n$ and $\mathbb{S}^n$ via the stereographic projection w.r.t. to a pole $N$ at a distance $d$, $St_N : \mathbb{S}^n \backslash \{ N \} \rightarrow \mathbb{E}^n$. 
 Then we define $\phi := St_N \circ \psi \circ St_N^{-1}$ for every $\psi \in \conf(\mathbb{S}^n)$ is a transformation of $\confloc(\mathbb{E}^n)$. When $\psi(N) \neq N$, the map $\phi$ is
a so-called M\"obius transformation \cite{blairconformal} and takes the explicit form
\begin{equation}\label{eqmob} 
 \phi(y) = K   \frac{R(y-p_1)}{|y-p_1|^2} + p_2,
\end{equation}
where $K \in \mathbb{R^+}$, $R$ is a rotation and $p_1,p_2$ are the points in $\mathbb{E}^n$ satisfying $N =  (\psi \circ St_N^{-1})(p_1)$ and $N =  (\psi^{-1} \circ St_N^{-1})(p_2)$. 
This defines a map $\mathbb{E}^n\backslash\{p_1\} \rightarrow \mathbb{E}^n\backslash\{p_2\}$. When $\psi(N) = N$, $\phi$ is an affine transformation of $\mathbb{E}^n$, hence a global diffeomorphism. Given an open set $\neigh \subset \mathbb{E}^n$ the elements of $\confloc(\Sigma,\gamma_{\mathbb{E}^n})$ whose domain
is  $\neigh$ are precisely  the collection of M\"obius transformations \eqref{eqmob} satisfying $p_1 , p_2 \in \mathbb{E}^{n} \setminus \neigh$, together with the set of all affine transformations.

We have now the necessary tools to define the notion of conformal class of CKVFs. 
We define:

\begin{definition}\label{defconfclass}
  Let $\Y$ be a CKVF of a Riemannian manifold $(\Sigma,\gamma)$. The {\bf conformal class} of $\Y$ is the set of all CKVFs $\Y'$ defined in some non-empty
  open neighbourhood $\neigh$ and generated by an element
  $\phi \in \confloc(\Sigma,\gamma)$ whose domain is $\neigh$. Specifically, it consists of all fields $\phi_\star(\restr{\Y}{\neigh}) = \restr{\Y'}{\phi(\neigh)}$.
A conformal class is said to be {\bf global} if $\neigh = \phi(\neigh) = \Sigma$.
\end{definition}

This definition is local, and nothing guarantees that $\xi'$ can be extended to
a global CKVF in $\Sigma$. However, when $\Sigma$ is locally conformally flat, we can show that there is a precise sense in which this local conformal class can be put in a one-to-one correspondance with a global conformal class in the sphere. We do this next.

Let $(\Sigma,\gamma)$ be a locally conformally flat manifold and a global CKVF $\Y$. Let $\Y'$ be an element of the conformal class of $\Y$ and let $\phi\in \confloc(\Sigma,\gamma)$ be the map relating them, defined in a neighbourhood $\neigh \subset  \Sigma$. Let also $(\neighp_b, \pc_b)$ and $(\neighp_c,\pc_c)$ be pairs in the maximal conformal cover of $(\Sigma,\gamma)$ with non-empty intersections $\neigh \cap \neighp_b$ and $\phi(\neigh) \cap \neighp_c$. Denote their images as $\mathcal{W}_b = \pc_b(\neigh \cap \neighp_b)$ and $\mathcal{W}_c = \pc_c(\phi(\neigh) \cap \neighp_c)$. One can locally assign CKVFs of $\mathbb{S}^n$ in $\mathcal{W}_b$ and $\mathcal{W}_c$, through the maps  $\pc_b$ and $\pc_c$, i.e. $\zeta := \pc_{b \star}(\Y)$ and $\zeta':=\pc_{c \star}(\Y')$.
The sphere being simply connected, it follows easily that 
$\zeta, \zeta'$ extend uniquely to global CKVFs in the sphere (as each one of them is the generator of a unique $\psi\in \conf(\mathbb{S}^n)$ \cite{IntroCFTschBook}).
The vector fields $\zeta, \zeta'$ are locally related by the map $\psf := \pc_c \circ \phi \circ \pc_b^{-1}$, which  obviously satisfies $\psf \in \conf(\mathbb{S}^n)$ and we have already mentioned that $\psf$ extends to an element in
$\conf(\mathbb{S}^n)$. 
The relation $\psf_\star(\zeta) = \zeta'$ is global because  $\psf_\star(\zeta)$
is a CKVF that equals $\zeta'$ in $\mathcal{W}_c$, so it must equal $\zeta'$ everywhere, by the uniqueness of extensions of CKVFs on the sphere.

The vector field $\zeta$, associated to a given global CKVF $\Y$ of $(\Sigma,\gamma)$, depends on the element $(\neighp_b,\pc_b)$ of the maximal cover used to define it. However, let $(\neighp_b, \pc_b)$ and $(\neighp_c,\pc_c)$ in the maximal cover have domains with non-empty intersection, i.e. $\neighp_b\cap \neighp_c \neq \emptyset$. In $\mathbb{S}^n$, define the CKVFs of $\zeta_b := \pc_{b\star}(\Y)$ and $\zeta_c := \pc_{c\star}(\Y)$. Then, the map $\psf :=  \pc_{c} \circ  \pc_{b}^{-1}$ restricted to $\pc_b(\neighp_b\cap \neighp_c)$ satisfies 
$\psf_\star(\zeta_b) = \zeta_c$. But $\psi$ extends to a global map in $\conf (\mathbb{S}^n)$ and, by the argument above, this relation also extends globally to $\mathbb{S}^n$. Therefore, the vector fields  $\zeta_b, \zeta_c$ associated to $\Y$ are in the same global conformal class of $\mathbb{S}^n$ if $\neighp_b$ and $\neighp_c$ intersect. Moreover, if $\Sigma$ is connected, this is true even if $\neighp_b\cap \neighp_c = \emptyset$, because $\neighp_b$ and $\neighp_c$ can be joined through a finite sequence of neighbourhoods\footnote{Connected manifolds are path connected so there exists  a continuous
  curve $\alpha: [0,1] \rightarrow \Sigma$ joining a point $p \in \neighp_b$
  with a point $q \in \neighp_c$. The set of points $\alpha([0,1])$ is compact, so from any cover one can extract a finite subcover.  It suffices to start with the full cover $\{ \neighp_a \}$ associated to the maximal conformal cover, extract a finite subcover and, in necessary, supplement with $\neighp_b$, $\neighp_c$  to fulfill all the properties that we require.}
$\{\neighp_k\}_{i=1}^K$ in the maximal cover  $\{\neighp_k,\pc_k \}_{k=1}^K \subset \{ \neighp_a,\pc_a \}$ such that $\neighp_k \cap \neighp_{k+1} \neq \emptyset$ and $\neighp_1= \neighp_b$ and $\neighp_K = \neighp_c$.  In $\neighp_k \cap \neighp_{k+1}$,  the map $\psf_k = \pc_{k+1}^{-1} \circ \pc_k$ establishes a conformal map. All such maps, extended globally in $\mathbb{S}^n$ and combined $\psf := \psf_1 \circ \cdots \circ \psf_{K-1}$, determine a conformal relation  $\zeta_b = \psf_\star(\zeta_c)$.  
 
Thus, the above discussion shows that all CKVFs in the conformal class of $\Y$ and $\Y'$ of a connected, locally conformally flat manifold $\Sigma$, determine a unique global conformal class of  CKVF in $\mathbb{S}^n$.
The converse is also true because of the following argument. Let $(\neighp_b,\pc_b)$ belong to the maximal conformal cover and consider $\zeta = \pc_{b\star}(\Y)$ and $\zeta' = \psf_\star(\zeta)$ for any $\psf \in \conf(\mathbb{S}^n)$. Then, as a consequence of Lemma \ref{lemmaxcover}, there exists a pair $(\neighp_c,\pc_c)$ in the maximal conformal cover such that $\pc_c(\neighp_c) \cap \psf(\pc_b(\neighp_b)) \neq \emptyset $. Hence, in $\pc_c(\neighp_c) \cap \psf(\pc_b(\neighp_b)) $ the vector field $\zeta'$ induces, via $\pc_c^{-1}$, a CKVF $\Y'$ of
$\gamma$. By construction, the map $\phi :=\pc_{c}^{-1} \circ \psf \circ \pc_b$ belongs to $\confloc(\Sigma,\gamma)$ and
satisfies $\phi_\star(\Y) = \Y'$ on a non-empty domain. Thus, $\Y'$ is in the conformal class of $\Y$. Summarizing

\begin{proposition}\label{propconfclass}
  Let $(\Sigma,\gamma)$ be a Riemannian, connected and locally conformally flat $n$-manifold  with $n >2$. Then, the conformal classes of CKVF in $(\Sigma,\gamma)$ as given in
  Definition \ref{defconfclass} are in one-to-one correspondence with
  global conformal classes of CKVFs of $\mathbb{S}^n$.
\end{proposition}

% 
% \newpage
% 
%  
% 
% \bibliography{bibliografia}{}
% \bibliographystyle{plain2}

\section*{Acknowledgements}

The authors acknowledge financial support under the projects
PGC2018-096038-B-I00
(Spanish Ministerio de Ciencia, Innovaci\'on y Universidades and FEDER)
and SA096P20 (JCyL). C. Pe\'on-Nieto also acknowledges the Ph.D. grant BES-2016-078094 (Spanish Ministerio de Ciencia, Innovaci\'on y Universidades).

\end{document}